\documentclass[]{siamart220329}

\usepackage{amsmath,amssymb,graphicx, tikz}
 
 \newsiamremark{defi}{Definition}
\newsiamremark{example}{Example}
\newsiamremark{rmk}{Remark}

\headers{Identifiability of Networks}{E.S.~Allman, H. Ba\~nos, M.~Garrote-Lopez, and J.A.~Rhodes}

 \title{Identifiability of Level-1 Species Networks from Gene Tree Quartets
 \thanks{\today \funding{ESA and JAR were partially supported by NSF grant DMS-2051760,  HB by NSF grant DMS-2331660 and MG-L by NIH grant P20GM103395,
 		an NIGMS Institutional Development Award (IDeA) and by the Borsa FSB 2020 from Fundació Ferran Sunyer i Balaguer.
 		Part of this research was performed while MG-L and JAR visited the Institute for Mathematical and Statistical Innovation (IMSI), 
 		which is supported by the National Science Foundation (Grant DMS-1929348)} }
} 

 \author{Elizabeth S. Allman\thanks{University of Alaska Fairbanks (\email{e.allman@alaska.edu}, \email{j.rhodes@alaska.edu} ).}
 \and Hector Ba\~nos\thanks{California State University San Bernadino (\email{hector.banos@csusb.edu}).}
 \and Marina Garrote-Lopez\thanks{Max Planck Institute for Mathematics in the Sciences (\email{marina.garrote@mis.mpg.de}).}
\and John A. Rhodes\footnotemark[2]}

\usepackage{amsopn}

\begin{document}

\maketitle
 	
\begin{abstract}
When hybridization or other forms of lateral gene transfer have occurred, evolutionary relationships of 
species are better represented by phylogenetic networks than by trees. 
While inference of such networks remains challenging, several recently proposed methods are 
based on quartet concordance factors --- the probabilities that a tree relating a 
gene sampled from the species displays the possible 4-taxon relationships. 
Building on earlier results, we investigate what level-1 network features are identifiable from 
concordance factors under the network multispecies coalescent model. We obtain results on both topological 
features of the network, and  numerical parameters, 
uncovering a number of failures of identifiability related to 3-cycles in the network.
\end{abstract}


 	 	
 \section{Introduction}
 	
Statistical inference of phylogenetic networks, showing evolutionary relationships between species when hybridization or 
other horizontal gene transfer has occurred,
 poses substantial theoretical and practical problems. With data in the form of many sequenced and aligned genes, 
 standard phylogenetic 
 methods can be used to infer gene trees. However, due to both horizontal inheritance and the population genetic 
 effect of incomplete lineage sorting, 
 these gene trees reflect the species network topology only indirectly. Extracting the network signal with an 
 acceptable computational time, and even 
 determining what aspects of the network can be inferred under the Network Multispecies Coalescent  (NMSC) 
 model, is challenging.

Several recently-developed network inference methods  utilize summaries of (inferred) gene trees 
through counts of their displayed quartet trees, that is, empirical quartet \emph{concordance factors
($CF$s)}.
SNaQ \cite{Solis-Lemus2016} using pseudolikelihood on these $CF$s to pick an optimal network among those of level 1.
NANUQ \cite{Allman2019}  also uses quartet counts in the level-1 setting, but avoids even pseudolikelihood computations, 
by conducting hypothesis tests for each quartet, followed by a distance-based approach to avoid searching over networks. 
(PhyloNet \cite{Nakhleh2015} similarly  uses pseudolikelihood, though with rooted triple counts and without the level-1 restriction.) 

While these methods strike a balance between thorough statistical analysis and computational effort, a 
complete exploration of 
what level-1 network features are 
identifiable from CFs under the NMSC has yet to be undertaken. First results in this direction  \cite{Solis-Lemus2016} 
showed certain semidirected level-1 network topologies 
were distinguishable from those obtained by dropping a hybrid edge, and that in some cases  numerical 
parameters were  identifiable 
up to a finite number of possibilities, 
i.e., were locally identifiable. Topological identifiability was later investigated \cite{Banos2019}, establishing 
that semidirected  level-1 network topologies are identifiable up to contraction 
of 2- and 3-cycles and 
directions of hybrid edges in 4-cycles, for generic parameters. While these works provide our starting point, we seek to fill in 
unaddressed gaps. (\Cref{app:appeasement} of the supplementary materials gives more detail on 
how this work complements its predecessors, and discusses the claims and arguments 
in \cite{Solis-Lemus2020} for work described in \cite{Solis-Lemus2016}.) 

We rigorously establish what can be identified, and what cannot, from quartet $CF$s under the NMSC.
Our concern here is only the theoretical question of  identifiability. We thus delineate what \emph{might} be consistently inferred by a method 
using quartet counts, 
although particular methods may not be able to do so. 
Our results also imply parameter identifiability results for  
 data types from which quartet counts can be obtained 
(e.g., topological gene trees, or metric gene trees), although for such data it is possible 
that stronger identifiability claims could be established.

\smallskip 

Our main results address identifiability of the full semidirected topology of a binary network, including hybrid edge directions  
(\cref{thm:mainTop}), and the numerical parameters of edge lengths and hybridization (or inheritance) probabilities (\cref{thm:mainNum}). 
One interesting aspect is that the presence of a 3-cycle can generally be detected, but whether the hybrid node of that cycle can be 
identified or not depends on the numerical parameters. The subsets of parameters on which this question has a positive or 
negative answer both have positive measure, and thus neither set can be dismissed as non-generic.

The precise statements of these theorems have exceptions for cycles adjacent to pendant edges. However, these simplify if one has multiple samples 
per taxon. Then the semidirected network topology is generically identifiable except for the presence of 2-cycles and (sometimes) hybrid nodes 
in 3-cycles. If the semidirected network topology with 2-cycles removed is known, then all numerical parameters except those relating to 3-cycles 
and their adjacent edges are generically identifiable. 

\smallskip

Underlying our results are  analyses of algebraic varieties associated with certain small networks using computational algebra software. These lead to algebraic (polynomial equality) tests of quartet $CF$s for different network substructures. However, 3-cycle identifiability results depend on semialgebraic tests (polynomial inequalities). These tests were motivated by
equalities found for related networks, but their construction is not purely computational.

Our identifiability results for numerical parameters are based on explicit rational formulas for parameters in terms of $CF$s,  so if a topological network is known or proposed, one could in principle estimate numerical parameters with them. But while some of these formulas are simple, others are quite complicated,  and should not be expected to provide good estimates from data. These formulas may, however, give useful initial estimates of parameters that could then be refined through optimization, such as with likelihood methods.

\smallskip

\Cref{sec:defs,sec:NMSC} give definitions and earlier results that we use as our starting point. In \Cref{sec:topID} we study topological identifiability of level-1 binary networks from quartet $CF$s, and in \Cref{sec:numID} the identifiability of numerical network parameters from the same information. \Cref{sec:imp} discusses implications for data analysis.

In  the Supplementary Materials, \Cref{app:appeasement} explains how our results complement earlier work, and  \Cref{app:props}  catalogs the computational results for specific networks that underly our arguments.

 	
 \section{Definitions}\label{sec:defs}
 	\subsection{Rooted and unrooted phylogenetic networks}\label{sec::def}
	
A \emph{topological binary rooted phylogenetic network}  $N^ + $  is a finite rooted graph, with all edges directed away from the root, 
whose non-root internal nodes form two classes: \emph{Tree nodes} have indegree 1 and outdegree 2, while 
\emph{hybrid nodes} have indegree 2 and outdegree 1. \emph{Hybrid edges} and \emph{tree edges} are classified 
according to their child nodes. Leaves of the network are bijectively labelled by \emph{taxa} in a set $X$. 
A network is 
\emph{metric} if in addition each tree edge is assigned a positive length, each hybrid edge a non-negative length and a 
positive probability $\gamma$, such that for every pair of hybrid edges $e, e'$ with a common child $\gamma+\gamma'=1$. 
More formal definitions of  phylogenetic networks appear in \cite{Steel2016,Solis-Lemus2016,Banos2019}.

We often depict these networks with their root at the top, referring to edges and nodes as \emph{above} or \emph{below} one another in the natural way.
	
As explained more fully in \cite{Allman2019,Banos2019}, for gene quartet-based methods of  inference a useful form of an unrooted 
network is more subtle than that for a tree, since substructures above the \emph{least stable ancestor} (LSA) of the taxa \cite{Steel2016} 
are undetectable by them. The \emph{topological unrooted phylogenetic network induced from $N^ + $}, is the semidirected network $N=N^-$,  
obtained by  $N^ + $ by deleting all vertices and nodes above the 
LSA, undirecting tree edges, and suppressing the LSA.
Since our concern in this work is the identifiability of unrooted phylogenetic networks, we will often use $N$ rather than the 
more cumbersome $N^-$ to denote them. We refer to $N$ as unrooted or semidirected interchangeably.  Note that $N$ naturally inherits a 
metric structure if $N^+$ has one. 

\Cref{fig::net} shows an example of a network $N^+$ and its semidirected network $N^-$. While in that example all leaves are equidistant from the root of $N^+$, we do not assume that generally.

Tree edges can be further partitioned into \emph{cut} and \emph{noncut} edges, according to whether their deletion results in a graph with 
2 connected components or not. Note that hybrid edges are never cut edges.

 \begin{figure}
	\include{Figures/bigNetwork}
	\caption{(L) A rooted network $N^ + $ on $X$ with root $r = \text{LSA}(X)$, and (R) The unrooted network $N^-$ obtained from $N^ + $.}\label{fig::net}
\end{figure}
	
\smallskip
 
Of particular interest are unrooted networks on four taxa obtained from a larger network by restricting the taxon set. 
Recall that a binary unrooted topological tree on four taxa $a,b,c,d$ is called a \emph{quartet} $ab|cd$ if deletion 
 	of its sole internal edge gives connected components with taxa $\{a, b\}$ and $\{c,d\}$.  When $n \ge 4$, 
 	an $n$-taxon tree \emph{displays} a quartet $ab|cd$ if the induced unrooted tree on the four taxa is $ab|cd$.
 A formal extension of this concept to quartet networks follows.

\begin{defi}	
Let $N^ + $ be a rooted network on $X$, and let $a,b,c,d\in X$. The \emph{induced quartet network} on $a,b,c,d$ is the unrooted network obtained by 	
\begin{enumerate}
			\item retaining only the nodes and edges of $N^ + $ ancestral to at least one of $a,b,c,d$,
			\item suppressing nodes of degree 2, and
			\item unrooting the resulting network.
\end{enumerate}
If $N^ + $ is a metric network, its quartet networks naturally are as well.
\end{defi}
	
An analogous definition for induced quartet networks of $N$ is given in \cite{Banos2019}, which also shows the quartet networks induced from $N^ + $ and $N$ are isomorphic.   \Cref{fig::quartet} shows  some metric quartet networks induced from the networks of \cref{fig::net}.  
\begin{figure}
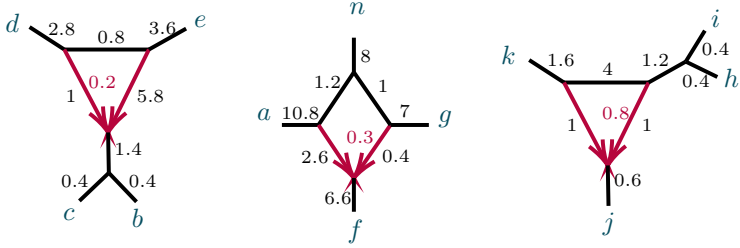

  \begin{center}
  \include{Figures/subNetworks_bigNetworks_v2}
 \end{center}
\caption{Several semidirected quartet  networks induced from the network in \cref{fig::net}.}\label{fig::quartet}
\end{figure}

\subsection{Level-1 networks}
 We restrict our study to the  family of level-1 phylogenetic networks. These have been the focus of many works \cite{Huber2015,HusonRuppScorn,Gusfield2007,Rossello2009,Solis-Lemus2016,Banos2019, Allman2019}, though only a few of these
incorporate the coalescent model that is central here. 

By a \emph{cycle} in either a rooted or unrooted phylogenetic network we mean a set of edges and nodes that form a cycle when all edges are treated as undirected.
 		
\begin{defi}\label{level1}Let $N$ be a (rooted or unrooted) binary topological  network. If no two 
cycles in the undirected graph of $N$ share a node, then $N$ is \emph{level-1}. \
 \end{defi}

In some works  level-1 networks are defined as those in which no cycles share an edge; i.e., cycles are edge-disjoint rather than the stricter vertex-disjoint condition we adopt. However, in our context of binary networks they are equivalent  \cite{Rossello2009}.

 In a level-1 network  a cycle that is composed of $m$ edges, (2 hybrid edges and $m-2$ tree edges) is said to be an $m$-cycle.
More specifically, it is an \emph{$m_k$-cycle} if there are exactly $k$ taxa descended from its unique hybrid node  \cite{Banos2019}.
This terminology can be used for semidirected networks, since `descended from a hybrid node' is unambiguous, regardless of where the network is rooted.

 Let $N$ be an unrooted level-1 network on $X$ with an $m$-cycle $C$. Then $C$ induces a partition of $X$ into $m$ subsets according to the connected components obtained by deleting all edges in the cycle. Elements of this partition are the \emph{blocks} of $C$.  The  \emph{hybrid block of $C$} is the block of taxa descended from the hybrid node in $C$.  If the blocks of $C$ have $n_1,n_2,\dots,n_m$ taxa, then we say $C$ induces a $(n_1,n_2,\dots,n_m)$ partition.
 

\section{The Network Multispecies Coalescent Model and quartet concordance factors}\label{sec:NMSC}
	
 The Network Multispecies Coalescent (NMSC) model \cite{Meng2009} mechanistically describes the formation of gene trees within a species network, as gene lineages are traced backward in time to common ancestors in the edge populations of the network. Under it, gene trees may differ in topology from any displayed trees on the species network.  Given a metric rooted phylogenetic network, the NMSC assigns positive probabilities to all resolved metric gene trees, and, through marginalization, to topological gene trees and induced gene quartet topologies.

  \begin{defi}
 		Let  $N^ + $ be a metric rooted network on a taxon set $X$, and $A$, $B$, $C$, $D$ a  gene sampled from 
 		individuals in species $a,b,c,d\in X$ respectively.  The \emph{(scalar) quartet concordance factor} 
 		$CF_{ab|cd} = CF_{ab|cd} (N^ + )$ is the probability under the NMSC on $N^ + $ that a gene tree displays the quartet $AB|CD$. The 
 		\emph{(vector) quartet concordance factor} $\overline{CF}_{abcd} = \overline{CF}_{abcd}(N^ + )$ is  the triple  
 		$$\overline{CF}_{abcd}=(CF_{ab|cd},CF_{ac|bd},CF_{ad|bc})$$
 		of concordance factors of each possible quartet on the taxa $a,b,c,d$.
\end{defi}

That $CF$s for quartet networks depend only on the semidirected quartet network, was proved in \cite{Banos2019}.
That result implies the following.

\begin{lemma}\label{lem:root}
Under the NMSC on a level-1 network $N^+$ the values of the quartet $CF$s  depend only on the induced semidirected network $N$. 
\end{lemma}

	Following on the first steps investigating level-1 network  identifiability from quartet $CF$s  taken in \cite{Solis-Lemus2016},  the next result, that most topological features of a level-1 species network are identifiable from quartet $CF$s, appeared in \cite{Banos2019}.
		
\begin{theorem}\label{thm:topExcept4}\cite{Banos2019} 
Let $N$ be a binary semidirected metric  level-1 species network.  Let
$ N'$ be the semidirected topological network obtained from $N$ by contracting all 2- and 3-cycles, and undirecting hybrid edges in 4-cycles.
Under the NMSC model with generic numerical parameters, the network $ N'$ is identifiable from quartet $CF$s for $N$.
\end{theorem}

We take this theorem as our starting point, and in \Cref{sec:topID} focus on the remaining questions of topological identifiability: From quartet $CF$s can any aspects of 2-cycles or 3-cycles can be identified, and  for 4-cycles can the hybrid node be identified?  In \Cref{sec:numID} we turn to identifiability of the numerical parameters of edge lengths and hybridization probabilities. While these were not a focus in \cite{Banos2019}, partial results on local identifiabiity of numerical parameters were given  in \cite{Solis-Lemus2016}.

\smallskip

Unless explicitly stated otherwise, we assume that exactly 1 gene lineage is sampled per taxon. If 2 lineages were sampled for a taxon, say $a$,  `pseudotaxa,' $a_1$ and $a_2$ 
can be introduced by attaching a cherry leading to these at the leaf $a$ of the network. Under the NMSC, $CF$s for the modified network with 1 sample from each $a_i$ are 
identical to those for the original network with 2 samples from $a$. Sampling more than 2 lineages per taxon only introduces new $CF$s in which 3 or 4 pseudotaxa  from 
the same taxon appear, but due to exchangeability of lineages under the NMSC these $CF$s
are always 1/3. Thus identifiability results for any multiple sampling scheme 
will follow from the single sample case on a modified network. No edge lengths are needed in the pseudotaxa cherries, since no coalescent event may occur on them.
 
 \smallskip

Under the NMSC one can derive formulas for $CF$s for any fixed network in terms
of the numerical parameters. These have the form of polynomials in the hybridization parameters $\gamma$
and the $\exp(-t)$ for all edge lengths $t$.
The expression $\exp(-t)$ has a simple interpretation as the probability that two gene lineages entering an edge of length $t$ coalescent units (tracing time backwards) do not coalesce within that edge.
By reparameterizing using \emph{edge probabilities} $\ell=\exp(-t)\in (0,1]$ rather than lengths $t\in[0,\infty)$, all formulas for $CF$s  are given by polynomial 
formulas in the $\ell$s and $\gamma$s. 

The $3{ n \choose 4}$ scalar quartet $CF$s  for a fixed topological  network $N$ on $n$ taxa then define a polynomial map from the numerical parameter space 
into $\mathbb R^{3{n \choose 4}}$. Extending the map to allow complex $\ell,\gamma$, gives a parameterized algebraic variety. The set of multivariate polynomials 
in the $CF$s that vanish on the parameterization's image is an ideal, denoted $\mathcal I(N)=\mathcal I(N^+)=\mathcal I(N^-)$. The zero set $\mathcal V(N)$ of the 
polynomials in $\mathcal  I(N)$ is the Zariski closure of the parameterized variety. These notions from applied algebraic geometry provide a framework for our work. 
Elements of $\mathcal I(N)$ are called \emph{invariants}, and depend only the network topology, and not its numerical parameters.

Our arguments use symbolic computations with CFs from specific networks, performed and verified by
the software {\tt Singular} \cite{Singular} and {\tt Macaulay2} \cite{M2}. Despite their essential role, for brevity all  computational results are stated in the supplementary materials, \Cref
{app:props}. That section also contains an exposition of certain linear invariants that can be derived without computation, and which help simplify both computations and statements of results.


\section{Identifiability of  semidirected network topologies} \label{sec:topID}

\subsection{2-cycles}\label{ssec:2cycle}

We first show 2-cycles (parallel edges) in level-1 networks are never identifiable. By \emph{replacing} a 2-cycle with parental node $u$ and 
child node $v$ \emph{by an edge}, we mean removing  its two edges, introducing a new directed edge $(u,v)$ with a specified edge probability, 
and suppressing resulting nodes of degree 2.

The content of the following Lemma was essentially given in \cite{Solis-Lemus2016}, and has appeared in other works subsequently, 
most recently \cite{AneEtAl2023}. We restate it here for completeness. 

\begin{lemma}\label{lem:2cyc}
Let $N^ + $ be a level-1 rooted binary metric phylogenetic network, with a 2-cycle composed of hybrid edges with edge probabilities $
h_1,h_2$, and corresponding hybridization parameters $\gamma_1,\ \gamma_2=1-\gamma_1$. Then quartet $CF$s  for $N^ + $ under the NMSC 
are unchanged if the 2-cycle is replaced by an edge with edge probability $\ell \in (0,1)$ determined by the equation
$$1-\ell =\gamma_1^2(1-h_1)  + (1-\gamma_1)^2(1-h_2).$$
\end{lemma}

Since varying the 2-cycle parameters in the above expression causes  
$\ell$ to range over the full interval $(0,1)$, we obtain the following.

\begin{corollary} Using quartet $CF$s,  under the NMSC a topological level-1 phylogenetic network $N$ with a 2-cycle cannot be 
	distinguished from the network $\widetilde N$ obtained by  replacing that two cycle with an edge.\end{corollary}

\subsection{3-cycles}\label{ssec:3top}

Using
\cref{thm:topExcept4}, the question of identifying topological 3-cycles in a network is reduced to distinguishing between the network that theorem identifies, and networks obtained from it by replacing some set of non-cycle tree nodes with 3-cycles. We only consider networks with 5 or more taxa, as the 4-taxon case is fully studied in  \cite{Banos2019}.
 
 \subsubsection {3-cycles near leaves}
 We begin with a non-identifiability result, for certain 3-cycles adjacent to two pendant edges of a network, as shown in \cref{fig:two3cyc11}.
 
 \begin{figure}
	\begin{center}
	\tikzset{every picture/.style={line width=0.75pt}} 

\begin{tikzpicture}[x=0.75pt,y=0.75pt,yscale=-0.9,xscale=0.9]

\draw [line width=1.5]    (43.76,51.98) -- (91.14,51.98) ;
\draw [line width=1.5]    (91.14,51.98) -- (108.57,45.57) ;
\draw [color={rgb, 255:red, 182; green, 5; blue, 60 }  ,draw opacity=1 ][line width=1.5]    (43.76,51.98) -- (66.88,95.82) ;
\draw [shift={(68.28,98.47)}, rotate = 242.19] [color={rgb, 255:red, 182; green, 5; blue, 60 }  ,draw opacity=1 ][line width=1.5]    (14.21,-4.28) .. controls (9.04,-1.82) and (4.3,-0.39) .. (0,0) .. controls (4.3,0.39) and (9.04,1.82) .. (14.21,4.28)   ;
\draw [color={rgb, 255:red, 182; green, 5; blue, 60 }  ,draw opacity=1 ][line width=1.5]    (91.14,51.98) -- (69.61,95.78) ;
\draw [shift={(68.28,98.47)}, rotate = 296.18] [color={rgb, 255:red, 182; green, 5; blue, 60 }  ,draw opacity=1 ][line width=1.5]    (14.21,-4.28) .. controls (9.04,-1.82) and (4.3,-0.39) .. (0,0) .. controls (4.3,0.39) and (9.04,1.82) .. (14.21,4.28)   ;
\draw [line width=1.5]    (67.73,117.29) -- (68.28,98.47) ;
\draw  [fill={rgb, 255:red, 155; green, 155; blue, 155 }  ,fill opacity=0.55 ][line width=1.5]  (108.57,45.57) -- (117.72,10.79) -- (144,53.89) -- cycle ;
\draw [line width=1.5]    (207.77,39.77) -- (227.26,52.37) ;
\draw [color={rgb, 255:red, 182; green, 5; blue, 60 }  ,draw opacity=1 ][line width=1.5]    (227.26,52.37) -- (271.64,52.37) ;
\draw [shift={(274.64,52.37)}, rotate = 180] [color={rgb, 255:red, 182; green, 5; blue, 60 }  ,draw opacity=1 ][line width=1.5]    (14.21,-4.28) .. controls (9.04,-1.82) and (4.3,-0.39) .. (0,0) .. controls (4.3,0.39) and (9.04,1.82) .. (14.21,4.28)   ;
\draw [line width=1.5]    (274.64,52.37) -- (292.07,45.96) ;
\draw [color={rgb, 255:red, 0; green, 0; blue, 0 }  ,draw opacity=1 ][line width=1.5]    (227.26,52.37) -- (251.78,98.86) ;
\draw [color={rgb, 255:red, 182; green, 5; blue, 60 }  ,draw opacity=1 ][line width=1.5]    (273.31,55.06) -- (251.78,98.86) ;
\draw [shift={(274.64,52.37)}, rotate = 116.18] [color={rgb, 255:red, 182; green, 5; blue, 60 }  ,draw opacity=1 ][line width=1.5]    (14.21,-4.28) .. controls (9.04,-1.82) and (4.3,-0.39) .. (0,0) .. controls (4.3,0.39) and (9.04,1.82) .. (14.21,4.28)   ;
\draw [line width=1.5]    (251.23,117.68) -- (251.78,98.86) ;
\draw  [fill={rgb, 255:red, 155; green, 155; blue, 155 }  ,fill opacity=0.55 ][line width=1.5]  (292.07,45.96) -- (301.22,11.18) -- (327.5,54.28) -- cycle ;
\draw  [fill={rgb, 255:red, 0; green, 0; blue, 0 }  ,fill opacity=1 ] (272.87,52.37) .. controls (272.87,51.39) and (273.66,50.6) .. (274.64,50.6) .. controls (275.61,50.6) and (276.4,51.39) .. (276.4,52.37) .. controls (276.4,53.34) and (275.61,54.13) .. (274.64,54.13) .. controls (273.66,54.13) and (272.87,53.34) .. (272.87,52.37) -- cycle ;
\draw [line width=1.5]    (24.27,39.38) -- (43.76,51.98) ;

\draw (11.53,24.93) node [anchor=north west][inner sep=0.75pt]  [font=\normalsize]  {$b$};
\draw (61.94,122.19) node [anchor=north west][inner sep=0.75pt]  [font=\normalsize]  {$a$};
\draw (195.03,24.47) node [anchor=north west][inner sep=0.75pt]  [font=\normalsize]  {$b$};
\draw (245.94,122.19) node [anchor=north west][inner sep=0.75pt]  [font=\normalsize]  {$a$};
\draw (270.04,32.03) node [anchor=north west][inner sep=0.75pt]  [font=\normalsize]  {$v$};

\end{tikzpicture}\vspace*{-1cm}
 	\end{center}
	\caption{Networks with 3-cycles inducing $(1,1,n-2)$ partitions. The shaded triangle represents an arbitrary semidirected subnetwork. (L,R) correspond to cases (1,2) of \cref{prop:3-cyc11}.
}\label{fig:two3cyc11} 
\end{figure}
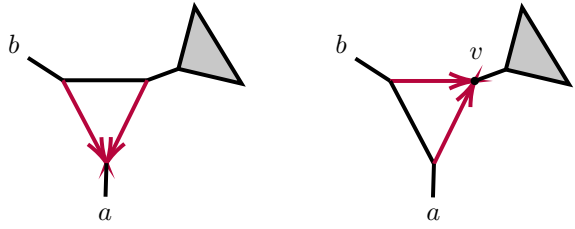
 
\begin{proposition} \label{prop:3-cyc11}
Suppose a binary semidirected network $ N$ on $n\ge 4$ taxa has a 3-cycle $C$ inducing a $(1,1,n-2)$ partition of the taxa. Let $N'$ be the network
obtained by contracting $C$ to a node. Then under the NMSC

\begin{enumerate}
\item \label{case:1} If $C$ is a $3_1$-cycle, so its hybrid node has only 1 descendant taxon, 
the topologies of $N$ and $N'$ cannot be distinguished using quartet  $CF$s.
That is, for any choice of parameters on one of these networks, there exist parameters on the other giving identical $CF$s. 
Moreover, the parameters other than those associated to $C$ and internal edges adjacent to $C$
 may be chosen to be identical on both networks.

\item \label{case:2} 
If $C$ is a $3_k$-cycle with $k=n-2\ge2$, and the parameter spaces are extended to allow all real edge lengths in the 
CF formulas,  then for any choice of extended parameters on $N$
there are extended parameters on  $N'$ giving identical  $CF$s, and \emph{vice versa}. 
Moreover, the parameters other than those associated to $C$  and internal edges 
adjacent to $C$ may be chosen to be identical on
$N$ and $N'$.

\smallskip
Furthermore, for strictly
positive edge lengths on $N$ and $N'$,  
there are two positive-measure subsets of parameters, 
$\Theta_1, \Theta_2$,
for $N$, such that  on 
$\Theta_1$
the topologies of $N$  and $N'$ are not distinguishable using quartet  $CF$s, 
and 
on
$\Theta_2$
are distinguishable.
\end{enumerate}
\end{proposition}

Note that case \eqref{case:1} implies that if $N$ is as shown in \cref{fig:two3cyc11}(L) then $N$ and $N'$ also cannot be distinguished from the 
network obtained from $N$ by interchanging the 
$a$ and $b$ labels.
In case \eqref{case:2}, if parameters are such that $N$ is not distinguishable from $N'$,
then case \eqref{case:1} implies that they are also not
distinguishable from the two networks obtained by redesignating the hybrid node in the 3-cycle to be a singleton.
When
$N$ is distinguishable from $N'$, 
then by case \eqref{case:1} it is distinguishable from those two other networks
as well.

\begin{proof}
Let $a,b$ denote the taxa in the singleton blocks.
For case \eqref{case:1},  we may assume the  network is rooted, with the root outside $C$ and not on the pendant edges leading to 
$a,b$ (\cref{fig:two3cyc11} (L)). Then under the NMSC there is a probability $p\in(0,1)$, 
depending on the numerical parameters of the 3-cycle, that lineages $a$ and $b$ fail to coalesce before leaving the 3-cycle. 
Replacing the 3-cycle
and its adjacent edges by a 3-leaf tree where the edge leading toward the $n-2$
taxa has edge probability $p$ leaves the distribution of topological gene trees, and hence quartet $CF$s, unchanged. Varying 
parameters over the 3-cycle or over the 3-leaf tree allows all probabilities $p\in (0,1)$ to be achieved.
 
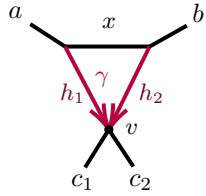
\begin{figure}
	\begin{center}
	\tikzset{every picture/.style={line width=0.75pt}} 

\begin{tikzpicture}[x=0.75pt,y=0.75pt,yscale=-0.9,xscale=0.9]

\draw [line width=1.5]    (50.57,24.38) -- (70.06,36.98) ;
\draw [line width=1.5]    (70.06,36.98) -- (117.44,36.98) ;
\draw [color={rgb, 255:red, 182; green, 5; blue, 60 }  ,draw opacity=1 ][line width=1.5]    (70.06,36.98) -- (93.18,80.82) ;
\draw [shift={(94.58,83.47)}, rotate = 242.19] [color={rgb, 255:red, 182; green, 5; blue, 60 }  ,draw opacity=1 ][line width=1.5]    (14.21,-4.28) .. controls (9.04,-1.82) and (4.3,-0.39) .. (0,0) .. controls (4.3,0.39) and (9.04,1.82) .. (14.21,4.28)   ;
\draw [color={rgb, 255:red, 182; green, 5; blue, 60 }  ,draw opacity=1 ][line width=1.5]    (117.44,36.98) -- (95.91,80.78) ;
\draw [shift={(94.58,83.47)}, rotate = 296.18] [color={rgb, 255:red, 182; green, 5; blue, 60 }  ,draw opacity=1 ][line width=1.5]    (14.21,-4.28) .. controls (9.04,-1.82) and (4.3,-0.39) .. (0,0) .. controls (4.3,0.39) and (9.04,1.82) .. (14.21,4.28)   ;
\draw [line width=1.5]    (81.3,104.75) -- (94.58,83.47) ;
\draw [line width=1.5]    (117.44,36.98) -- (138.3,25.25) ;
\draw [line width=1.5]    (108.3,104.25) -- (94.58,83.47) ;
\draw  [fill={rgb, 255:red, 0; green, 0; blue, 0 }  ,fill opacity=1 ] (92.58,83.47) .. controls (92.58,82.37) and (93.48,81.47) .. (94.58,81.47) .. controls (95.69,81.47) and (96.58,82.37) .. (96.58,83.47) .. controls (96.58,84.58) and (95.69,85.47) .. (94.58,85.47) .. controls (93.48,85.47) and (92.58,84.58) .. (92.58,83.47) -- cycle ;

\draw (139.83,11.43) node [anchor=north west][inner sep=0.75pt]  [font=\normalsize]  {$b$};
\draw (36.74,12.69) node [anchor=north west][inner sep=0.75pt]  [font=\normalsize]  {$a$};
\draw (71.74,107.19) node [anchor=north west][inner sep=0.75pt]  [font=\normalsize]  {$c_{1}$};
\draw (104.24,107.19) node [anchor=north west][inner sep=0.75pt]  [font=\normalsize]  {$c_{2}$};
\draw (85.2,49.2) node [anchor=north west][inner sep=0.75pt]  [font=\small,color={rgb, 255:red, 182; green, 5; blue, 60 }  ,opacity=1 ]  {$\gamma $};
\draw (65.6,56) node [anchor=north west][inner sep=0.75pt]  [font=\small,color={rgb, 255:red, 108; green, 4; blue, 35 }  ,opacity=1 ]  {$h_{1}$};
\draw (109.8,56) node [anchor=north west][inner sep=0.75pt]  [font=\small,color={rgb, 255:red, 108; green, 4; blue, 35 }  ,opacity=1 ]  {$h_{2}$};
\draw (89.2,19.2) node [anchor=north west][inner sep=0.75pt]  [font=\small]  {$x$};
\draw (102.33,78.4) node [anchor=north west][inner sep=0.75pt]  [font=\small]  {$v$};

\end{tikzpicture}
\vspace*{-1cm} 
	\end{center} 
	\caption{Figure for the proof of \cref{prop:3-cyc11}, case \eqref{case:2}. A  $3_2$-cycle quartet network with 
		internal cut edge contracted to length 0, other edge probabilities $h_1,h_2,x$, and hybridization parameter $\gamma$.}\label{fig:3-cyc11} 
\end{figure}
 
In case \eqref{case:2}, let $v$ denote the hybrid node in the 3-cycle, so $a,b$ are not descendants of $v$ for any rooting. 
(\Cref{fig:two3cyc11} (R)).
The value of any $CF$ involving at most one of $a,b$ is determined by the network and numerical parameters below $v$, 
since as a gene tree forms either a coalescent event occurs below $v$,
or 3 or 4 lineages reach $v$, so that
all three gene quartet topologies have probability 1/3. Thus the 3-cycle only affects values of $CF$s involving both $a$
and $b$, and only through events in which no coalescence has occurred below $v$. We may thus replace the cycle and its adjacent edges to $a,b$ 
with any graphical structure and parameters that produce the same 
probabilities of gene quartet topologies when exactly two lineages enter at $v$. These conditional probabilities are
the CFs of the quartet network shown in \cref{fig:3-cyc11}:
\begin{gather}
CF_{ac|bc}=\left( \gamma^2 h_1+(1-\gamma)^2h_2 \right )/3 +\gamma(1-\gamma)\left ( 1-x/3 \right ),\ \ \ \ CF_{ab|cc}= 1-2  CF_{ac|bc}. \notag\\
CF_{ac|bc}=   (\gamma^2 h_1+(1-\gamma)^2h_2 )/3 +\gamma (1 - \gamma) (1 - x/3 )    \label{eq:CFacbc}
\end{gather}
Note that we have dropped the subscripts $1,2$ from the $c$ taxa, since by 
 exchangeability of those lineages under the NMSC, they may be assigned arbitrarily.
 
 Now a quartet tree with topology $ab|cc$ and internal edge probability $z$ yields 
$$
CF_{ac|bc}=z/3,\ \ \ \ 
 CF_{ab|cc}= 1-2  CF_{ac|bc}.
$$
  so,  using (\ref{eq:CFacbc}),  without changing the $CF$s the 3-cycle and edges to $a,b$ in $N$ 
  could be replaced by a 3-leaf tree with an edge leading to an $ab$ cherry  having edge probability 
  $$z=  \gamma^2 h_1+(1-\gamma)^2h_2 +\gamma(1-\gamma)\left ( 3-x \right ).$$
  provided
 $0 < z < 1$.
 Since this inequality holds on a set of positive measure in 
parameter space, on that set the topologies $ N$ and $N'$ are not distinguishable.

However, $z>1$ also occurs on a set of positive measure.
Suppose in this case that the edge  $e=(v,w)$ below $v$ has as its child $w$ a node outside of a cycle, and let $c_1,c
_2$ be taxa chosen from distinct taxon blocks below that node. Then if parameters on $ N$ are in the set determined by $z>1$  
and the edge probability  $p$ for $e$ satisfies $pz>1$, then for $ N$ 
$$CF_{ac|bc}= pz/3>1/3.$$
Since for a quartet tree $CF_{ac|bc}<1/3,$ $N$ is distinguishable from $N'$ on this set. 

If $w$ is  instead in a cycle, a similar argument applies.
 \end{proof}
  
This proof essentially follows arguments given in \cite{Banos2019} for quartet networks with a $3_1$-cycle and  $3_2$-cycle. In case (2) the parameters for which the 3-cycle is topologically identifiable are ones that make the quartet network anomalous, in the sense of \cite{AneEtAl2023}.

\subsubsection{3-cycles on small networks: Algebraic conditions}\label{sec:3smallalg}

\begin{figure}
\begin{center}
	\tikzset{every picture/.style={line width=0.75pt}} 

\begin{tikzpicture}[x=0.75pt,y=0.75pt,yscale=-0.9,xscale=0.9]

\draw [line width=1.5]    (46,49.17) -- (71,47.17) ;
\draw [line width=1.5]    (71,47.17) -- (96.5,71.67) ;
\draw [line width=1.5]    (81.72,127.45) -- (95,106.17) ;
\draw [line width=1.5]    (96.5,71.67) -- (125,58.67) ;
\draw [line width=1.5]    (108.72,126.95) -- (95,106.17) ;
\draw [line width=1.5]    (96.5,71.67) -- (95,106.17) ;
\draw [line width=1.5]    (219.57,50.17) -- (239.06,62.77) ;
\draw [line width=1.5]    (239.06,62.77) -- (286.44,62.77) ;
\draw [color={rgb, 255:red, 182; green, 5; blue, 60 }  ,draw opacity=1 ][line width=1.5]    (239.06,62.77) -- (262.18,106.61) ;
\draw [shift={(263.58,109.26)}, rotate = 242.19] [color={rgb, 255:red, 182; green, 5; blue, 60 }  ,draw opacity=1 ][line width=1.5]    (14.21,-4.28) .. controls (9.04,-1.82) and (4.3,-0.39) .. (0,0) .. controls (4.3,0.39) and (9.04,1.82) .. (14.21,4.28)   ;
\draw [color={rgb, 255:red, 182; green, 5; blue, 60 }  ,draw opacity=1 ][line width=1.5]    (286.44,62.77) -- (264.91,106.57) ;
\draw [shift={(263.58,109.26)}, rotate = 296.18] [color={rgb, 255:red, 182; green, 5; blue, 60 }  ,draw opacity=1 ][line width=1.5]    (14.21,-4.28) .. controls (9.04,-1.82) and (4.3,-0.39) .. (0,0) .. controls (4.3,0.39) and (9.04,1.82) .. (14.21,4.28)   ;
\draw [line width=1.5]    (286.44,62.77) -- (307.3,51.04) ;
\draw [line width=1.5]    (263.58,109.26) -- (264,131.96) ;
\draw [line width=1.5]    (58,24.67) -- (71,47.17) ;
\draw [line width=1.5]    (307.3,51.04) -- (318,33.67) ;
\draw [line width=1.5]    (307.3,51.04) -- (325,59.67) ;
\draw [line width=1.5]    (219.57,50.17) -- (207,33.17) ;
\draw [line width=1.5]    (200.5,55.17) -- (219.57,50.17) ;
\draw [line width=1.5]    (420.57,50.17) -- (440.06,62.77) ;
\draw [line width=1.5]    (440.06,62.77) -- (487.44,62.77) ;
\draw [color={rgb, 255:red, 182; green, 5; blue, 60 }  ,draw opacity=1 ][line width=1.5]    (440.06,62.77) -- (463.18,106.61) ;
\draw [shift={(464.58,109.26)}, rotate = 242.19] [color={rgb, 255:red, 182; green, 5; blue, 60 }  ,draw opacity=1 ][line width=1.5]    (14.21,-4.28) .. controls (9.04,-1.82) and (4.3,-0.39) .. (0,0) .. controls (4.3,0.39) and (9.04,1.82) .. (14.21,4.28)   ;
\draw [color={rgb, 255:red, 182; green, 5; blue, 60 }  ,draw opacity=1 ][line width=1.5]    (487.44,62.77) -- (465.91,106.57) ;
\draw [shift={(464.58,109.26)}, rotate = 296.18] [color={rgb, 255:red, 182; green, 5; blue, 60 }  ,draw opacity=1 ][line width=1.5]    (14.21,-4.28) .. controls (9.04,-1.82) and (4.3,-0.39) .. (0,0) .. controls (4.3,0.39) and (9.04,1.82) .. (14.21,4.28)   ;
\draw [line width=1.5]    (487.44,62.77) -- (508.3,51.04) ;
\draw [line width=1.5]    (464.58,109.26) -- (465,131.96) ;
\draw [line width=1.5]    (420.57,50.17) -- (408,33.17) ;
\draw [line width=1.5]    (401.5,55.17) -- (420.57,50.17) ;
\draw [line width=1.5]    (465,131.96) -- (448.5,147.67) ;
\draw [line width=1.5]    (465,131.96) -- (480.5,148.17) ;

\draw (30.33,44.93) node [anchor=north west][inner sep=0.75pt]  [font=\normalsize,color={rgb, 255:red, 6; green, 79; blue, 97 }  ,opacity=1 ]  {$b_{1}$};
\draw (69.74,129.69) node [anchor=north west][inner sep=0.75pt]  [font=\normalsize,color={rgb, 255:red, 6; green, 79; blue, 97 }  ,opacity=1 ]  {$a_{1}$};
\draw (128.24,49.69) node [anchor=north west][inner sep=0.75pt]  [font=\normalsize,color={rgb, 255:red, 6; green, 79; blue, 97 }  ,opacity=1 ]  {$c$};
\draw (105.24,129.19) node [anchor=north west][inner sep=0.75pt]  [font=\normalsize,color={rgb, 255:red, 6; green, 79; blue, 97 }  ,opacity=1 ]  {$a_{2}$};
\draw (99.8,83.3) node [anchor=north west][inner sep=0.75pt]  [font=\small,color={rgb, 255:red, 0; green, 0; blue, 0 }  ,opacity=1 ]  {$\ell _{1}$};
\draw (319.83,20.22) node [anchor=north west][inner sep=0.75pt]  [font=\normalsize,color={rgb, 255:red, 6; green, 79; blue, 97 }  ,opacity=1 ]  {$b_{1}$};
\draw (184.24,48.98) node [anchor=north west][inner sep=0.75pt]  [font=\normalsize,color={rgb, 255:red, 6; green, 79; blue, 97 }  ,opacity=1 ]  {$a_{1}$};
\draw (259.24,133.98) node [anchor=north west][inner sep=0.75pt]  [font=\normalsize,color={rgb, 255:red, 6; green, 79; blue, 97 }  ,opacity=1 ]  {$c$};
\draw (327,53.07) node [anchor=north west][inner sep=0.75pt]  [font=\normalsize,color={rgb, 255:red, 6; green, 79; blue, 97 }  ,opacity=1 ]  {$b_{2}$};
\draw (252.2,74.99) node [anchor=north west][inner sep=0.75pt]  [font=\small,color={rgb, 255:red, 182; green, 5; blue, 60 }  ,opacity=1 ]  {$\gamma $};
\draw (226.1,41.29) node [anchor=north west][inner sep=0.75pt]  [font=\small,color={rgb, 255:red, 0; green, 0; blue, 0 }  ,opacity=1 ]  {$\ell _{1}$};
\draw (293.8,59.59) node [anchor=north west][inner sep=0.75pt]  [font=\small,color={rgb, 255:red, 0; green, 0; blue, 0 }  ,opacity=1 ]  {$\ell _{2}$};
\draw (258.2,49.99) node [anchor=north west][inner sep=0.75pt]  [font=\small]  {$x$};
\draw (41.33,7.93) node [anchor=north west][inner sep=0.75pt]  [font=\normalsize,color={rgb, 255:red, 6; green, 79; blue, 97 }  ,opacity=1 ]  {$b_{2}$};
\draw (82.8,43.8) node [anchor=north west][inner sep=0.75pt]  [font=\small,color={rgb, 255:red, 0; green, 0; blue, 0 }  ,opacity=1 ]  {$\ell _{2}$};
\draw (188.24,22.48) node [anchor=north west][inner sep=0.75pt]  [font=\normalsize,color={rgb, 255:red, 6; green, 79; blue, 97 }  ,opacity=1 ]  {$a_{2}$};
\draw (384.33,45.72) node [anchor=north west][inner sep=0.75pt]  [font=\normalsize,color={rgb, 255:red, 6; green, 79; blue, 97 }  ,opacity=1 ]  {$b_{1}$};
\draw (437.74,150.48) node [anchor=north west][inner sep=0.75pt]  [font=\normalsize,color={rgb, 255:red, 6; green, 79; blue, 97 }  ,opacity=1 ]  {$a_{1}$};
\draw (392,16.57) node [anchor=north west][inner sep=0.75pt]  [font=\normalsize,color={rgb, 255:red, 6; green, 79; blue, 97 }  ,opacity=1 ]  {$b_{2}$};
\draw (455.2,74.99) node [anchor=north west][inner sep=0.75pt]  [font=\small,color={rgb, 255:red, 182; green, 5; blue, 60 }  ,opacity=1 ]  {$\gamma $};
\draw (466.6,114.29) node [anchor=north west][inner sep=0.75pt]  [font=\small]  {$\ell _{1}$};
\draw (428.3,39.09) node [anchor=north west][inner sep=0.75pt]  [font=\small]  {$\ell _{2}$};
\draw (459.2,47.99) node [anchor=north west][inner sep=0.75pt]  [font=\small]  {$x$};
\draw (476.24,150.48) node [anchor=north west][inner sep=0.75pt]  [font=\normalsize,color={rgb, 255:red, 6; green, 79; blue, 97 }  ,opacity=1 ]  {{$a_{2}$}};
\draw (511.24,41.98) node [anchor=north west][inner sep=0.75pt]  [font=\normalsize,color={rgb, 255:red, 6; green, 79; blue, 97 }  ,opacity=1 ]  {$c$};
\draw (481.8,77.59) node [anchor=north west][inner sep=0.75pt]  [font=\small,color={rgb, 255:red, 108; green, 4; blue, 35 }  ,opacity=1 ]  {$h_{2}$};
\draw (433.3,78.09) node [anchor=north west][inner sep=0.75pt]  [font=\small,color={rgb, 255:red, 108; green, 4; blue, 35 }  ,opacity=1 ]  {$h_{1}$};

\end{tikzpicture}

\vspace*{-1cm}
\end{center}
\caption{ (L) The 5-taxon unrooted binary tree $ T_5$;
 (C) the 5-taxon network $N_{5-3_1}$ with a $3_1$-cycle; and (R) the 5-taxon network $N_{5-3_2}$ with a $3_2$-cycle,
with numerical parameters shown. Edge probabilities of hybrid edges in $N_{5-3_1}$ and of pendant edges in networks are omitted, 
since they do not appear in formulas for the $CF$s. }\label{fig:5_3cycle}
\end{figure}
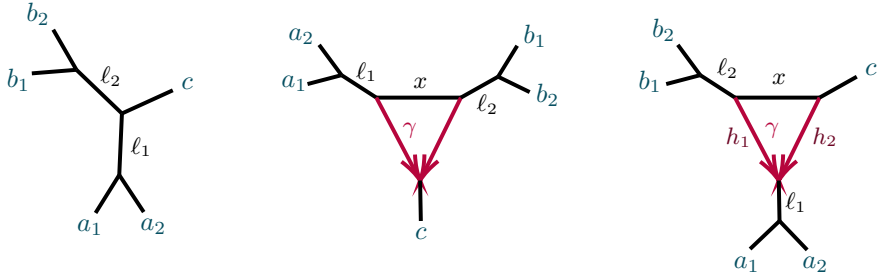

\Cref{fig:5_3cycle} shows a 5-taxon tree, $T_5$, and  two 5-taxon networks with 3-cycles, $N_{5-3_1}, N_{5-3_2}$.
 \Cref{prop:T5,prop:N5-3-1,prop:N5-3-2} of the supplementary materials give computational results on the ideals $\mathcal I(T_5)$, $\mathcal I(N_{5-3_1})$, and $\mathcal I(N_{5-3_2})$, showing that the polynomial
\begin{equation}
f_{abc}= 3CF_{ab|ac}CF_{ab|bc} - CF_{ab|ab}\label{eq:f}
 \end{equation}
 is in $\mathcal I(T_5)$, but not in $\mathcal I(N_{5-3_1})$ nor $\mathcal I(N_{5-3_2})$. 
 Using expressions for CFs in terms of parameters from \cref{prop:T5}, $f_{abc}$ can be 
 interpreted as expressing the total internal path length in the tree $T_5$ is the sum of the 
 lengths of the two internal edges. This polynomial, and variants of it, will play an important role in 
 identifying 3-cycles. The first result in this direction is the following.

\begin{theorem}\label{thm:5tax3cycle} Under the NMSC model, the vanishing of 
	$f_{abc}$ distinguishes a 5-taxon unrooted tree $T_5$ from the 5-taxon semidirected 
	networks with a central 3-cycle whose contraction yields the tree $T_5$, for generic numerical parameters.
\end{theorem}
\begin{proof} Consider the networks of \cref{fig:5_3cycle} and a fourth obtained by 
	interchanging the $a,b$ taxa in \cref{fig:5_3cycle}(R). Since $f_{abc}\notin \mathcal I(N)$ 
	for the non-tree $N$, it does not vanish for all parameters on them, and is zero only for a 
	set of measure zero in their parameter space. Thus generically  the vanishing of $f_{abc}$ distinguishes $T_5$ from the others.
\end{proof}

 \cref{prop:T5,prop:N5-3-1,prop:N5-3-2}  also show that the two 5-taxon networks of \cref{fig:5_3cycle} have the same associated ideals, 
 $\mathcal I(N_{5-3_1})=\mathcal I(N_{5-3_2})\subset \mathcal I(T_5)$. As a result, there is no purely algebraic means (using only polynomial equalities) of distinguishing them using $CF$s. 
 
\begin{figure}
	\begin{center}
	\tikzset{every picture/.style={line width=0.75pt}} 

\begin{tikzpicture}[x=0.75pt,y=0.75pt,yscale=-0.9,xscale=0.9]

\draw [line width=1.5]    (43.5,61.67) -- (68.5,59.67) ;
\draw [line width=1.5]    (68.5,59.67) -- (94,84.17) ;
\draw [line width=1.5]    (79.22,139.95) -- (92.5,118.67) ;
\draw [line width=1.5]    (94,84.17) -- (127,69.17) ;
\draw [line width=1.5]    (106.22,139.45) -- (92.5,118.67) ;
\draw [line width=1.5]    (94,84.17) -- (92.5,118.67) ;
\draw [line width=1.5]    (253.57,49.17) -- (273.06,61.77) ;
\draw [line width=1.5]    (273.06,61.77) -- (320.44,61.77) ;
\draw [color={rgb, 255:red, 182; green, 5; blue, 60 }  ,draw opacity=1 ][line width=1.5]    (273.06,61.77) -- (296.18,105.61) ;
\draw [shift={(297.58,108.26)}, rotate = 242.19] [color={rgb, 255:red, 182; green, 5; blue, 60 }  ,draw opacity=1 ][line width=1.5]    (14.21,-4.28) .. controls (9.04,-1.82) and (4.3,-0.39) .. (0,0) .. controls (4.3,0.39) and (9.04,1.82) .. (14.21,4.28)   ;
\draw [color={rgb, 255:red, 182; green, 5; blue, 60 }  ,draw opacity=1 ][line width=1.5]    (320.44,61.77) -- (298.91,105.57) ;
\draw [shift={(297.58,108.26)}, rotate = 296.18] [color={rgb, 255:red, 182; green, 5; blue, 60 }  ,draw opacity=1 ][line width=1.5]    (14.21,-4.28) .. controls (9.04,-1.82) and (4.3,-0.39) .. (0,0) .. controls (4.3,0.39) and (9.04,1.82) .. (14.21,4.28)   ;
\draw [line width=1.5]    (320.44,61.77) -- (341.3,50.04) ;
\draw [line width=1.5]    (297.58,108.26) -- (298,130.96) ;
\draw [line width=1.5]    (55.5,37.17) -- (68.5,59.67) ;
\draw [line width=1.5]    (341.3,50.04) -- (352,32.67) ;
\draw [line width=1.5]    (341.3,50.04) -- (359,58.67) ;
\draw [line width=1.5]    (253.57,49.17) -- (241,32.17) ;
\draw [line width=1.5]    (234.5,54.17) -- (253.57,49.17) ;
\draw [line width=1.5]    (298,130.96) -- (281.5,146.67) ;
\draw [line width=1.5]    (298,130.96) -- (313.5,147.17) ;
\draw [line width=1.5]    (127,69.17) -- (137.7,51.8) ;
\draw [line width=1.5]    (127,69.17) -- (144.7,77.8) ;

\draw (25.83,53.43) node [anchor=north west][inner sep=0.75pt]  [font=\normalsize,color={rgb, 255:red, 6; green, 79; blue, 97 }  ,opacity=1 ]  {$b_{1}$};
\draw (66.24,141.69) node [anchor=north west][inner sep=0.75pt]  [font=\normalsize,color={rgb, 255:red, 6; green, 79; blue, 97 }  ,opacity=1 ]  {$a_{1}$};
\draw (102.74,141.69) node [anchor=north west][inner sep=0.75pt]  [font=\normalsize,color={rgb, 255:red, 6; green, 79; blue, 97 }  ,opacity=1 ]  {$a_{2}$};
\draw (77.3,94.8) node [anchor=north west][inner sep=0.75pt]  [font=\small]  {$\ell _{1}$};
\draw (352.83,19.22) node [anchor=north west][inner sep=0.75pt]  [font=\normalsize,color={rgb, 255:red, 6; green, 79; blue, 97 }  ,opacity=1 ]  {$c_{1}$};
\draw (218.24,47.98) node [anchor=north west][inner sep=0.75pt]  [font=\normalsize,color={rgb, 255:red, 6; green, 79; blue, 97 }  ,opacity=1 ]  {$b_{1}$};
\draw (362,53.07) node [anchor=north west][inner sep=0.75pt]  [font=\normalsize,color={rgb, 255:red, 6; green, 79; blue, 97 }  ,opacity=1 ]  {$c_{2}$};
\draw (286.2,74.99) node [anchor=north west][inner sep=0.75pt]  [font=\small,color={rgb, 255:red, 182; green, 5; blue, 60 }  ,opacity=1 ]  {$\gamma $};
\draw (260.1,40.29) node [anchor=north west][inner sep=0.75pt]  [font=\small]  {$\ell _{2}$};
\draw (327.8,58.59) node [anchor=north west][inner sep=0.75pt]  [font=\small]  {$\ell _{3}$};
\draw (292.2,48.99) node [anchor=north west][inner sep=0.75pt]  [font=\small]  {$x$};
\draw (39.83,21.43) node [anchor=north west][inner sep=0.75pt]  [font=\normalsize,color={rgb, 255:red, 6; green, 79; blue, 97 }  ,opacity=1 ]  {$b_{2}$};
\draw (79.3,55.8) node [anchor=north west][inner sep=0.75pt]  [font=\small]  {$\ell _{2}$};
\draw (222.24,18.48) node [anchor=north west][inner sep=0.75pt]  [font=\normalsize,color={rgb, 255:red, 6; green, 79; blue, 97 }  ,opacity=1 ]  {$b_{2}$};
\draw (270.74,149.48) node [anchor=north west][inner sep=0.75pt]  [font=\normalsize,color={rgb, 255:red, 6; green, 79; blue, 97 }  ,opacity=1 ]  {$a_{1}$};
\draw (309.24,149.48) node [anchor=north west][inner sep=0.75pt]  [font=\normalsize,color={rgb, 255:red, 6; green, 79; blue, 97 }  ,opacity=1 ]  {$a_{2}$};
\draw (139.83,39.4) node [anchor=north west][inner sep=0.75pt]  [font=\normalsize,color={rgb, 255:red, 6; green, 79; blue, 97 }  ,opacity=1 ]  {$c_{1}$};
\draw (147,73.26) node [anchor=north west][inner sep=0.75pt]  [font=\normalsize,color={rgb, 255:red, 6; green, 79; blue, 97 }  ,opacity=1 ]  {$c_{2}$};
\draw (108.8,78.8) node [anchor=north west][inner sep=0.75pt]  [font=\small]  {$\ell _{3}$};
\draw (300.08,114.16) node [anchor=north west][inner sep=0.75pt]  [font=\small]  {$\ell _{1}$};
\draw (267.58,82.66) node [anchor=north west][inner sep=0.75pt]  [font=\small,color={rgb, 255:red, 108; green, 4; blue, 35 }  ,opacity=1 ]  {$h_{1}$};
\draw (312.08,84.16) node [anchor=north west][inner sep=0.75pt]  [font=\small,color={rgb, 255:red, 108; green, 4; blue, 35 }  ,opacity=1 ]  {$h_{2}$};

\end{tikzpicture}
\vspace*{-1cm} 
	\end{center}
	\caption{(L) The 6-taxon tree $T_6$ with three cherries. (R) The 6-taxon network $N_a$ with a central 3-cycle surrounded by 3 cherries, with $a_1,a_2$ descending from the hybrid node. 
	The network $N_b$ is obtained by `rotating' the three pairs of taxa so $b_1,b_2$  descend from the hybrid node, and similarly for $N_c$.} \label{fig:222nets}
\end{figure}
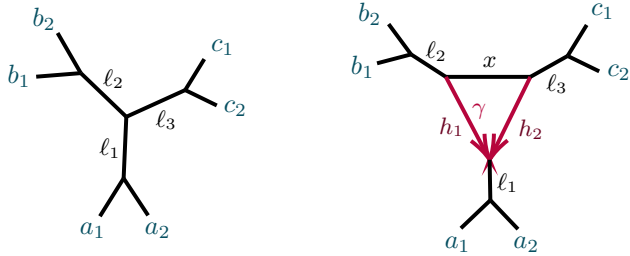

\smallskip 

Computational results for the 6-taxon networks $T_6$ and $N_a$ of \cref{fig:222nets} appear in \cref{prop:T6,prop:N6-2-2-2}. Note that $\mathcal I(T_6)$ contains 3 polynomials, $f_{abc}, f_{bca}, f_{cab}$, none of which are in $\mathcal I(N_a)$, expressing three different internal path length relationships in the tree.
\cref{prop:N6-2-2-2} implies $\mathcal I(N_a)=\mathcal I(N_b)=\mathcal I(N_c)$, where $N_b$ and $N_c$ differ from $N_a$ in which taxa are below the hybrid node. Thus the hybrid node of the 3-cycle in these three networks cannot be determined from purely algebraic conditions on   $CF$s. While the vanishing of any of the three 
$f_{abc},f_{bca}, f_{cab}$ (and hence all) distinguishes the tree $T_6$ from $N_a$, $N_b,$ and $N_c$, that was already implicit in \cref{thm:5tax3cycle}. 

\smallskip

\subsubsection{3-cycles on small networks: Semialgebraic conditions}

While \cref{sec:3smallalg} has shown the presence of a 3-cycle can be
detected
in some networks, that result pertains only to the undirected cycle.
To obtain information on the hybrid node, we use a semialgebraic approach, focusing on polynomial inequalities.
						
\begin{proposition} \label{prop:5tax3cyc}
Let $N$ be one of $T_5$, $N_{5-3_1}$, $N_{5-3_2}$   of \Cref{fig:5_3cycle}, or the network 
$N'_{5-3_2}$ obtained from interchanging the $a,b$ taxon labels on $N_{5-3_2}$.
Let $f_{abc}$ be as in \cref{eq:f}.
Then  for generic numerical parameters under the NMSC,			 
			\begin{equation*}
			N=\begin{cases}
			T_5, & \text{ if, and only if, }  f_{abc}=0\\
			N_{5-3_2}\text{ or } N'_{5-3_2}, & \text{ if  } f_{abc}<0\\
 			N_{5-3_1}, N_{5-3_2},\text{ or } N'_{5-3_2}, &\text{ if  } f_{abc}>0
 			\end{cases}
 		\end{equation*}
		
Moreover, $f_{abc}$ is identical on the networks $N_{5-3_2}$  and $N'_{5-3_2}$ for the same parameter values, so $f_{abc}$ gives no information to distinguish between these. 
		
Finally, there are positive measure subsets of the numerical parameter space for  $N_{5-3_2}$ on which  $f_{abc}<0$ and on which  $f_{abc}>0$.
\end{proposition}
 		
 \begin{proof}  \Cref{thm:5tax3cycle} states that $f_{abc}=0$ for generic parameters if, and only if, $N=T_5$.
				If $N = N_{5-3_1}$ then using the formulas for $CF$s in \cref{prop:N5-3-1} gives, for $\gamma,x, \ell_1, \ell_2 \in(0,1)$,
 			\begin{align*}
 				f_{abc} &= [\ell_1\ell_2 (\gamma + x-\gamma x)(1-\gamma + \gamma x) -\ell_1\ell_2x]/3 \\
 				&=\ell_1\ell_2\gamma(1-\gamma)(x-1)^2/3>0.
 			\end{align*}		

Since $f_{abc}$ is invariant under interchanging the $a$s and $b$s, its values for $N_{5-3_2}$  and $N'_{5-3_2}$ are the same.
		
Specific examples of parameters on $N_{5-3_2}$ show both $f_{abc}<0$ and $f_{abc}>0$ can occur, 
and by continuity there are positive measure subsets of parameter space on which these occur. 		
\end{proof}

If a 5-taxon network does have a 3-cycle $C$, then this proposition may provide some information on the hybrid node's location. 
For instance,  $f_{abc}<0$ implies the taxon $c$ which is not in a cherry on the tree 
obtained by contracting $C$ to a vertex
is also not a hybrid descendant of the 3-cycle. However, for other numerical  parameters
$f_{abc} > 0$, in which case there is no information on the hybrid location.

\smallskip

To further develop semialgebraic tests for 3-cycle hybrid nodes, we again consider the 
6-taxon networks $N_a,N_b,N_c$ described in  \cref{fig:222nets}.
Define the following functions of the $CF$s, building on the $f_{xyz}$:
\begin{align}\label{eq:Gs}
		G_{abc} &=-f_{abc}CF_{ac|bc}+2 f_{bca}C_{ab|ac} -f_{cab}C_{ab|bc}\notag \\
		&=CF_{ac|ac}CF_{ab|bc}-2CF_{bc|bc}CF_{ab|ac} + CF_{ab|ab}CF_{ac|bc},\\
		G_{cab}&=CF_{bc|bc}CF_{ab|ac}-2CF_{ab|ab}CF_{ac|bc} + CF_{ac|ac}CF_{ab|bc}\notag,\\
		G_{bca}&=CF_{ab|ab}CF_{ac|bc}-2CF_{ac|ac}CF_{ab|bc} + CF_{bc|bc}CF_{ab|ac}.\notag
\end{align}
Note that $G_{xyz}\in\mathcal  I(T_6)$, $G_{xyz}=G_{xzy}$ and $G_{abc} + G_{cab} + G_{bca}=0$.

\begin{proposition}\label{prop:222semi}
Under the NMSC, for $CF$s arising from the tree $T_6$,  $G_{xyz}=0$
for all $x,y,z$, while $G_{xyz}>0$
for $CF$s arising from the network $N_x$.

 If a network is known to have one of the topologies $N_a,N_b,N_c$,
  then at least one of these topologies can be ruled out by the signs of 
$G_{abc},G_{cab},G_{bca}$: If $G_{xyz}<0$ then the network is not $N_x$.

Finally, there are positive measure subsets of the numerical parameter space for  
$N_y$ and $N_z$ on which  $G_{xyz}<0$ and on which  $G_{xyz}>0$.

\end{proposition}

\begin{proof} That $G_{xyz}=0$ for $T_6$ restates  that $G_{xyz}\in\mathcal  I(T_6)$.
Using formulas from \cref{prop:N6-2-2-2}, for $CF$s from $N_a$,
\begin{align*}
9\, G_{abc}&=CF_{ac|ac}CF_{ab|bc}-2CF_{bc|bc}CF_{ab|ac} + CF_{ab|ab}CF_{ac|bc}\\
&=\ell_1\ell_3(\gamma^2h_1x + \gamma^2h_2-2\gamma^2-2\gamma h_2 + 2\gamma + h_2) 
\ell_2(x + \gamma-\gamma x) \\
&  \ \ \ \ -2x\ell_2\ell_3  \ell_1(\gamma^2h_1 + \gamma^2h_2 + \gamma^2x-3\gamma^2-2\gamma h_2-\gamma x + 3\gamma + h_2)\\
&\ \ \ \ + \ell_1\ell_2(\gamma^2h_1 + \gamma^2h_2x-2\gamma^2-2\gamma h_2x + 2\gamma + h_2x) \ell_3(\gamma x-\gamma + 1)\\
&=\gamma(1-\gamma)(1-x)^2\ell_1\ell_2\ell_3  [h_1\gamma  + h_2(1-\gamma)  + 2]>0.
\end{align*}

Since $G_{abc}+G_{cab}+G_{bca}=0$  and one of these terms is positive for each of $N_a,N_b,N_c$, at  least one is negative.

One can find specific parameters on $N_x$ for which $G_{xyz}<0$ and $G_{xyz}>0$,
and by continuity these conditions hold on sets of positive measure.
\end{proof} 

\begin{figure}
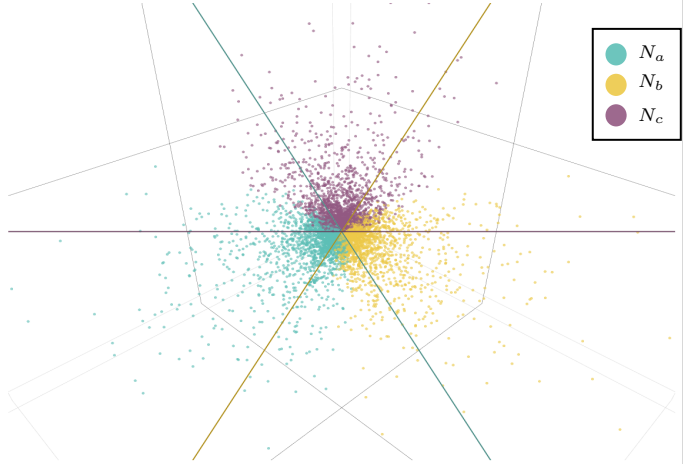

\begin{center}
\include{Figures/Gxyz_newColors.tex}
\end{center}
\caption{Values of $(G_{abc}, G_{bca},G_{cab})$ plotted in three dimensions, for random numerical parameter 
	values on each of the three networks $N_a,N_b,N_c$.
Color indicates network topology. Plotted points lie in the plane $x+y+z=0$, which is viewed orthogonally. 
The three coordinate planes $x=0,y=0,z=0$ intersect this plane in the colored lines, separating the points by color 
into overlapping half-planes.  
Numerical parameters for networks were chosen uniformly from the interval $[0,1]$. }\label{fig:Gplot}
\end{figure}

\Cref{fig:Gplot} illustrates the proposition, showing $(G_{abc}, G_{bca},G_{cab})$ for randomly chosen numerical parameters on 
each of the networks $N_a,N_b,N_c$, with color indicating the network topology. Since the points lie in a plane $P$ through the origin, 
the axes have been rotated to view the plane orthogonally.
The three planes $G_{xyz}=0$ intersect $P$ in lines which divide the plot into six sectors. On three of these sectors 
exactly
one color appears, indicating that the network topology is determined by the positivity of 
exactly one $G_{xyz}$. On the  
3 sectors where two colors appear, two of the $G_{xyz}$ are positive, 
so only one of the network topologies is ruled out. 

\begin{rmk} 
It is natural to ask if $f_{xyz}$ or $G_{xyz}$ could be used to detect 3-cycles in situations where incomplete lineage sorting 
is negligible, so that all gene trees are displayed on the species network.
This scenario is modeled by immediate coalescence of gene lineages on entering a common network edge
or, equivalently, by a limiting model of the NMSC, in which all edge probabilities go to 0. (See \cite[Section 6.2]{logdetNet} for more details.)
The formulae for 
$CF$s given in this work still apply, and for all the 5-taxon networks of  \cref{prop:5tax3cyc} $f_{abc}=0$, while for
all the 6-taxon networks of \cref{prop:222semi} $G_{xyz}=0$. Indeed, these functions depend only on 
$CF$s for quartets \emph{not}  displayed on the networks, which
are therefore all zero. 
A coalescent process is thus essential to detecting 3-cycles with these functions.
\end{rmk}

\cref{prop:222semi} and \cref{fig:Gplot}  suggest determining which of $N_a,N_b,$ or $N_c$ produced certain numerical $CF$s
may be impossible, which we rigorously show by the following example.

\begin{example}[Non-identifiability of the hybrid node in a 3-cycle]\label{ex:NonID}
Consider the network $N_a$ with parameters
	\begin{gather*}
		\gamma^{(a)}=\frac{28}{100} ,\quad h_1^{(a)}=\frac{83}{100}, \quad h_2^{(a)}=\frac{78}{100}, \quad x^{(a)}=\frac{98}{100},\\ 
			\quad \ell_1^{(a)}=\frac{88}{100}, \quad \ell_2^{(a)}=\frac{61}{100}, \quad \ell_3^{(a)}=\frac{50}{100},
	\end{gather*}
%
	and the network $N_b$ with parameters 
	\begin{align*}
		&\gamma^{(b)}=\frac{236700}{253367}, \quad h_1^{(b)}=\frac{84456638}{87243675}, \quad h_2^{(b)}=\frac{27286250}{31593489}, \\
		&x^{(b)}=\frac{2722883}{2976250},  \quad \ell_1^{(b)}=\frac{809409}{1315000}, \quad \ell_2^{(b)}=\frac{1}{2}, \quad \ell_3^{(b)}=\frac{26191}{31250},
	\end{align*}
	where the parameters for $N_b$ are as shown for $N_a$ in \cref{fig:222nets} but with taxon labels 
	$(a_1,a_2)$, $(b_1,b_2)$ and $(c_1,c_2)$ replaced by $(b_1,b_2)$, $(c_1,c_2)$ and $(a_1,a_2)$ respectively.
	Then the $CF$s of $N_a$ and $N_b$ are equal. Specifically, for both $N_a$ and $N_b$,
	\begin{align*}
		&CF_{bc|bc} = \frac{2989}{30000}, \quad CF_{ab|ab}=\frac{906412969}{5859375000} ,\quad CF_{ac|ac}=\frac{29951713}{234375000}, \\
		&CF_{ab|ac}=\frac{602701}{2343750} \quad CF_{ab|bc}=\frac{9394}{46875}, \quad CF_{ac|bc}=\frac{1243}{7500}.
	\end{align*}
\end{example}

In fact, there is a neighborhood in $\mathcal V(N_a)=\mathcal V(N_b)$ of the $CF$ point of this example 
contained in the image of the parameterizations of both $N_a$ and $N_b$.
Indeed, a computation of the Jacobians for the two parameterization maps at the example parameters shows that 
locally the images are of dimension 6, which matches the dimension of the variety. A sufficiently small neighborhood 
of the $CF$ point is thus in the image of the parameterizations for both $N_a$ and $N_b$, with inverse images of positive measure. 
One may similarly show, using a $CF$ point that arises only from $N_a$ (lying in a uniformly colored sector in \cref{fig:Gplot}), 
that there  is a set of positive measure in the $N_a$ parameter space which gives $CF$s in the image of the parametrization of $N_a$ only.  
We combine these results formally in the following theorem.

\begin{theorem}\label{thm:222twosets}
There exists a positive measure subset of the numerical parameter space of $N_a$ for which it is distinguishable from $T_6$, $N_b$, and $N_c$, and a positive measure subset of the parameter space for 
which only the undirected network can be distinguished from $T_6$, with 1 node in the 3-cycle determined to be non-hybrid.
\end{theorem}

Again using the parameter values in \cref{ex:NonID}, an analog of this result for 5-taxon networks 
with a single 3-cycle can be established. 

\begin{theorem}\label{thm:221twosets}
There exist positive measure subsets of the numerical parameter spaces of $N_{5-3_1}$  and $N_{5 -3_2}$ for 
which the semidirected network topologies are 
distinguishable from the other networks among $T_5,N_{5-3_1}, N_{5-3_2}, N'_{5-3_2},$
and  positive measure subset of the parameter spaces for which they are not distinguishable from at least one other of  $N_{5-3_1}, N_{5-3_2}, N'_{5-3_2}$.
\end{theorem}
\begin{proof}
First, suppose the network is $N_{5-3_2}$. Dropping a taxon to pass to a quartet network with a $3_2$-cycle,  
\cref{prop:3-cyc11} implies that the semidirected topology is identifiable on some positive measure subset of parameters. 
That there is such a set on which the semidirected topology is not identifiable follows from using the parameter values 
of \cref{ex:NonID} (after dropping an appropriately chosen taxon) on such networks with different hybrid cherries, and computing Jacobians to verify 
that an open set  of such examples exists. 

To investigate identifiability for the network $N_{5-3_1}$, consider the function 
\begin{equation}\tilde f=f_{abc}-(1/2)CF_{ab|ab}.\label{eq:ftilde}
\end{equation}
We first show that $\tilde f <0$ for all parameters on $N_{5-3_2}$. Using \cref{prop:N5-3-2} to expand in terms of parameters, 
\begin{multline*}\tilde f=\ell_1\ell_2[ (\gamma^2h_1 + \gamma(1-\gamma)(3-x)+ (1-\gamma)^2h_2) (\gamma+ (1-\gamma) x)/3\\ 
-(\gamma^2h_1 + 2\gamma(1-\gamma)+ (1-\gamma)^2 h_2x)/2].\end{multline*}
Since $h_1$ appears linearly in this expression with a negative coefficient, we set $h_1=0$ to bound $\tilde f$ above. 
The coefficient of $h_2$, which also appears linearly, may be positive or negative, so we consider $h_2=0$ and $1$.
If $h_2=0$,
\begin{align*}\tilde f&=\ell_1\ell_2\left [ ( \gamma(1-\gamma)(3-x)) (\gamma+ (1-\gamma) x)/3 -\gamma(1-\gamma)\right ]\\
&=-\ell_1\ell_2\gamma(1-\gamma)\left [3(1-x)(1-\gamma)+x(\gamma+ (1-\gamma) x)\right ]/3<0,
 \end{align*}
while if $h_2=1$,
\begin{align*}\tilde f&=\ell_1\ell_2 \left[ ( \gamma(1-\gamma)(3-x)+ (1-\gamma)^2) (\gamma+ (1-\gamma) x)/3-( 2\gamma(1-\gamma)+ (1-\gamma)^2 x)/2\right]\\
&=-\ell_1\ell_2(1-\gamma) \left [ \gamma x(\gamma+ (1-\gamma) x)+
(1-\gamma)\left( 2 \gamma(1-x)  + x/2\right)\right ]/3<0.
\end{align*}
 It is easy, however, to find an open set of parameters for $N_{5-3_1}$ for which  $\tilde f>0$, and on that set $c$ is identifiable as the hybrid block.

We obtain a set on which the semidirected topology  of $N_{5-3_1}$ is not identifiable by again using the parameter values in \cref{ex:NonID}.
\end{proof}

\smallskip

\subsubsection{Large networks with 3-cycles}\label{sec:large3}

After considering specific 5- and 6-taxon networks with a single 3-cycle, we shift focus to 3-cycles in 
general networks  $N^ + $. We extend the previous results on semialgebraic identifiability of both cycles and hybrid nodes, using 
a decomposition of $N^ + $ into 4 subnetworks, as in \cref{fig:net3cycles}.  A similar decomposition 
is used in \cite{GrossEtAl2023}, of a level-1 network into trees and `sunlets,'  but that work does 
not model coalescence,  so the details are quite different. 
Our decomposition extends to larger cycles
but we present only the 3-cycle case needed here.

\begin{figure}\label{fig:decomp}
\begin{center}
\end{center}
\caption{(L) A decomposition of a level-1 network $ N^ + $ with a 3-cycle into 4 subnetworks,
	denoted
	$A,B,C,D$, with root in $C$.
(R) The semidirected 3-cycle network $N_{6-3_2}$ with 3 cherries, which is a simple instance of the
 network on the left.}\label{fig:net3cycles}
\end{figure}

The subnetworks in \cref{fig:net3cycles} are:
\begin{description}

\item[\ \ \ \ $D$:] The 3-cycle and its three adjacent cut edges, with pendant\\ vertices $a,b,c$, 
where $a$ is the child of the hybrid node of the cycle;
 
\item [\ \ \ \ $A,B,C$:] The connected components containing $a,b,c$, respectively,\\ when the edges and internal nodes of $D$  are deleted from $N^ + $. 

\end{description}
Note that $a,b,c$ are each in two of these subnetworks.
Since the root must be above $D$'s hybrid node, and the semidirected network is unchanged  by moving the root along tree edges, we may assume the root lies in $B$ or $C$, and, after renaming, in $C$.

\smallskip

The $CF$ of any quartet  under the NMSC on $N^ + $ has an algebraic decomposition into terms associated to the subnetworks $A,B,C,D$, which we next develop.  We use two facts about coalescent events between 4 lineages leading to gene quartets: 
\begin{enumerate}
\item The first coalescent event between 2 of the lineages determines the gene quartet tree that forms, and 
\item Conditioned on  3 or 4 lineages reaching a common node with no previous coalescence, by exchangeability of lineages each quartet has probability 1/3.
\end{enumerate}

For $S\in \{A,B,C,D\}$ and a gene quartet $xy|zw$ where $x,y,z,w\in X$ are taxa on $N^+$, 
we define an event, denoted $\mathcal C_S \to xy|zw,$ 
that captures whether the behavior of gene lineages in $S$ ensures that under the coalescent model 
the gene tree $xy|zw$ is formed, or will be, with a determined probability. This may be due to a coalescent 
event occurring in $S$, or 3 or 4 lineages reaching a common node in $S$ without having yet coalesced. 
Since a coalescent event between the lineages occuring before they 
enter $S$ would already determine the quartet tree, we define this event conditional on  
lineages from $x,y,z,w$ entering $S$ distinctly.
More formally, consider the events: 
\begin{align*}
E=E(S,xy|zw)&= \text{No coalescence between any of the lineages $x,y,z,w$ }\\ 
&\hskip.25in \text{that may enter $S$ occurs before they enter $S$.}\\
F=F(S,xy|zw)&=\text{3 or 4 of the $x,y,z,w$ lineages reach a common node in $S$ }\\
&\hskip.25in \text{before any coalescence, and afterwards $xy|zw$ forms.}\\
G=G(S,xy|zw)&= \text{A first coalescence occurs in $S$ between $x,y$ or between }\\
&\hskip.25in \text{$z,w$, without 3 or 4 lineages having reached a }\\
&\hskip.25in\text{common node previously}
\end{align*}
Then $\mathcal C_S \to xy|zw$ denotes $(F\cup G)|E$.

Let  $P(\mathcal C_S \to xy|zw)$ denote the conditional probability of the event $\mathcal C_S \to xy|zw$.
Then with $a_i,b_i,c_i$ distinct taxa from $A,B,C$, respectively, a few example decompositions of $CF$s  are:
\begin{align*}
CF_{a_1a_2|a_3 b_1}&=P(\mathcal C_{A}\to a_1a_2|a_3b_1),\\
CF_{a_1a_2|b_1 c_1}&=P(\mathcal C_A\to a_1a_2|b_1 c_1) + (1- P(\mathcal C_A\to a_1a_2|b_1 c_1)  )P(\mathcal C_{D}\to a_1a_2|b_1 c_1),\\
CF_{a_1a_2|c_1 c_2}&=P(\mathcal C_A\to a_1a_2|c_1 c_2) + (1-P(\mathcal C_A\to a_1a_2|c_1 c_2) )P(\mathcal C_{D}\to a_1a_2|c_1 c_2) \\
&\hskip -.1in +  (1-P(\mathcal C_A\to a_1a_2|c_1 c_2) )(1-P(\mathcal C_D\to a_1a_2|c_1 c_2) )P(\mathcal C_C\to a_1a_2|c_1 c_2).
\end{align*}

For calculating probabilities associated to $D$, we  suppress indices on taxa. This is allowable since, conditioned on the lineages entering $D$ distinctly, those from $A$ are exchangeable, as are those from $B$.
Thus, for instance, 
$$P(\mathcal C_D\to {ab|bc})=P(\mathcal C_D\to {a_1b_1|b_2c_1})=P(\mathcal C_D\to {a_2b_2|b_1c_2}).$$ 

Significantly, all $CF$s for $N^+$
can be computed using only the following probabilities associated to $D$ together with expressions dependent only on $A,B,C$:

\begin{align*}
p_1=P(\mathcal C_D\to ab|cc) &=1-\ell_3( 1-\gamma +\gamma x), \\
p_2=P(\mathcal C_D\to aa|cc) &= 1-\ell_1\ell_3(\gamma^2h_1x+2\gamma(1-\gamma) +(1-\gamma)^2h_2),\\
p_3= P(\mathcal C_D\to bb|cc) &=1-x\ell_2\ell_3,\\
p_4=P(\mathcal C_{D}\to ab|bc) &=\ell_2\left ( \gamma + (1 -  \gamma) x  \right ) /3\\
  P(\mathcal C_{D}\to bb|ac) &=1-2p_4,\\
p_5=P(\mathcal C_{D}\to ab|ac) &= \ell_1(\gamma^2h_1 + \gamma(1-\gamma)(3-x) + (1-\gamma)^2 h_2  )/3,\\
  P(\mathcal C_{D}\to aa|bc) &=1-2 p_5,\\
p_6=P(\mathcal C_{D}\to ab|ab) &= \ell_1\ell_2(\gamma^2h_1  + 2\gamma(1-\gamma)  + (1-\gamma)^2 h_2x)/3,\\
  P(\mathcal C_{D}\to aa|bb) &=1-2 p_6.
\end{align*}

The 6 linear independent polynomials, $p_1,p_2,\dots, p_6$ parameterize a variety, $\mathcal V(D)$. 
Combined with the previous discussion of decomposing $CF$ formulas, this yields the following.

\begin{proposition}\label{prop:mapfactor}
Let $\mathcal V_{ N}$ be the $CF$ variety for a semidirected network (not necessarily level-1) 
$N$ with the form shown in \cref{fig:net3cycles}, and  numerical parameter space 
$\Theta(N)=\Theta_{A,B,C}\times \Theta_D$. Let $\mathcal V_D$ denote the Zariski closure 
of the image of the  parameterization $\phi:\mathbb C^7\to\mathbb C^6$, defined by 
$$\phi(\gamma,\ell_1,\ell_2,\ell_3,h_1,h_2,x)=(p_1,p_2,p_3,p_4,p_5,p_6),$$ 
with the $p_i$ given above. 
Then the map $CF: \Theta(N)\to\mathbb C^{3{n\choose 4}} $ factors as 
\begin{equation}
CF: \Theta( N) =\Theta_{A,B,C}\times \Theta_D \xrightarrow{\pi\times \phi}\Theta_{A,B,C}\times \mathcal V_D\to \mathcal V_N\subset\mathbb C^{3{n\choose 4}}. \label{eq:factor}
\end{equation}
where $\pi$ is the map projecting $\Theta(N)$ onto the numerical parameters on $A,B,C$ only.
\end{proposition}

\cref{prop:factorVar} shows that $\mathcal V_D=\mathbb C^6$, and thus $\phi$ is an 
infinite-to-1 map, establishing the following.

\begin{corollary} Consider a semidirected topological network $N$ with a 3-cycle, with decomposition 
as in \cref{fig:net3cycles} (L). Then no test using polynomial equalities in quartet $CF$s can  identify the 
hybrid node in the 3-cycle. 

Specifically,
if N's root must be in the subnetwork $C$ because of the semidirected topology of $C$,
then the network $N_B$ which has $A,B$ interchanged from $N=N_A$, so that $B$ is below the 3-cycle's hybrid node, 
leads to the same ideal of invariants, that is, $\mathcal I(N_A)=\mathcal I(N_B)$. 
If the semidirected topology of $N$ allows 
for rooting in either subnetwork $B$ or $C$, then 
$\mathcal I(N_A)=\mathcal I(N_B)= \mathcal I(N_C)$.
\end{corollary}

\begin{proof} If deleting the 3-cycle from the network induces a $(n_1,n_2,n_3)$  partition of the taxa with all $n_i\ge2$, 
	then the corollary follows directly from \cref{prop:mapfactor,prop:factorVar}.
Cases with $n_i=1$ then follow by deleting taxa from an appropriate network with all $n_i\ge 2$, intersecting the ideals
with a ring generated by fewer CFs.
\end{proof}

Note that this corollary applies to networks with more than one 3-cycle.
However, when multiple cycles are present, the location of
one cycle's hybrid node indicates that one of the nodes in a descendant cycle cannot
be hybrid.  Thus for  a network with $k$ 3-cycles, there are between $2^k$ and $3^k$ 
networks differing only in the choice of hybrid nodes in the 3-cycles,
all of which are algebraically indistinguishable using $CF$s.

Nonetheless, using semialgebraic tests, we can obtain additional information on hybrid node location, as the following generalization of \cref{prop:222semi} shows.

\begin{proposition} \label{prop:3cycle} Consider a partition of a taxon set $X$ into three blocks of size 
at least 2. For any network $N$ (not necessarily level-1) with a node or 3-cycle inducing these blocks, 
denote the node or 3-cycle and its 
adjacent edges by $D$, and the 
subgraphs attached to $D$ as $A,B,C$ (as in \cref{fig:decomp} for a cycle).

Let $G_{abc},\ G_{cab},\ G_{bca}$ be as defined by equation \eqref{eq:Gs}, for any distinct taxa $a_i$ on $A$, $b_i$ on $B$, and $c_i$ on $C$.
		Then  $D$ is:
		\begin{equation*}
			\begin{cases}
			\text{a 3-leaf tree} & \text{ if  }G_{abc}=G_{bca}=G_{cab}=0,\\
			\text{a 3-cycle and adjacent edges}
			  & \text{ if  } G_{xyz}>0,\ G_{yzx}\le 0,\ G_{zxy}\le 0 \\
			\text{\ \ with $x$ below the hybrid node} &\text{\ \ \  \ \ for $\{x,y,z\}=\{a,b,c \}$.}\\
			\text{a 3-cycle and adjacent edges with}
			  & \text{ if  } G_{xyz}>0,\ G_{yzx}> 0,\ G_{zxy}< 0 \\
			\text{\ \ $x$ or $y$ below the hybrid node} &\text{\ \ \  \ \ for $\{x,y,z\}=\{a,b,c \}$.}
 			\end{cases}
			\end{equation*}		
Moreover, for a network $N$ with a 3-cycle $D$ and descendants of  $x$ forming its hybrid block, 
there exist positive measure subsets of parameters on $D$ on which $G_{yzx}$ and $G_{zxy}$ satisfy both of the above sign conditions.	
\end{proposition}

\begin{proof}
First suppose $D$ is a 3-cycle and, without loss of generality, $A$ is below the hybrid node.
Then we decompose formulas for $CF$s  for $N$ as
\begin{align*}
       CF_{ab|ab} &=(1-P(\mathcal C_A\to aa|bb))   (1-P(\mathcal C_B\to aa|bb)) P(\mathcal C_D\to ab|ab), \\ 
       CF_{ac|bc} &= (1-P(\mathcal C_D\to ab|cc))P(\mathcal C_C\to ac|bc),\\
       CF_{ac|ac} &=(1-P(\mathcal C_A \to aa|cc)) (1-P(\mathcal C_D\to aa|cc))P(\mathcal C_C\to ac|ac),\\
      CF_{ab|bc} &=(1-P(\mathcal C_B \to ac|bb)) P(\mathcal C_D \to ab|bc),\\
      CF_{bc|bc}&=(1-P(\mathcal C_B \to bb|cc))   (1-P(\mathcal C_D\to bb|cc))P(\mathcal C_C\to bc|bc),\\
       CF_{ab|ac}&= (1-P(\mathcal C_A \to aa|bc))  P(\mathcal C_D \to ab|ac).
\end{align*}    
Since 
   \begin{align*}
   P(\mathcal C_A\to aa|bb)&=P(\mathcal C_A \to aa|cc)=P(\mathcal C_A \to aa|bc),\\
   P(\mathcal C_B\to aa|bb)&=P(\mathcal C_B \to bb|cc)=P(\mathcal C_B \to ac|bb),\\
   P(\mathcal C_C\to ac|bc)&=P(\mathcal C_C\to ac|ac)=P(\mathcal C_C\to bc|bc),
   \end{align*}
it follows that 
	\begin{multline*}
	G_{abc}(N)= 
	(1-P(\mathcal C_A\to aa|bb) )(1-P(\mathcal C_B\to aa|bb)) P(\mathcal C_C\to ac|ac)\: \times\\
	\big [ (1-P(\mathcal C_D\to aa|cc) ) P(\mathcal C_D\to ab|bc) 
	 -2 (1-P(\mathcal C_D\to bb|cc)) P(\mathcal C_D\to ab|ac)   \\
	 + P(\mathcal C_D\to ab|ab)(1-P(\mathcal C_D\to ab|cc)  \big ].
\end{multline*}
But the terms in the last factor, all of which depend only on $D$, arise as multiples of $CF$s  on the network $N_{6-3_2}$ of  \cref{fig:net3cycles}
(R),
    \begin{align*}
    1- P(\mathcal C_D\to aa|cc) &= 3CF_{ac|ac}(N_{6-3_2}), &
    P(\mathcal C_D\to ab|bc) &=CF_{ab|bc}(N_{6-3_2}),\\
    1- P(\mathcal C_D\to bb|cc) &=3CF_{bc|bc}(N_{6-3_2}), &
    P(\mathcal C_D\to ab|ac)&=CF_{ab|ac}(N_{6-3_2}),\\
    P(\mathcal C_D\to ab|ab) &=CF_{ab|ab}(N_{6-3_2}), &
    1-P(\mathcal C_D\to ab|cc) &=3CF_{ac|bc}(N_{6-3_2}).
    \end{align*}
Thus, by \cref{prop:222semi}, $G_{abc}(N)=$
$$
3(1-P(\mathcal C_A\to aa|bb) )(1-P(\mathcal C_B\to aa|bb)) P(\mathcal C_C\to ac|ac)G_{abc}(N_{6-3_2})>0.
$$

Since $G_{abc}+G_{cab} +G_{bca}=0$, either 1 or 2 of these terms are positive, and the two cases  for 3-cycle $D$s
follow. The case of the network $N$ with $D$ a node
is obtained by setting $x=h_1=h_2=0$ in the formulas 
for any of the 3-cycle networks, showing, for instance, that $G_{abc}(N)$ is a multiple of $G_{abc}(T_6)=0.$

The final statement on positive measure subsets of parameter space  follows from \cref{prop:222semi}.
\end{proof}

\smallskip

\cref{prop:3cycle} yields the following generalization of \cref{thm:222twosets}.

\begin{theorem}\label{thm:3cyc} Consider a partition of a taxon set $X$ into three blocks of size at least 2. Then for all networks (not necessarily level-1) with a node or 3-cycle inducing these blocks, the presence of the node or the (undirected) 3-cycle along with one non-hybrid block is identifiable. If the network has a 3-cycle then there are positive measure subsets of its parameter space on which the hybrid node can be determined, and on which it cannot.
\end{theorem}
\begin{proof} By \cref{prop:3cycle}, an undirected 3-cycle is signaled by the non-vanishing of at least one of $G_{abc}$, $G_{bca}$ or $G_{cab}$, and
for a 3-cycle, a non-hybrid block is identifiable since one of the $G$s must be negative. That the hybrid node can be identified on a positive measure set follows from the existence of such a set for which only one $G$ is positive. That the hybrid node cannot be identified on another set is seen by choosing specific parameters on 3-cycles with different hybrid nodes (e.g., using parameters given in \cref{ex:NonID} for the 3-cycle and adjacent edge parameters) which produce the same values for the 
$p_i$.
\end{proof}

If $n_i=1$  for some $i$, then similar arguments as given for \cref{prop:3cycle} and \cref{thm:3cyc} shows the function $f_{xyz}$  can identify the presence of a 3-cycle, but possibly not its hybrid node. While we omit the proof, we state the result.

\begin{proposition} Consider a partition of a taxon set $X$ into three blocks of size $1, n_1,n_2$ with $n_i\ge 2$. For any network $N$ (not necessarily level-1) with a node or 3-cycle inducing these blocks, let $D$ denote the node or 3-cycle and adjacent edges, and $A,B,C$ the subgraphs attached to $D$ by the adjacent edges, with $C$ being a single node.
Let
$f_{abc}$ be as in \cref{prop:5tax3cyc}, for any distinct $a_i$ on $A$, $b_i$ on $B$, and $c$ on $C$.
Then
for generic numerical parameters, $D$ is:			 
			\begin{equation*}
			\begin{cases}
			\text{a 3-leaf tree} & \text{ if, and only if, }f_{abc}=0,\\
			\text{a 3-cycle and adjacent edges with}\\
			\text{\ \ $A$ or $B$ below the hybrid node}&\text{ if  }f_{abc}<0,\\
 			\text{a 3-cycle and adjacent edges with} \\
			\text{\ \ $A$, $B$, or $C$ below the hybrid node}&\text{ if  }f_{abc}>0.
 			\end{cases}
 		\end{equation*}
		
  Finally, there are positive measure subsets of the numerical parameter space for  the networks with a 3-cycle $D$ and either $A$ or $B$ below its hybrid node on which  $f_{abc}<0$ and on which  $f_{abc}>0$.
\end{proposition}

For identifying the hybrid node in a $3$-cycle inducing a $(1,n_1,n_2)$ partition when there is a single descendant of the hybrid node, we generalize \cref{thm:221twosets}.

\begin{theorem}
Consider a partition of a taxon set $X$ into three blocks of sizes $1,n_1,n_2$ with $n_i\ge 2$. Then for all networks (not necessarily level-1) with a  3-cycle inducing these blocks, there are positive measure subsets of the parameter space on which the hybrid node of the 3-cycle is identifiable, and on which it is not.
\end{theorem}
\begin{proof} Let $\tilde f$ be as defined in \cref{eq:ftilde}. We first show the result for
a general network $N$ with a 3-cycle with a single hybrid descendant. For such a network, using decompositions as in \cref{fig:decomp} but with $C$ a hybrid singleton taxon and the root in 
$B$, we find that for any choices of two taxa in the $A$ and $B$ blocks
$$\tilde f(N)=3(1-P(C_A\to aa|bc))P(C_B\to ab|bc)\tilde f(N_{5-3_1}),$$
where $N_{5-3_1}$ is given parameters from the 3-cycle and adjacent edges of $N$. Similarly,
if a non-hybrid block $C$ is the singleton
$$\tilde f(N)=3(1-P(C_A\to aa|bc))P(C_B\to ab|bc)\tilde f(N_{5-3_2}),$$
where $N_{5-3_2}$ is given parameters from the 3-cycle and adjacent edges of $N$. Thus the signs of $\tilde f$ on $N$ can be used as in the proof of \cref{thm:221twosets} to obtain the claim when the hybrid block is a singleton.

If the singleton block is not hybrid on $N$ the claim is established as for  \cref{thm:221twosets}, by passing to a subnetwork with a $3_2$-cycle and using the parameters of \cref{ex:NonID}.
\end{proof}

Finally, if a 3-cycle induces a $(1,1,n-2)$ partition then  \cref{prop:3-cyc11} applies directly to analyze identifiability.

\subsection{4-cycles}\label{ssec:4top}

\begin{figure}
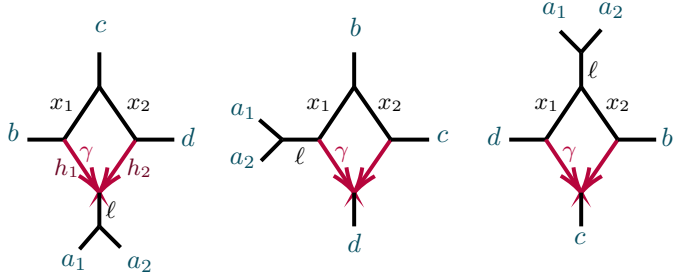

	\begin{center}
	\include{Figures/4cycles_SWN}
	\end{center} 
	\caption{The semidirected 5-taxon binary networks with a single 4-cycle,  up to taxon labelling. We denote these by $N_s$, $N_w$, $N_n$ from left to right, according to compass directions for the $a_1,a_2$ cherry when the hybrid node is  located at south. Note that $N_e$ is omitted since, up to taxon labelling, it is the same as $N_w$. Edge probabilities and the hybridization parameter $\gamma$ are shown next to edges.}
	\label{fig:4cyc}
\end{figure}
	
To study topological 4-cycle identifiability beyond the results of  \cite{Solis-Lemus2016} and \cite{Banos2019}, we consider first the networks $N_s$, $N_w$, $N_n$ on 5-taxa
of \cref{fig:4cyc}. Note that  hybrid edge probabilities  are not labeled for the networks $N_w$ and $N_n$,  since no coalescence can occur in those edges as they have only one descendant taxon.  These three networks have equal subideals of linear invariants, as described in \Cref{app:props}, $$\mathcal J(N_s)=\mathcal J(N_w)=\mathcal J(N_n),$$ since the networks have the same undirected topology. 
\Cref{prop:S,prop:W,prop:N}, with additional computation, give an identifiability result for the hybrid nodes of 4-cycles in any level-1 network with at least 5 taxa.

\begin{figure}
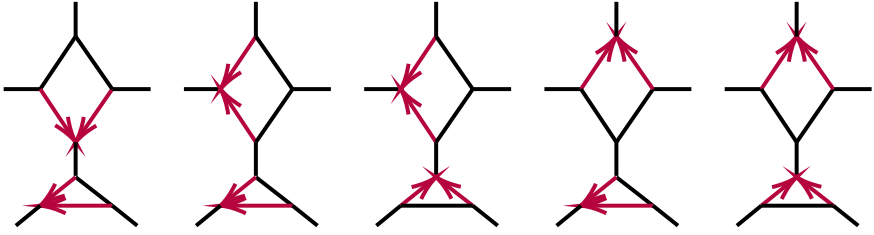

\begin{center}
\include{Figures/4-3-cycles}
\end{center}
\caption{The semidirected 5-taxon level-1 binary networks with a single 4-cycle and 3-cycle, up to taxon labelling.  }
\label{fig:4and3cyc}
\end{figure}

\begin{proposition}\label{prop:4hybrid} Consider a semidirected binary level-1 network on $n\ge5$ taxa whose topology is known up to contracting 2- and 3-cycles and the direction of  hybrid edges in 4-cycles. Then for generic numerical parameter values on the network, the 4-cycle hybrid edge directions are identifiable from  $CF$s.
\end{proposition}

\begin{proof}
Suppose first a network $N$ has exactly 5 taxa, and a 4-cycle. Then after contracting 2-cycles $N$ yields  $N_s$, $N_w$, $N_n$, or one of the five networks shown in \cref{fig:4and3cyc}. Although we do not know whether $N$ has a 3-cycle, if it does then by \cref{prop:3-cyc11}
 it has the same associated variety as the network with that 3-cycle contracted, so we investigate the relationships of $\mathcal V(N_s)$, $\mathcal V(N_w),$ and  $\mathcal V(N_n)$.
 
\cref{prop:S,prop:W,prop:N} show that
$\mathcal V(N_s)$, $\mathcal V(N_w),$ and  $\mathcal V(N_n)$ have dimensions 5, 4, and 3, respectively. 
Moreover, $\mathcal V(N_s)$ contains both $\mathcal V(N_w)$ and  $\mathcal V(N_n)$. 
Additional computations show
$$\mathcal V(N_w) \cap \mathcal V(N_n)= \mathcal{V}_1 \cup \mathcal{V}_2 \cup\mathcal{V}_3$$
 has dimension 2, with three irreducible components whose ideals are
\begin{align*}
	&\mathcal{I}(\mathcal{V}_1) = \langle CF_{ab|cd} - CF_{ac|bd}, CF_{ab|ad} - CF_{ac|ad}, 3CF_{ac|ad}CF_{ac|bd} - CF_{ab|ac}  \rangle + \mathcal{I}(N_s) \\ 
	&\mathcal{I}(\mathcal{V}_2) = \langle CF_{ab|cd} + 2CF_{ac|bd} - 1, CF_{ab|ad} - CF_{ab|ac}, 3CF_{ab|ac}CF_{ac|bd} - CF_{ac|ad} \rangle\\ 
	&\hskip 4.5in+ \mathcal{I}(N_s) \\ 
	&\mathcal{I}(\mathcal{V}_3) = \langle CF_{ac|ad}, CF_{ab|ac}, CF_{ab|ad} \rangle + \mathcal{I}(N_s) 
\end{align*}

Thus generic parameters for  $N_s$ give points on neither $\mathcal V(N_w)$ nor $\mathcal V(N_n)$, 
while generic parameters for $N_w$ give points not on $\mathcal V(N_n)$, and generic parameters for $N_n$ give points not on $\mathcal V(N_w)$. Thus for generic parameters, the hybrid node in the 4-cycle can be determined by testing invariants to see whether the $CF$s lie on $\mathcal V(N_n)$ or $\mathcal V(N_w)$, or neither. 

If $N$ has more than 5 taxa, choose one taxon from each of 3 of the taxon blocks determined by a 4-cycle, and 2 from the remaining block, and pass to the induced network on these 5 taxa to apply the result for 5-taxon networks.
\end{proof}

\begin{rmk}
The components $\mathcal V_1, \mathcal V_2,$ and $\mathcal V_3$ of $\mathcal V(N_w) \cap \mathcal V(N_n)$ arise naturally from the parameterizations.
Restricting to $\gamma=1$ on $N_w$ and $\gamma=0$ on $N_n$, essentially giving the unrooted tree $((a_1,a_2),(b,c),d)$ for both, yields
$\mathcal V_1$. $\mathcal V_2$ arises from $\gamma=0$ on $N_w$ and $\gamma=1$ on $N_n$ which gives the unrooted tree $((a_1,a_2),b,(c,d))$.
$\mathcal V_3$ arises from $\ell=0$ on both $N_w$ and $N_n$,  which by corresponding to an infinite edge length, ensures $a_1,a_2$ form a cherry in any gene tree involving those two taxa, and for those involving only one $a_i$, gives $CF$s from a 4-cycle with unidentifiable hybrid node.
\end{rmk}
		
\subsection{Summary of topological identifiability} 

The results of this section combined with \cref{thm:topExcept4}  yield the following theorem.

\begin{theorem}[Topological Identifiability from quartet  $CF$s]\label{thm:mainTop}
Let  $N^+$ be a binary level-1 phylogenetic network on $n\ge 5$ taxa, with generic numerical parameters. Then no 2-cycle on the semidirected network can be identified, so let $\widetilde N$ be the topological semidirected network induced by $N^+$ with all 2-cycles
replaced with edges. Then the topological structure of $\widetilde N$, including directions of hybrid edges, is identifiable from quartet $CF$s of $N^+$,   with the following exceptions:
 \begin{enumerate}
\item If a 3-cycle induces a $(1,1,n-2)$ partition of taxa, then if the hybrid node has a single descendant taxon the network cannot be distinguished from the network in which the cycle is contracted to a node, or from the network in which the hybrid and other singleton block are interchanged. If the hybrid node has $n-2$ descendant taxa, then there are positive-measure subsets of parameters on which the semidirected 3-cycle is and is not identifiable.

\item If a 3-cycle induces a $(1,n_1,n_2)$ partition with $n_1,n_2\ge 2$ then the undirected 3-cycle can be identified. 
There are  positive measure subsets of parameters on which  the semidirected 3-cycle is and is not identifiable.

\item\label{case:nice3} If a 3-cycle induces an $(n_1,n_2,n_3)$ partition with all $n_i\ge 2$, then the undirected 3-cycle can be identified, and at least 1 of the 3-cycle nodes can be determined not to be hybrid, but there are positive measure subsets of parameters on which the semidirected 3-cycle is and is not identifiable.\end{enumerate}
\end{theorem}

\section{Identifiability of numerical parameters}\label{sec:numID}

To address identifiability of numerical parameters --- both edge lengths and hybridization parameters --- we assume the network has no 2-cycles, as these are not identifiable. For the remainder of the section we thus study $\widetilde N$, the  semidirected metric binary phylogenetic network induced from a rooted  network $N^+$, with  2-cycles replaced by edges. 
In showing an edge in $\widetilde N$
has identifiable length, we are showing that if the original network did have a 2-cycle, then an ``effective" length of an edge resulting
from replacing  the cycle  as in \cref{lem:2cyc}  is identifiable.

Since we assume exactly one sample per taxon for each gene, no coalescent event can occur in pendant edges.  Thus no
pendant edge length appears in $CF$ parameterizations,  and such lengths cannot be identified from  $CF$s, yielding the following.

\begin{proposition}\label{prop:pendantNonID} Let $N$ be a semidirected phylogenetic network. Then pendant edge lengths are not identifiable from quartet $CF$s under the NMSC model with one sample per taxon.
\end{proposition}

\subsection{Lengths of edges defined by 4 taxa}

For some edges in $\widetilde N$ it is simple to identify the edge length. We first focus on one type of such edges. 

\begin{defi} Let $e$ be an edge in $\widetilde N$. Then we say $e$ is \emph{defined} by a set $Q=\{a,b,c,d\}$ of 4 taxa if:
\begin{enumerate}
\item Edge $e$ lies in the subnetwork $\widetilde N(Q)$ of $\widetilde N$ composed of all edges and nodes which form the induced $\widetilde N|_Q$ once degree 2 nodes are suppressed,
\item Edge $e$ is a cut edge of  $\widetilde N(Q)$ separating pairs of
taxa, say $a,b$ from $c,d$, and
\item In $\widetilde N(Q)$ there are 4 cut edges adjacent to $e$, separating each of $a,b,c,d$, respectively, from the others.  \end{enumerate}
\end{defi}

In an unrooted  tree, every internal edge is defined by some $Q$, even if the tree is not binary. 
But for a network, even if binary and level-1, as in  \cref{fig:Qdef}, this is not the case: A $k$-cycle, with $k\ge 5,$ has $k-4$ edges in it that are defined by such sets, with the hybrid edges and those adjacent to them exceptions, as will be proved in the next proposition. Edges descended from hybrid nodes are also never defined by a set $Q$.
These examples show edges defined by a set $Q$ need not be cut edges, and not all cut edges are defined by a set $Q$.

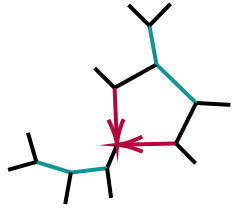
\begin{figure}
\begin{center}
\tikzset{every picture/.style={line width=0.75pt}} 

\begin{tikzpicture}[x=0.75pt,y=0.75pt,yscale=-0.9,xscale=0.9]

\draw [color={rgb, 255:red, 0; green, 0; blue, 0 }  ,draw opacity=1 ][line width=1.5]    (216.31,64.52) -- (235.5,65.57) ;
\draw [color={rgb, 255:red, 14; green, 149; blue, 148 }  ,draw opacity=1 ][line width=1.5]    (194.06,42.93) -- (216.31,64.52) ;
\draw [color={rgb, 255:red, 0; green, 0; blue, 0 }  ,draw opacity=1 ][line width=1.5]    (194.06,42.93) -- (170.58,56.02) ;
\draw [color={rgb, 255:red, 182; green, 5; blue, 60 }  ,draw opacity=1 ][line width=1.5]    (204.97,87.22) -- (174.97,87.9) ;
\draw [shift={(171.97,87.97)}, rotate = 358.68] [color={rgb, 255:red, 182; green, 5; blue, 60 }  ,draw opacity=1 ][line width=1.5]    (14.21,-4.28) .. controls (9.04,-1.82) and (4.3,-0.39) .. (0,0) .. controls (4.3,0.39) and (9.04,1.82) .. (14.21,4.28)   ;
\draw [color={rgb, 255:red, 14; green, 149; blue, 148 }  ,draw opacity=1 ][line width=1.5]    (190.1,21.09) -- (194.06,42.93) ;
\draw [color={rgb, 255:red, 0; green, 0; blue, 0 }  ,draw opacity=1 ][line width=1.5]    (171.97,87.97) -- (166.41,101.06) ;
\draw [color={rgb, 255:red, 14; green, 149; blue, 148 }  ,draw opacity=1 ][line width=1.5]    (166.41,101.06) -- (146.07,103.52) ;
\draw [line width=1.5]    (216.31,64.52) -- (204.97,87.22) ;
\draw [color={rgb, 255:red, 0; green, 0; blue, 0 }  ,draw opacity=1 ][line width=1.5]    (201.95,8.9) -- (190.1,21.09) ;
\draw [color={rgb, 255:red, 14; green, 149; blue, 148 }  ,draw opacity=1 ][line width=1.5]    (126.75,97.57) -- (146.07,103.52) ;
\draw [color={rgb, 255:red, 182; green, 5; blue, 60 }  ,draw opacity=1 ][line width=1.5]    (170.58,56.02) -- (171.84,84.98) ;
\draw [shift={(171.97,87.97)}, rotate = 267.51] [color={rgb, 255:red, 182; green, 5; blue, 60 }  ,draw opacity=1 ][line width=1.5]    (14.21,-4.28) .. controls (9.04,-1.82) and (4.3,-0.39) .. (0,0) .. controls (4.3,0.39) and (9.04,1.82) .. (14.21,4.28)   ;
\draw [color={rgb, 255:red, 0; green, 0; blue, 0 }  ,draw opacity=1 ][line width=1.5]    (190.1,21.09) -- (178.41,9.52) ;
\draw [color={rgb, 255:red, 0; green, 0; blue, 0 }  ,draw opacity=1 ][line width=1.5]    (159.38,44.92) -- (170.58,56.02) ;
\draw [color={rgb, 255:red, 0; green, 0; blue, 0 }  ,draw opacity=1 ][line width=1.5]    (146.07,103.52) -- (142.93,121.19) ;
\draw [color={rgb, 255:red, 0; green, 0; blue, 0 }  ,draw opacity=1 ][line width=1.5]    (215.99,98.36) -- (204.97,87.22) ;
\draw [color={rgb, 255:red, 0; green, 0; blue, 0 }  ,draw opacity=1 ][line width=1.5]    (122.03,81.24) -- (126.75,97.57) ;
\draw [color={rgb, 255:red, 0; green, 0; blue, 0 }  ,draw opacity=1 ][line width=1.5]    (126.75,97.57) -- (110.9,102) ;
\draw [color={rgb, 255:red, 0; green, 0; blue, 0 }  ,draw opacity=1 ][line width=1.5]    (166.41,101.06) -- (168.63,117.71) ;

\end{tikzpicture}
\vspace*{-1cm}
\end{center}
\caption{A semidirected network with edges defined by sets $Q$ of 4 taxa highlighted in blue.}\label{fig:Qdef}
\end{figure}

For a binary network, an alternate characterization of edges defined by sets $Q$ can be given.

\begin{proposition}
For a binary level-1 semidirected network $\widetilde N$, there is a set $Q$ of 4 taxa  defining an edge $e$ if, 
and only if, $e$ is an internal edge that is neither hybrid nor adjacent to a hybrid edge.
\end{proposition}
\begin{proof}
Suppose $e$  is defined by $Q$. If $e$ were either hybrid or adjacent to a hybrid edge, then $Q$ would contain a descendant of a hybrid node. But then $\widetilde N(Q)$ contains all edges of the cycle in which the hybrid edge lies.  This contradicts that both $e$ and its adjacent edges are cut edges in $\widetilde N(Q)$, since the hybrid edges are not cut. 

Conversely, suppose $e$ is neither hybrid nor adjacent to a hybrid edge.  
If none of these 5 edges is in a cycle in $\widetilde N$, then choosing one taxon in each component obtained by deleting $e$ and its incident nodes and adjacent edges gives a set $Q$ defining $e$. 

If any one of these edges is in a cycle, then  since $\widetilde N$ is level-1 and binary, exactly one of the following holds: a) $e$ is in a cycle, together with exactly 2 adjacent edges, one at each endpoint of $e$, b)
$e$ is not in a cycle, but
exactly one cycle contains  two edges adjacent to $e$ at the same endpoint of $e$, or c) $e$ is not in a cycle, but all 4 edges adjacent to $e$ are, with $e$ adjacent to
two different cycles.

For case (a), the 2 edges adjacent to $e$ that are not in the cycle must be cut edges, and the two adjacent to $e$ that are in the cycle
must be adjacent to 2 other distinct cut edges not in the cycle. Choosing taxa from the non-$e$ components left by deleting these 4 cut edges gives a set $Q$ defining $e$.

 In case (b), The two edges in the cycle must be adjacent to distinct cut edges other than $e$ which are not in the cycle. Choosing taxa from the non-$e$  components of the graph obtained by deleting these two edges and  the two non-cycle edges adjacent to $e$ gives a quartet defining $e$. Case (c) is similar, treating each cycle the same way. 
\end{proof}

 For any network, regardless of level or other special structure, lengths of edges defined by sets $Q$ are easily identified.

\begin{proposition}\label{prop:defQid} If an edge $e$ in a metric network $\widetilde N$ is defined by a set $Q$ of 4 taxa, then its length is identifiable from quartet $CF$s.
\end{proposition}
\begin{proof}
If $e$ is defined by $Q=\{a,b,c,d\}$ has length $t$ and in $\widetilde N|_Q$ induces the split $ab|cd$, then 
$\overline{CF}_{ac|bd}= \exp (-t)/3,$ so $t=-\log(3 CF_{ac|bd})$.
\end{proof}

\subsection{Numerical parameters associated to 3-cycles}

Edges either in or adjacent to a 3-cycle are always adjacent to a hybrid edge. Thus in binary networks, these edges are not defined by sets of 4 taxa, so
\cref{prop:defQid} does not apply.  \cref{prop:N5-3-1}(c), \cref{prop:N6-2-2-2}(c) and  \cref{prop:N5-3-2}(c)
 illustrate that, at least for specific small networks, the numerical parameters associated to 3-cycles are not identifiable. 
 More generally, we obtain the following.
 
 \begin{proposition}\label{prop:3cycNonID}
If $C$ is a 3-cycle on a semidirected binary  level-1 network $\widetilde N$, then neither the hybridization parameters nor the lengths of any edges in or adjacent to $C$ can be identified from quartet  $CF$s.
\end{proposition}
 
 \begin{proof}
Suppose first the 3-cycle induces an $(n_1,n_2,n_3)$-partition of the taxa with all $n_i\ge 2$. Then using  \cref{prop:mapfactor} and \ref{prop:factorVar}(a) we  see that the map from numerical parameters to $CF$s factors by sending the 7 numerical parameters associated to the 3-cycle and its adjacent edges 
into a 6-dimensional variety. This implies that the numerical parameters cannot all be identifiable. 
To see that no single parameter can be identified, first observe that from the factorization of maps in \cref{eq:factor}, if a single parameter were identifiable, it would have to be identifiable from a point in $\mathcal V_D$. However, \cref{prop:factorVar}(b) shows that is not the case.

If a 3-cycle induces a $(1,n_2,n_3)$- or $(1,1,n_3)$-partition of taxa, then  by considering samples of 2 individuals for each gene from the singleton taxa, we can modify the network by attaching cherries of pseudotaxa for each singleton. Since in this case we already know that numerical parameters around the 3-cycle are not identifiable from all $CF$s,   with access only to $CF$s using only one of the pseudotaxa, they are still not identifiable. But that means they are not identifiable for the original network. 
\end{proof}

\subsection{Other numerical parameters}
The remaining numerical parameters on a binary level-1 network to be considered include lengths 
of hybrid edges, lengths of edges adjacent to hybrid edges, and hybridization parameters, all 
when the relevant cycle is of size $\ge 4$. 

\begin{proposition} \label{prop:LargeCyc} Let $\widetilde N$ be a level-1 metric binary semidirected network with no 2-cycles, containing a $k$-cycle $C$ with $k\ge 5$. Then  hybridization parameters and lengths of the cycle edges adjacent to the hybrid edges in $C$ can be identified from quartet $CF$s. If the hybrid node of $C$ has at least 2 descendant taxa, the lengths of the hybrid edges can also be identified. If the hybrid node has only one descendant taxon then the lengths of the hybrid edges are not identifiable. 
\end{proposition}

\begin{proof}  From \cref{prop:defQid} we already know that the $k-4$ edges in the cycle that are not hybrid or adjacent to a hybrid edge have identifiable lengths. If the taxon blocks for the cycle are, proceeding from the hybrid around the cycle, $X_1,X_2,\dots, X_k$, then pick one taxon from
each of $X_1,X_2,X_3,X_4,$ and $X_k$  and pass to the induced subnetwork. Replacing any 2-cycles with edges, we may assume we have a 5-cycle sunlet network as in \cref{fig:5cyc}(L), in which the edge probability $y$ of the edge opposite the hybrid node is known, and the edge probability $x$ is that of the edge in $C$ which is adjacent to a hybrid edge, lying between blocks $X_2$ and $X_3$.

Using $y$  and $CF$s we can identify $\gamma$, and then $x$ through
$$CF_{ac|de}-CF_{ad|ce}=  \gamma\left(1- y\right),\ \ \ \ CF_{ab|cd}-CF_{ac|bd}=  \gamma\left(1- x\right).$$
Similarly, the other edge in $C$ adjacent to a hybrid edge has identifiable length.

\smallskip

If the hybrid node has 2 descendant taxa, then by picking two taxa from $X_1$ and one from each of $X_2,X_3,X_k$ we pass to an induced subnetwork which, after replacing 2-cycles by edges, has the form of the network of  \cref{fig:5cyc}(C) or (R) with the same hybrid edge lengths as the full network. 
 In case (C), a cherry below the hybrid node,  applying the result of \cref{prop:S} on $N_S$  identifies the hybrid edge lengths from $CF$s using the already identified $\gamma$. In case (R), a $3_1$-cycle below the hybrid node, by \cref{prop:3-cyc11} all $CF$s are unchanged if the 3-cycle is contracted to a node and the edge length above it modified appropriately. Then the identifiability of the hybrid edge lengths follows from the cherry case.

If the hybrid node has only 1 descendant taxon,  then at most 1 lineage may enter (going backwards in time) the hybrid edges of $C$, so no coalescent events may occur on the hybrid edges. Thus the $CF$s do not depend on the lengths of those edges, which are therefore not identifiable from  $CF$s.
\end{proof}

\begin{figure}
\begin{center}
\tikzset{every picture/.style={line width=0.75pt}} 

\begin{tikzpicture}[x=0.75pt,y=0.75pt,yscale=-0.9,xscale=0.9]

\draw [color={rgb, 255:red, 0; green, 0; blue, 0 }  ,draw opacity=1 ][line width=1.5]    (79.21,50.85) -- (92.75,36.5) ;
\draw [color={rgb, 255:red, 0; green, 0; blue, 0 }  ,draw opacity=1 ][line width=1.5]    (48.21,50.89) -- (79.21,50.85) ;
\draw [color={rgb, 255:red, 0; green, 0; blue, 0 }  ,draw opacity=1 ][line width=1.5]    (48.21,50.89) -- (40.5,76.64) ;
\draw [color={rgb, 255:red, 182; green, 5; blue, 60 }  ,draw opacity=1 ][line width=1.5]    (86.9,75.04) -- (65.87,96.44) ;
\draw [shift={(63.77,98.58)}, rotate = 314.49] [color={rgb, 255:red, 182; green, 5; blue, 60 }  ,draw opacity=1 ][line width=1.5]    (14.21,-4.28) .. controls (9.04,-1.82) and (4.3,-0.39) .. (0,0) .. controls (4.3,0.39) and (9.04,1.82) .. (14.21,4.28)   ;
\draw [color={rgb, 255:red, 0; green, 0; blue, 0 }  ,draw opacity=1 ][line width=1.5]    (31.25,37) -- (48.21,50.89) ;
\draw [line width=1.5]    (79.21,50.85) -- (86.9,75.04) ;
\draw [color={rgb, 255:red, 182; green, 5; blue, 60 }  ,draw opacity=1 ][line width=1.5]    (40.5,76.64) -- (61.59,96.52) ;
\draw [shift={(63.77,98.58)}, rotate = 223.32] [color={rgb, 255:red, 182; green, 5; blue, 60 }  ,draw opacity=1 ][line width=1.5]    (14.21,-4.28) .. controls (9.04,-1.82) and (4.3,-0.39) .. (0,0) .. controls (4.3,0.39) and (9.04,1.82) .. (14.21,4.28)   ;
\draw [color={rgb, 255:red, 0; green, 0; blue, 0 }  ,draw opacity=1 ][line width=1.5]    (22.25,76.5) -- (40.5,76.64) ;
\draw [color={rgb, 255:red, 0; green, 0; blue, 0 }  ,draw opacity=1 ][line width=1.5]    (102.57,75.35) -- (86.9,75.04) ;
\draw [line width=1.5]    (63.77,98.58) -- (63.77,117.49) ;
\draw [line width=1.5]    (148.3,76.64) -- (168.8,76.64) ;
\draw [line width=1.5]    (168.8,76.64) -- (188.63,47.09) ;
\draw [color={rgb, 255:red, 182; green, 5; blue, 60 }  ,draw opacity=1 ][line width=1.5]    (168.8,76.64) -- (186.97,103.89) ;
\draw [shift={(188.63,106.39)}, rotate = 236.31] [color={rgb, 255:red, 182; green, 5; blue, 60 }  ,draw opacity=1 ][line width=1.5]    (14.21,-4.28) .. controls (9.04,-1.82) and (4.3,-0.39) .. (0,0) .. controls (4.3,0.39) and (9.04,1.82) .. (14.21,4.28)   ;
\draw [color={rgb, 255:red, 182; green, 5; blue, 60 }  ,draw opacity=1 ][line width=1.5]    (209.15,76.64) -- (190.33,103.92) ;
\draw [shift={(188.63,106.39)}, rotate = 304.59] [color={rgb, 255:red, 182; green, 5; blue, 60 }  ,draw opacity=1 ][line width=1.5]    (14.21,-4.28) .. controls (9.04,-1.82) and (4.3,-0.39) .. (0,0) .. controls (4.3,0.39) and (9.04,1.82) .. (14.21,4.28)   ;
\draw [line width=1.5]    (209.15,76.64) -- (230.65,76.64) ;
\draw [line width=1.5]    (188.63,106.39) -- (188.63,125.3) ;
\draw [line width=1.5]    (209.15,142.02) -- (188.63,126.02) ;
\draw [color={rgb, 255:red, 0; green, 0; blue, 0 }  ,draw opacity=1 ][line width=1.5]    (168.8,142.02) -- (188.63,126.02) ;
\draw [line width=1.5]    (188.63,47.09) -- (209.15,76.64) ;
\draw [line width=1.5]    (188.63,27.69) -- (188.63,47.09) ;
\draw [line width=1.5]    (278.55,76.64) -- (299.05,76.64) ;
\draw [line width=1.5]    (299.05,76.64) -- (318.88,47.09) ;
\draw [color={rgb, 255:red, 182; green, 5; blue, 60 }  ,draw opacity=1 ][line width=1.5]    (299.05,76.64) -- (317.22,103.89) ;
\draw [shift={(318.88,106.39)}, rotate = 236.31] [color={rgb, 255:red, 182; green, 5; blue, 60 }  ,draw opacity=1 ][line width=1.5]    (14.21,-4.28) .. controls (9.04,-1.82) and (4.3,-0.39) .. (0,0) .. controls (4.3,0.39) and (9.04,1.82) .. (14.21,4.28)   ;
\draw [color={rgb, 255:red, 182; green, 5; blue, 60 }  ,draw opacity=1 ][line width=1.5]    (339.4,76.64) -- (320.59,103.92) ;
\draw [shift={(318.88,106.39)}, rotate = 304.59] [color={rgb, 255:red, 182; green, 5; blue, 60 }  ,draw opacity=1 ][line width=1.5]    (14.21,-4.28) .. controls (9.04,-1.82) and (4.3,-0.39) .. (0,0) .. controls (4.3,0.39) and (9.04,1.82) .. (14.21,4.28)   ;
\draw [line width=1.5]    (339.4,76.64) -- (360.9,76.64) ;
\draw [line width=1.5]    (318.88,106.39) -- (318.88,125.3) ;
\draw [color={rgb, 255:red, 182; green, 5; blue, 60 }  ,draw opacity=1 ][line width=1.5]    (337.03,140.17) -- (318.88,126.02) ;
\draw [shift={(339.4,142.02)}, rotate = 217.95] [color={rgb, 255:red, 182; green, 5; blue, 60 }  ,draw opacity=1 ][line width=1.5]    (14.21,-4.28) .. controls (9.04,-1.82) and (4.3,-0.39) .. (0,0) .. controls (4.3,0.39) and (9.04,1.82) .. (14.21,4.28)   ;
\draw [color={rgb, 255:red, 0; green, 0; blue, 0 }  ,draw opacity=1 ][line width=1.5]    (299.05,142.02) -- (318.88,126.02) ;
\draw [line width=1.5]    (318.88,47.09) -- (339.4,76.64) ;
\draw [line width=1.5]    (318.88,27.69) -- (318.88,47.09) ;
\draw [color={rgb, 255:red, 182; green, 5; blue, 60 }  ,draw opacity=1 ][line width=1.5]    (299.05,142.02) -- (336.4,142.02) ;
\draw [shift={(339.4,142.02)}, rotate = 180] [color={rgb, 255:red, 182; green, 5; blue, 60 }  ,draw opacity=1 ][line width=1.5]    (14.21,-4.28) .. controls (9.04,-1.82) and (4.3,-0.39) .. (0,0) .. controls (4.3,0.39) and (9.04,1.82) .. (14.21,4.28)   ;
\draw [line width=1.5]    (285.4,153.22) -- (299.05,142.02) ;
\draw [line width=1.5]    (353.4,153.22) -- (339.4,142.02) ;

\draw (276.38,153.55) node [anchor=north west][inner sep=0.75pt]  [font=\normalsize,color={rgb, 255:red, 6; green, 79; blue, 97 }  ,opacity=1 ]  {$a_{1}$};
\draw (347.48,153.55) node [anchor=north west][inner sep=0.75pt]  [font=\normalsize,color={rgb, 255:red, 6; green, 79; blue, 97 }  ,opacity=1 ]  {$a_{2}$};
\draw (264.8,67.66) node [anchor=north west][inner sep=0.75pt]  [font=\normalsize,color={rgb, 255:red, 6; green, 79; blue, 97 }  ,opacity=1 ]  {$b$};
\draw (314.68,12.45) node [anchor=north west][inner sep=0.75pt]  [font=\normalsize,color={rgb, 255:red, 6; green, 79; blue, 97 }  ,opacity=1 ]  {$c$};
\draw (363.13,67.16) node [anchor=north west][inner sep=0.75pt]  [font=\normalsize,color={rgb, 255:red, 6; green, 79; blue, 97 }  ,opacity=1 ]  {$d$};
\draw (232.37,67.16) node [anchor=north west][inner sep=0.75pt]  [font=\normalsize,color={rgb, 255:red, 6; green, 79; blue, 97 }  ,opacity=1 ]  {$d$};
\draw (161.63,144.05) node [anchor=north west][inner sep=0.75pt]  [font=\normalsize,color={rgb, 255:red, 6; green, 79; blue, 97 }  ,opacity=1 ]  {$a_{1}$};
\draw (183.43,12.45) node [anchor=north west][inner sep=0.75pt]  [font=\normalsize,color={rgb, 255:red, 6; green, 79; blue, 97 }  ,opacity=1 ]  {$c$};
\draw (203.23,144.05) node [anchor=north west][inner sep=0.75pt]  [font=\normalsize,color={rgb, 255:red, 6; green, 79; blue, 97 }  ,opacity=1 ]  {$a_{2}$};
\draw (134.55,67.16) node [anchor=north west][inner sep=0.75pt]  [font=\normalsize,color={rgb, 255:red, 6; green, 79; blue, 97 }  ,opacity=1 ]  {$b$};
\draw (10.8,67.81) node [anchor=north west][inner sep=0.75pt]  [font=\normalsize,color={rgb, 255:red, 6; green, 79; blue, 97 }  ,opacity=1 ]  {$b$};
\draw (48,72.51) node [anchor=north west][inner sep=0.75pt]  [font=\small,color={rgb, 255:red, 182; green, 5; blue, 60 }  ,opacity=1 ]  {$\gamma $};
\draw (29.55,53.62) node [anchor=north west][inner sep=0.75pt]  [font=\small]  {$x$};
\draw (60.55,33.62) node [anchor=north west][inner sep=0.75pt]  [font=\small]  {$y$};
\draw (57.88,119.2) node [anchor=north west][inner sep=0.75pt]  [font=\normalsize,color={rgb, 255:red, 6; green, 79; blue, 97 }  ,opacity=1 ]  {$a$};
\draw (19.68,23.1) node [anchor=north west][inner sep=0.75pt]  [font=\normalsize,color={rgb, 255:red, 6; green, 79; blue, 97 }  ,opacity=1 ]  {$c$};
\draw (94.63,19.31) node [anchor=north west][inner sep=0.75pt]  [font=\normalsize,color={rgb, 255:red, 6; green, 79; blue, 97 }  ,opacity=1 ]  {$d$};
\draw (104.13,67.31) node [anchor=north west][inner sep=0.75pt]  [font=\normalsize,color={rgb, 255:red, 6; green, 79; blue, 97 }  ,opacity=1 ]  {$e$};

\end{tikzpicture}
\vspace*{-1cm}
\end{center}		
\caption{Subnetworks used in the proof of \cref{prop:LargeCyc}.}\label{fig:5cyc}
\end{figure}

We next turn to cut edges adjacent to a single hybrid edge.

\begin{proposition} Let $\widetilde N$ be a level-1 metric binary semidirected network with no 2-cycles, containing an internal cut edge 
	$e$ adjacent to exactly one hybrid edge (at its non-hybrid node), with the hybrid edge in a $k$-cycle. If $k\ge 4$, then the length of $e$ 
	is identifiable.
\end{proposition}
\begin{proof}
If $k\ge 4$, by passing to the induced network on a subset of the taxa, we may assume $k=4$. Since $e$ is not pendant, and not adjacent to a hybrid edge of another cycle, after again passing  to an induced subnetwork and replacing any 2-cycles with single edges, we may assume the network has the structure of $N_w$ in \cref{fig:4cyc}, with $e$ the edge joining the cherry to the 4-cycle. But then \cref{prop:W}(c)  gives the claim.
\end{proof}

If an edge is adjacent to hybrid edges at both of its endpoints, but neither endpoint is a hybrid node, as in \cref{fig:2_4cyc} (L), then the following applies.

\begin{proposition}\label{prop:join2hyb}

Let $\widetilde N$ be a level-1 metric binary semidirected network with no 2-cycles, containing an edge $e$ adjacent to exactly two hybrid edges which lie in two different cycles. If the sizes of both cycles are $\ge 4$, then the length of $e$ is identifiable. \end{proposition}
\begin{proof}
If both cycles are of size $\ge 4$, then the network has an induced subnetwork which, after suppressing 2-cycles has the form shown in \cref{fig:2_4cyc}(L), with the central edge arising from $e$, with edge probability $\ell$. 

Using \cref{prop:W} on the induced network after dropping taxon $f$ we may identify $\gamma, x_1,x_2$ and  the product $\ell y_1$. Similarly, dropping $a$ we may identify $y_1$, which then gives $\ell$.
\end{proof}

\begin{figure}
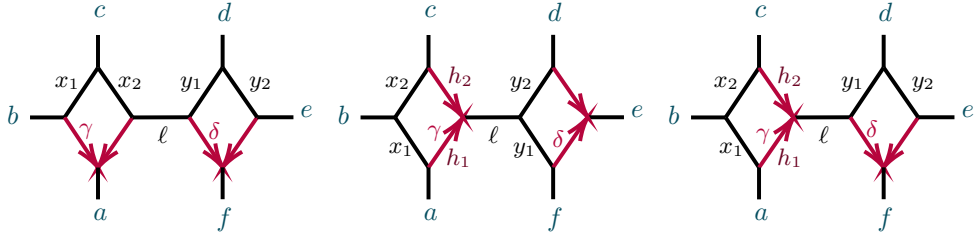

\begin{center}
\include{Figures/2_4cycles}
\end{center}
\caption{Semidirected binary networks on 6 taxa with two 4-cycles joined by an edge adjacent to two or more hybrid edges. }\label{fig:2_4cyc}
 \end{figure}

Next we consider edges adjacent to two hybrid edges at one endpoint, that is, edges with a hybrid node as an endpoint, as in \cref{fig:2_4cyc}(C,R). If the hybrid node is in a large cycle we obtain the following.

\begin{proposition}\label{prop:5cyc}
Let $\widetilde N$ be a level-1 metric binary semidirected network with no 2-cycles, containing an edge $q$ whose parent is the hybrid node of a $k$-cycle with $k\ge 5$. If $q$ has at least two descendant taxa, and the child node of $q$ is not in a 3-cycle, then the length of $q$ is identifiable.\end{proposition}
\begin{proof} 
Since the cycle is of size $\ge 5$, by \cref{prop:LargeCyc} its hybridization parameter $\gamma$ is identified.  

First suppose the child node of $q$ is not incident to a hybrid edge.
If $q$ has two descendant taxa, there is an induced subnetwork which, after replacing 2-cycles by edges, has  the form of $N_s$ of \cref{fig:4cyc}(L), with  $q$  the child edge of the hybrid node. With $\gamma$ in hand,  by \cref{prop:S}(c) the length of $q$ is identified.

If instead the child node of $q$ is incident to a hybrid edge, assume that edge lies in a cycle of size $\ge 4$. We may then pass to a network with the structure of \cref{fig:2_4cyc}(R) where $q$ is the edge joining the two cycles.
But dropping taxon $f$ again yields a network of form $N_s$, so using $\gamma$ we identify $\ell y_1$. Instead dropping $b$ from  \cref{fig:2_4cyc}(R), by \cref{prop:3-cyc11}, the 3-cycle on this can then be contracted to a node, adjusting the edge
length of $q$  (now possibly negative) so $CF$s are unchanged. Then \cref{prop:W} can be applied to identify $y_1$. Thus $\ell$ is identifiable.
\end{proof}

The remaining parameters to consider are  the edge probabilities and hybridization parameter in 4-cycles,  and the edge probability of the child edge of the hybrid node in a 4-cycle. Identifiability of these is more complicated, as it can depend on the sizes of the taxon blocks of the cycle. In handling these cases, we use the following.

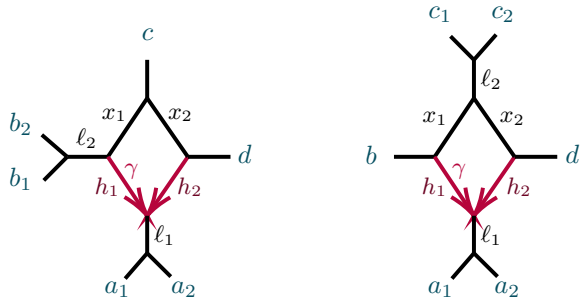
\begin{figure}
	\begin{center}
		\tikzset{every picture/.style={line width=0.75pt}} 

\begin{tikzpicture}[x=0.75pt,y=0.75pt,yscale=-1,xscale=1]

\draw [line width=1.5]    (52.05,101.79) -- (72.55,101.79) ;
\draw [line width=1.5]    (72.55,101.79) -- (92.38,72.25) ;
\draw [color={rgb, 255:red, 182; green, 5; blue, 60 }  ,draw opacity=1 ][line width=1.5]    (72.55,101.79) -- (90.72,129.04) ;
\draw [shift={(92.38,131.54)}, rotate = 236.31] [color={rgb, 255:red, 182; green, 5; blue, 60 }  ,draw opacity=1 ][line width=1.5]    (14.21,-4.28) .. controls (9.04,-1.82) and (4.3,-0.39) .. (0,0) .. controls (4.3,0.39) and (9.04,1.82) .. (14.21,4.28)   ;
\draw [color={rgb, 255:red, 182; green, 5; blue, 60 }  ,draw opacity=1 ][line width=1.5]    (112.9,101.79) -- (94.09,129.07) ;
\draw [shift={(92.38,131.54)}, rotate = 304.59] [color={rgb, 255:red, 182; green, 5; blue, 60 }  ,draw opacity=1 ][line width=1.5]    (14.21,-4.28) .. controls (9.04,-1.82) and (4.3,-0.39) .. (0,0) .. controls (4.3,0.39) and (9.04,1.82) .. (14.21,4.28)   ;
\draw [line width=1.5]    (112.9,101.79) -- (134.4,101.79) ;
\draw [line width=1.5]    (92.38,131.54) -- (92.38,150.45) ;
\draw [line width=1.5]    (92.38,72.25) -- (112.9,101.79) ;
\draw [line width=1.5]    (92.38,52.84) -- (92.38,72.25) ;
\draw [line width=1.5]    (104.28,162.26) -- (92.38,150.45) ;
\draw [line width=1.5]    (81.28,163.26) -- (92.38,150.45) ;
\draw [line width=1.5]    (216.85,101.79) -- (237.35,101.79) ;
\draw [color={rgb, 255:red, 0; green, 0; blue, 0 }  ,draw opacity=1 ][line width=1.5]    (237.35,101.79) -- (257.18,72.25) ;
\draw [color={rgb, 255:red, 182; green, 5; blue, 60 }  ,draw opacity=1 ][line width=1.5]    (237.35,101.79) -- (255.52,129.04) ;
\draw [shift={(257.18,131.54)}, rotate = 236.31] [color={rgb, 255:red, 182; green, 5; blue, 60 }  ,draw opacity=1 ][line width=1.5]    (14.21,-4.28) .. controls (9.04,-1.82) and (4.3,-0.39) .. (0,0) .. controls (4.3,0.39) and (9.04,1.82) .. (14.21,4.28)   ;
\draw [color={rgb, 255:red, 182; green, 5; blue, 60 }  ,draw opacity=1 ][line width=1.5]    (277.7,101.79) -- (258.88,129.07) ;
\draw [shift={(257.18,131.54)}, rotate = 304.59] [color={rgb, 255:red, 182; green, 5; blue, 60 }  ,draw opacity=1 ][line width=1.5]    (14.21,-4.28) .. controls (9.04,-1.82) and (4.3,-0.39) .. (0,0) .. controls (4.3,0.39) and (9.04,1.82) .. (14.21,4.28)   ;
\draw [color={rgb, 255:red, 0; green, 0; blue, 0 }  ,draw opacity=1 ][line width=1.5]    (277.7,101.79) -- (299.2,101.79) ;
\draw [color={rgb, 255:red, 0; green, 0; blue, 0 }  ,draw opacity=1 ][line width=1.5]    (257.18,131.54) -- (257.18,150.45) ;
\draw [color={rgb, 255:red, 0; green, 0; blue, 0 }  ,draw opacity=1 ][line width=1.5]    (257.18,72.25) -- (277.7,101.79) ;
\draw [color={rgb, 255:red, 0; green, 0; blue, 0 }  ,draw opacity=1 ][line width=1.5]    (257.18,52.84) -- (257.18,72.25) ;
\draw [line width=1.5]    (245.37,40.94) -- (257.18,52.84) ;
\draw [line width=1.5]    (268.38,40.11) -- (257.18,52.84) ;
\draw [line width=1.5]    (269.08,162.26) -- (257.18,150.45) ;
\draw [line width=1.5]    (246.08,163.26) -- (257.18,150.45) ;
\draw [line width=1.5]    (40.24,113.69) -- (52.05,101.79) ;
\draw [line width=1.5]    (39.24,90.69) -- (52.05,101.79) ;

\draw (301.86,94.01) node [anchor=north west][inner sep=0.75pt]  [font=\normalsize,color={rgb, 255:red, 6; green, 79; blue, 97 }  ,opacity=1 ]  {$d$};
\draw (201.04,94.61) node [anchor=north west][inner sep=0.75pt]  [font=\normalsize,color={rgb, 255:red, 6; green, 79; blue, 97 }  ,opacity=1 ]  {$b$};
\draw (244.74,104.01) node [anchor=north west][inner sep=0.75pt]  [font=\small,color={rgb, 255:red, 182; green, 5; blue, 60 }  ,opacity=1 ]  {$\gamma $};
\draw (55.05,86.24) node [anchor=north west][inner sep=0.75pt]  [font=\small,color={rgb, 255:red, 0; green, 0; blue, 0 }  ,opacity=1 ]  {$\ell _{2}$};
\draw (229.79,77.79) node [anchor=north west][inner sep=0.75pt]  [font=\small]  {$x_{1}$};
\draw (268.79,77.79) node [anchor=north west][inner sep=0.75pt]  [font=\small]  {$x_{2}$};
\draw (232.83,24.4) node [anchor=north west][inner sep=0.75pt]  [font=\normalsize,color={rgb, 255:red, 6; green, 79; blue, 97 }  ,opacity=1 ]  {$c_{1}$};
\draw (264.5,24.4) node [anchor=north west][inner sep=0.75pt]  [font=\normalsize,color={rgb, 255:red, 6; green, 79; blue, 97 }  ,opacity=1 ]  {$c_{2}$};
\draw (136.63,94.31) node [anchor=north west][inner sep=0.75pt]  [font=\normalsize,color={rgb, 255:red, 6; green, 79; blue, 97 }  ,opacity=1 ]  {$d$};
\draw (69.38,163.46) node [anchor=north west][inner sep=0.75pt]  [font=\normalsize,color={rgb, 255:red, 6; green, 79; blue, 97 }  ,opacity=1 ]  {$a_{1}$};
\draw (88.18,35.6) node [anchor=north west][inner sep=0.75pt]  [font=\normalsize,color={rgb, 255:red, 6; green, 79; blue, 97 }  ,opacity=1 ]  {$c$};
\draw (79.5,104.01) node [anchor=north west][inner sep=0.75pt]  [font=\small,color={rgb, 255:red, 182; green, 5; blue, 60 }  ,opacity=1 ]  {$\gamma $};
\draw (64.45,110.51) node [anchor=north west][inner sep=0.75pt]  [font=\small,color={rgb, 255:red, 108; green, 4; blue, 35 }  ,opacity=1 ]  {$h_{1}$};
\draw (94.38,134.94) node [anchor=north west][inner sep=0.75pt]  [font=\small,color={rgb, 255:red, 0; green, 0; blue, 0 }  ,opacity=1 ]  {$\ell _{1}$};
\draw (68.05,75.79) node [anchor=north west][inner sep=0.75pt]  [font=\small]  {$x_{1}$};
\draw (102.98,163.46) node [anchor=north west][inner sep=0.75pt]  [font=\normalsize,color={rgb, 255:red, 6; green, 79; blue, 97 }  ,opacity=1 ]  {$a_{2}$};
\draw (101.38,75.79) node [anchor=north west][inner sep=0.75pt]  [font=\small]  {$x_{2}$};
\draw (106.05,110.51) node [anchor=north west][inner sep=0.75pt]  [font=\small,color={rgb, 255:red, 108; green, 4; blue, 35 }  ,opacity=1 ]  {$h_{2}$};
\draw (232.38,163.46) node [anchor=north west][inner sep=0.75pt]  [font=\normalsize,color={rgb, 255:red, 6; green, 79; blue, 97 }  ,opacity=1 ]  {$a_{1}$};
\draw (265.98,163.46) node [anchor=north west][inner sep=0.75pt]  [font=\normalsize,color={rgb, 255:red, 6; green, 79; blue, 97 }  ,opacity=1 ]  {$a_{2}$};
\draw (21.29,78) node [anchor=north west][inner sep=0.75pt]  [font=\normalsize,color={rgb, 255:red, 6; green, 79; blue, 97 }  ,opacity=1 ]  {$b_{2}$};
\draw (21.29,105.61) node [anchor=north west][inner sep=0.75pt]  [font=\normalsize,color={rgb, 255:red, 6; green, 79; blue, 97 }  ,opacity=1 ]  {$b_{1}$};
\draw (229.45,110.51) node [anchor=north west][inner sep=0.75pt]  [font=\small,color={rgb, 255:red, 108; green, 4; blue, 35 }  ,opacity=1 ]  {$h_{1}$};
\draw (272.05,110.51) node [anchor=north west][inner sep=0.75pt]  [font=\small,color={rgb, 255:red, 108; green, 4; blue, 35 }  ,opacity=1 ]  {$h_{2}$};
\draw (259.18,134.94) node [anchor=north west][inner sep=0.75pt]  [font=\small,color={rgb, 255:red, 0; green, 0; blue, 0 }  ,opacity=1 ]  {$\ell _{1}$};
\draw (259.18,56.24) node [anchor=north west][inner sep=0.75pt]  [font=\small,color={rgb, 255:red, 0; green, 0; blue, 0 }  ,opacity=1 ]  {$\ell _{2}$};

\end{tikzpicture}
\vspace*{-1cm}
		\caption{Semidirected binary networks on 6 taxa with a 4-cycle and two cherries: (L) $N_{sw}$ and (R) $N_{sn}$. }\label{fig:WSandNS}
 	\end{center}
\end{figure}

\begin{lemma}\label{lem:4cyc2cher}Consider a 6-taxon semidirected network with a 4-cycle, a cherry below the cycle's hybrid node, and one other cherry, as shown in \cref{fig:WSandNS}. Then all numerical parameters are identifiable from quartet  $CF$s.
\end{lemma}

\begin{proof}
Consider \cref{fig:WSandNS}(L), $N_{sw}$. Then the subnetwork obtained by dropping taxon $a_2$ has the form of $N_w$, and \cref{prop:W} shows $\gamma,x_1,x_2,\ell_2$ are identifiable. But the network obtained by dropping taxon $b_2$ has the form of $N_s$, so using \cref{prop:S} and the known value of $\gamma$ identifies $h_1,h_2,\ell_1$.

The identifiability of all  parameters for \cref{fig:WSandNS}(R), $N_{sn}$, follows from another computation, presented as \cref{prop:SN4cycle}. 
\end{proof}

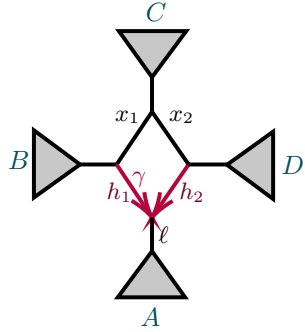
\begin{figure}
 \begin{center} 
  \tikzset{every picture/.style={line width=0.75pt}} 

\begin{tikzpicture}[x=0.75pt,y=0.75pt,yscale=-0.9,xscale=0.9]

\draw [line width=1.5]    (90.05,101.79) -- (110.55,101.79) ;
\draw [line width=1.5]    (110.55,101.79) -- (130.38,72.25) ;
\draw [color={rgb, 255:red, 182; green, 5; blue, 60 }  ,draw opacity=1 ][line width=1.5]    (110.55,101.79) -- (128.72,129.04) ;
\draw [shift={(130.38,131.54)}, rotate = 236.31] [color={rgb, 255:red, 182; green, 5; blue, 60 }  ,draw opacity=1 ][line width=1.5]    (14.21,-4.28) .. controls (9.04,-1.82) and (4.3,-0.39) .. (0,0) .. controls (4.3,0.39) and (9.04,1.82) .. (14.21,4.28)   ;
\draw [color={rgb, 255:red, 182; green, 5; blue, 60 }  ,draw opacity=1 ][line width=1.5]    (150.9,101.79) -- (132.09,129.07) ;
\draw [shift={(130.38,131.54)}, rotate = 304.59] [color={rgb, 255:red, 182; green, 5; blue, 60 }  ,draw opacity=1 ][line width=1.5]    (14.21,-4.28) .. controls (9.04,-1.82) and (4.3,-0.39) .. (0,0) .. controls (4.3,0.39) and (9.04,1.82) .. (14.21,4.28)   ;
\draw [line width=1.5]    (150.9,101.79) -- (172.4,101.79) ;
\draw [line width=1.5]    (130.38,131.54) -- (130.38,150.45) ;
\draw [line width=1.5]    (130.38,72.25) -- (150.9,101.79) ;
\draw [line width=1.5]    (130.38,52.84) -- (130.38,72.25) ;
\draw  [fill={rgb, 255:red, 155; green, 155; blue, 155 }  ,fill opacity=0.55 ][line width=1.5]  (172.4,101.97) -- (198.25,83.42) -- (198.25,121.14) -- cycle ;
\draw  [fill={rgb, 255:red, 155; green, 155; blue, 155 }  ,fill opacity=0.55 ][line width=1.5]  (90.05,101.79) -- (64.2,120.34) -- (64.2,82.62) -- cycle ;
\draw  [fill={rgb, 255:red, 155; green, 155; blue, 155 }  ,fill opacity=0.55 ][line width=1.5]  (130.38,52.84) -- (111.83,26.99) -- (149.56,26.99) -- cycle ;
\draw  [fill={rgb, 255:red, 155; green, 155; blue, 155 }  ,fill opacity=0.55 ][line width=1.5]  (130.38,150.45) -- (148.93,176.3) -- (111.21,176.3) -- cycle ;

\draw (201.34,95.31) node [anchor=north west][inner sep=0.75pt]  [font=\normalsize,color={rgb, 255:red, 6; green, 79; blue, 97 }  ,opacity=1 ]  {$D$};
\draw (122.24,180.51) node [anchor=north west][inner sep=0.75pt]  [font=\normalsize,color={rgb, 255:red, 6; green, 79; blue, 97 }  ,opacity=1 ]  {$A$};
\draw (125.18,9.46) node [anchor=north west][inner sep=0.75pt]  [font=\normalsize,color={rgb, 255:red, 6; green, 79; blue, 97 }  ,opacity=1 ]  {$C$};
\draw (117.5,104.01) node [anchor=north west][inner sep=0.75pt]  [font=\small,color={rgb, 255:red, 182; green, 5; blue, 60 }  ,opacity=1 ]  {$\gamma $};
\draw (103.45,110.51) node [anchor=north west][inner sep=0.75pt]  [font=\small,color={rgb, 255:red, 108; green, 4; blue, 35 }  ,opacity=1 ]  {$h_{1}$};
\draw (132.38,134.94) node [anchor=north west][inner sep=0.75pt]  [font=\small,color={rgb, 255:red, 0; green, 0; blue, 0 }  ,opacity=1 ]  {$\ell $};
\draw (108.05,69.79) node [anchor=north west][inner sep=0.75pt]  [font=\small]  {$x_{1}$};
\draw (138.38,69.79) node [anchor=north west][inner sep=0.75pt]  [font=\small]  {$x_{2}$};
\draw (144.05,110.51) node [anchor=north west][inner sep=0.75pt]  [font=\small,color={rgb, 255:red, 108; green, 4; blue, 35 }  ,opacity=1 ]  {$h_{2}$};
\draw (48.14,92.43) node [anchor=north west][inner sep=0.75pt]  [font=\normalsize,color={rgb, 255:red, 6; green, 79; blue, 97 }  ,opacity=1 ]  {$B$};

\end{tikzpicture}
\vspace*{-1cm}
  \caption{A 4-cycle in a larger network, partitioning the taxa into 4 blocks $A,B,C,D$.}\label{fig:4groups}
 \end{center}
\end{figure}

\begin{proposition}\label{prop:4cyc}
Let $\widetilde N$ be a level-1 metric binary semidirected network on $n\ge 5$ taxa with no 2-cycles, containing  a 4-cycle, as shown in \cref{fig:4groups},  with taxon blocks $A,B,C,D$ of size $n_A$, $n_B$, $n_C$, $n_D$ and edge probabilities and hybridization parameters on and below the cycle as shown. Then the parameters $x_1,x_2,h_1,h_2,\gamma, \ell$ are identifiable according to the following cases, at least one of which must hold.

\begin{itemize}
\item[a)] $n_B=n_C=n_D=1$: none identifiable
\item[b)] $n_A=1$ and
\begin{itemize}
\item[i)] $n_B=n_D=1$: none identifiable
\item[ii)] $n_B$ or $n_D\ge 2$: $x_1,x_2,\gamma$ identifiable, $h_1,h_2,\ell$ not identifiable
\end{itemize}
\item[c)] $n_A\ge 2$; $n_B$, $n_C$, or $n_D\ge 2$ and
\begin{enumerate}
\item[i)] the child of the edge with probability $\ell$ is not in a 3-cycle:  all identifiable
\item[ii)] the child of the edge with probability $\ell$ is in a 3-cycle:  $x_1,x_2,h_1,h_2,\gamma$ identifiable, $\ell$ not identifiable
\end{enumerate}
\end{itemize}

\end{proposition}

Simple instances of the 5 cases in the proposition may be helpful to consider. The network $N_s$ falls under case a), $N_n$ under b)i), $N_w$ under b)ii), and $N_{sw}$ and $N_{sn}$ under c)i). Examples for case c)ii) are obtained from  $N_{sw}$ and $N_{sn}$ by replacing the cherry below the hybrid edge with a 3-cycle. The proof of the proposition leverages computational results for these to obtain more general statements.

\begin{proof}
That at least one of these cases must hold is most easily seen by noting that  case c) is the complement of the union of
a) and b). We consider each case to establish its claim.

\smallskip\noindent{\bf Case a):} The 4-cycle determines a hybrid  block of taxa $A$ and  three taxa, $b,c,d$, in singleton blocks. The only $CF$s dependent on the parameters $\theta=(x_1,x_2,h_1,h_2,\gamma, \ell)$ are those involving at most two elements of $A$, since with 3 or 4 elements of $A$ either a coalescence has occurred below the hybrid node, or at least 3 lineages reach it and are then exchangeable, giving probabilities $1/3$ for each quartet tree. Those $CF$s dependent on $\theta$  decompose into sums of products of expressions involving only parameters outside of $\theta$ or only parameters in $\theta$, similar to the approach in  \Cref{sec:large3}. The expressions involving only parameters in $\theta$ can even be chosen from the $CF$s for the network $N_s$ of \ref{prop:S}. But that Proposition shows the parameters in $\theta$ are not identifiable from the $CF$s for $N_s$, so they cannot be identified from those for $\widetilde N$.

\smallskip\noindent 
{\bf Case b)i):} The 4-cycle determines a hybrid  singleton  $a$, two adjacent singleton blocks of $b$ and $d$, and a larger subnetwork $C$ opposite the hybrid. Viewing the network as rooted in $C$, the $CF$s for $\widetilde N$ depend on parameters $x_1,x_2,h_1,h_2,\gamma, \ell$ only through the various probabilities of first coalescent events among subsets of $\{a,b,d\}$ determining the quartet tree before lineages leave the 4-cycle and enter $C$ . Using $D$ to denote the subnetwork below $C$ which contains the 4-cycle, these are
\begin{align*}
p_1=P(\mathcal C_D\to ab|cc)&=\gamma(1-x_1) \\
p_2=P(\mathcal C_D\to ad|cc)&= (1-\gamma)(1-x_2)\\
P(\mathcal C_D\to bd|cc)&= 0\\
P(\mathcal C_D\to bd|ac)&=  (\gamma x_1 +(1-\gamma)x_2)/3=(1-p_1-p_2)/3\\
P(\mathcal C_D\to ab|dc)&= \gamma\left (1- 2 x_1/3\right )+(1-\gamma) x_2/3=(1+2p_1-p_2)/3\\
P(\mathcal C_D\to ad|bc)&= \gamma x_1/3+(1-\gamma)\left (1-2 x_2/3\right )=(1-p_1+2p_2)/3
\end{align*}
Since these probabilities are linear functions of $p_1,p_2$, and none of $\gamma,x_1,x_2$ are identifiable from $p_1,p_2$, none of the parameters are identifiable from $CF$s for $\widetilde N$.

\smallskip\noindent 
{\bf Case b)ii):} Pick two taxa in one of the blocks adjacent to the hybrid one, and one taxon in all others.  Passing to the induced subnetwork and removing 2-cycles yields either a network with the form $N_w$ or one where the cherry in $N_w$ is replaced by a 3-cycle. Using 
\cref{prop:3-cyc11}, we may replace such a 3-cycle with a node without changing $CF$s,  (provided we modify the edge length leading to the 4-cycle, including allowing  for a possibly negative branch length. But then the network has the form $N_w$ and applying \cref{prop:W} shows $\gamma, x_1,x_2$ can be identified. 

Since there is only one taxon descended from the hybrid node, there can be no coalescent event in either of the hybrid edges or their descendant, and thus these edge lengths do not appear in the formulas for the $CF$s for $\widetilde N$. Therefore these parameters cannot be identifiable.

\smallskip\noindent 
{\bf Case c)i):}  Pick two taxa in one of the non-hybrid blocks, two taxa in the hybrid block, and one taxon from each of the others. Passing to the induced subnetwork on these 6 taxa, and removing any 2-cycles, we obtain a network of one of the forms in \cref{fig:WSandNS}, or ones where 3-cycles appear in place of one or both cherry nodes.   If there are 3-cycles, by \cref{prop:3-cyc11} we may replace them with nodes without changing $CF$s (provided we modify edge lengths leading to the 4-cycle). Then using \cref{lem:4cyc2cher}  we can identify $\gamma,x_1,x_2,h_1,h_2$.

To identify $\ell$, let $v$ be the child node of the edge  with this probability. If $v$ is not in a cycle in $\widetilde N$, then picking one taxon descended from each of its child edges and passing to an induced subnetwork, $\ell$ is identifiable by \cref{lem:4cyc2cher}.

If $v$ is in a cycle, it is of size $\ge 4$. Passing to an induced subnetwork, we may assume that $v$ is in a 4-cycle. Note that $v$ cannot be the hybrid node of that cycle, else the semidirected network would not be rootable. If $v$ is opposite the hybrid node, then we may pass to an induced subnetwork which, after replacing 2-cycles with edges, has a cherry below $v$ and follow the previous argument. If $v$ is adjacent to the hybrid node, then the subnetwork has the form of \cref{fig:2_4cyc}(R).  Since $\gamma$ is identified, the argument used in \cref{prop:5cyc} then shows $\ell$ is identifiable.

\smallskip\noindent
{\bf Case c)ii):} The argument of  the first paragraph for Case c)i) shows $\gamma,x_1,x_2,h_1,h_2$ are identifiable. Since the edge descending from the hybrid node of the 4-cycle is incident to a 3-cycle, its length is not identifiable by \cref{prop:3cycNonID}.
\end{proof}

\subsection{Summary of numerical parameter identifiability}

We summarize this section's results with the following.

\begin{theorem} [Numerical parameter identifiability from quartet  $CF$s]\label{thm:mainNum}
Let  $\widetilde N$ be a level-1 metric binary semidirected network with no 2-cycles. Then from quartet $CF$s under the NMSC with one sample per taxon all numerical parameters on $\widetilde N$ are identifiable except
for the following, which are not identifiable:
\begin{enumerate}
\item\label{it:pendant} pendant edge lengths,
\item for 3-cycles, hybridization parameters and the lengths of the six edges in and adjacent to the cycle,
\item\label{it:4cyc} for 4-cycles, the hybridization parameter and edge lengths in the cycle  and descended from the hybrid node, as stated in \cref{prop:4cyc}.
\end{enumerate}
\end{theorem}

If two individuals are sampled in some taxon $x$, as discussed earlier this can be modeled by attaching a cherry of pseudotaxa $x_1,x_2$ at the leaf $x$, Doing so for all taxa resolves the non-identifiability issues of \cref{it:pendant,it:4cyc}, yielding the following.

\begin{corollary} Let  $\widetilde N$ be a level-1 metric binary semidirected network with no 2-cycles. Then from quartet $CF$s under the NMSC with two or more samples for all taxa, all numerical parameters on $\widetilde N$ are identifiable except  hybridization parameters and lengths of edges in and adjacent to 3-cycles.
\end{corollary}

\section{Implications for data analysis}\label{sec:imp}

Attempting to infer the non-identifiable can either be misleading (unless all possible alternatives are reported), or very slow (spending computational time considering equally good possibilities), so our results here should inform development and use of $CF$-based inference methods.

The nature of 3-cycle identifiability from quartet $CF$s poses particular issues for likelihood and pseudolikelihood approaches. Quartet CFs carry signals of undirected 3-cycles, so ignoring the possibility of such cycles could have unknown consequences.
But even if only the network topology is sought, these criteria require optimization over numerical parameters, so it seems necessary to include 3-cycle parameters in a search.
Since for some parameter values there is a signal of a 3-cycle's hybrid node in the $CF$s, the search cannot be limited to undirected 3-cycles. However, these numerical parameters are not themselves identifiable, so considering them will be slow.
Reducing the over-parameterization at 3-cycles (from 7 parameters to 3) would be desirable, but it is unclear how to do so while maintaining the same range of $CF$s.
As the numerical parameters vary, the semidirected topology may pass between identifiable and non-identifiable regimes, and the boundaries of these are not known.
 
SNaQ, with its default settings, ignores the possibility of a 3-cycle and it is unclear how this might affect its optimization. Exactly what information in $CF$s is extracted by maximizing the pseudolikelihood function is difficult to analyze theoretically. Using simulation, the impact of 3-cycles on inference needs to be studied thoroughly, both for SNaQ and  for PhyloNet's similar inference from rooted triples.

NANUQ does not suffer from these problems, as its inference goal is more modest, providing a statistically consistent estimate only of larger cycle topology, without any search over the numerical parameter space. Whether NANUQ can be supplemented to extract $CF$ information on the existence of 3-cycles should also be explored.

Finally, identifiability theorems needed to justify 
network inference methods from data types other than $CF$s are largely lacking.   Studies of the parameter identifiability question for 
these are also needed.

\appendix
\section{Context for this work}\label{app:appeasement}

\subsection{Earlier work} 
Here we discuss how results in this work fit into and complement earlier work. In what follows, all statements should be restricted to the class of level-1 binary networks.

\medskip

Questions of identifiability of both network topology and numerical parameters from quartet $CF$s were 
first considered in \cite{Solis-Lemus2016}, a work which also introduced the pseudolikelihood inference of networks from empirical $CF$s via the SNaQ software. 
With the publication directed primarily to a biological audience, much of the presentation of the mathematical results 
appeared in the supplementary material, without formal statements of propositions and theorems, making it difficult to 
discern exactly what has been fully established. Subsequently, \cite{Banos2019} reconsidered the identifiability of 
level-1 network topology from quartet $CF$s with more mathematical formality.
A later  preprint \cite{Solis-Lemus2020} gives more details on the work behind \cite{Solis-Lemus2016}. 

\medskip

As does this paper, \cite{Solis-Lemus2016} uses algebraic computations to study what information quartet $CF$s might carry about a network. 
However, the computations concerning network topology are focused not on the full question of  identifiability,  but on a more restricted question 
the authors called \emph{detectability} of  hybridizations. This notion is an instance of the broader concept of \emph{distinguishability} presented 
subsequently in \cite{Degnan2018}. 
Distinguishability for a specified collection of possible models using some specified information source means that 
	a model $x$ that is known to be in the collection can be determined from that information for $x$. 
In \cite{Solis-Lemus2016} the  models 
correspond to certain networks, and the information is expected $CF$s arising from them under the NMSC.
As stated in \cite{Solis-Lemus2016}, 
\begin{quote}\dots for networks with $n\ge  4$ taxa, we restrict our focus to the case when $N'$ is the network topology obtained
	 from $N$ by removing a single hybrid edge of interest.\dots The presence of the hybridization of interest can be detected if the quartet $CF$s from
$N'$ cannot all be equal to the quartet $CF$s from $N$ simultaneously. 
\end{quote}
In other words, that work focused on distinguishability of the 2-element set $\{N,N'\}$ using quartet  $CF$s. 
While investigating this question was certainly a strong first contribution to the broader question of identifiability of level-1 network topology from $CF$s,  addressing the full question  would require showing distinguishability of the set of \emph{all} possible level-1 networks on a taxon set $X$, including networks with completely unrelated topologies. More fully, one would need to show this for sets $X$ of arbitrary size. The approach taken to prove detectability by \cite{Solis-Lemus2016}, however, depends on equating formulas for $CF$s in terms of numerical parameters on the two networks and solving the system, and it is unclear how this could be applied to the full identifiability question. 

Note that while \cite{Solis-Lemus2016} states its detectability results extend to level-1 networks with many cycles (where multiple hybrid edges may be removed to get the set of distinguishable networks), a justification for this claim was only given in \cite{Solis-Lemus2020}, with Lemma 3 of that work being key. Unfortunately, that lemma is incorrect as stated. (See \cref{sec:Lem3counter} of this appendix for a counterexample.) While it might be possible to obtain a correct justification using similar ideas, doing so seems unnecessary given the results of  \cite{Banos2019}, which we describe next.

After presenting a detailed study of all level-1 networks on 4 taxa and their $CF$s,  \cite{Banos2019} used combinatorial arguments 
to show that information about larger networks could be obtained through induced quartet networks, and hence  $CF$s. In particular, 
the topological identifiability result stated in this paper as \cref{thm:topExcept4} was established. This work also clarified several points 
that were either implicit or unaddressed in \cite{Solis-Lemus2016}: First, the semidirected network was formally defined, highlighting 
that network structure above the LSA needed to be excised. Second, identifiability of large cycles (more than 4 edges) was explicitly 
addressed (although \cite{Solis-Lemus2020} later provided an argument for detectability).
Third, identifiability for  networks with multiple cycles was shown. However, because it considered only one $CF$ at a time to deduce 
information about an induced quartet network, without exploring whether additional information might be found in the relationship of $CF$s 
for overlapping sets of taxa, it  did not obtain the strongest possible result. As an example, while \cite{Solis-Lemus2016} showed the 
distinguishability of the hybrid node of a 4-cycle in certain sufficiently large networks, \cite{Banos2019} left all 4-cycles undirected in its 
network identifiability result. 

While \cite{Solis-Lemus2016} had shown in some cases 3-cycles were detectable, the main result of \cite{Banos2019} omits 3-cycles in its identifiability result. It did contain a theorem, though, suggesting a $3_2$-cycle on a 4-taxon network may be identifiable for some parameter values. In both works, then, questions about 3-cycle identifiability were left open.

In short, the general question, in arbitrary level-1 networks, of topological identifiability of both directed and undirected  3-cycles and of directed 4-cycles, the focus of  \cref{ssec:3top,ssec:4top} of this work, remained. 

\medskip

Although not considered by \cite{Banos2019}, the study of identifiability of numerical parameters was also initiated by \cite{Solis-Lemus2016}. Using calculations of Gr\"obner bases of ideals  of polynomial relationships between $CF$s, arguments in \cite{Solis-Lemus2016} (with a technical matter corrected in \cite{Solis-Lemus2020}) investigated whether the dimension of the associated algebraic variety allowed for all numerical parameters to be identifiable. If this dimension is less than the number of parameters, then not all parameters can be identified, though it is possible that a subset are. If the dimension and number of parameters are equal, then the parameterization map must be generically finite-to-1. For specific networks  \cite{Solis-Lemus2016} showed whether or not the parameterization was finite-to-1, but passing to networks with more than 1 cycle again depended on the faulty Lemma 3 of \cite{Solis-Lemus2020}.
Moreover, the calculations in \cite{Solis-Lemus2016} are focused on numerical  parameters associated to cycles (cycle edge lengths, hybridization parameters, and possibly lengths of  edges adjacent to a cycle). Although the arguments do not seem to cover edges not adjacent to any cycles, this omission is easily overcome (for instance as in our \cref {prop:defQid}). On the other hand, it is unclear how the computations can be applied to answer whether the edges adjacent to hybrid edges in two different cycles have identifiable lengths.
Moreover, when a dimension computation leads to a valid conclusion that a full set of numerical parameters cannot be identified, it still is possible that some subset of the parameters could be. This issue was not explored.

Finally, establishing that there can only be finitely many choices of parameters that yield the $CF$s on a network is 
	a \emph{local identifiability} result, and \cite{Solis-Lemus2016} emphasized it leaves open the question of \emph{global identifiability}: 
Does this finite set actually contain only one such choice?  
This cannot be settled by dimension computations of the sort in \cite{Solis-Lemus2016}. There are in fact simple statistical models (outside of phylogenetics) 
where parameterizations are finite-to-1 but not 1-1 in surprising ways \cite{ARSV2015}. Investigating  global identifiability is thus highly desirable.

Several questions of identifiability of numerical parameters questions thus remained open: identifiability of certain edge lengths, 
especially in multicycle networks; identifiability of subsets of parameters when a full set was not identifiable;  and the broad question of global identifiability. 
These questions are the focus of \Cref{sec:numID}.

\smallskip

As the manuscript for this work was being completed, the preprint \cite{TileySolisLemus2023} appeared. This includes work on distinguishability 
of several different 6-taxon networks with 4-cycles and several cherries, but not on the full network identifiability questions even for 6-taxon networks 
with 4-cycles. It does however include simulation work to investigate \emph{practical identifiability}, the extent to which with finite data sets one can 
infer such cycles, using several pseudolikelihood methods.

\subsection{An example}  \label{sec:Lem3counter} Consider the two networks shown in \cref{fig:counterex}, 
where the right network is obtained from the left by removing a hybrid edge in the top 4-cycle. This is an instance 
of the construction implied in the statement of Lemma 3 of \cite{Solis-Lemus2020}, which is used to justify
analyzing $CF$s only of level-1 networks with a single cycle to understand those with multiple cycles. 
We show here that the statement of Lemma 3 is not correct for this pair of networks. 

We first note that this lemma states that the $CF$'s for the two networks which depend on parameters associated to the cycle are equal. This is not strictly true as stated:
Focusing on the lower cycle of the left network, for instance, $CF_{ab|ce}$ depends on $\gamma, x_1$ as well as on $\delta$  for the left network, but for the right network has no dependence on $\delta$. 
Other interpretations of this statement are that the described set of $CF$s for the 
two networks produce the same values as parameters vary, or that the varieties defined by the parameterized $CF$s are the same. As this last interpretation 
is the broadest, we show it is also false.

Consider $CF_{ac|df }$ and $CF_{ac|ef}$ for the two networks. For the left network
\begin{align*}
CF_{ac|df}&=\frac 13 \ell y_1( \gamma+(1-\gamma)x_2),\\
CF_{ac|ef}&= \frac 13 \ell  ( \gamma+(1-\gamma)x_2)(\delta+(1-\delta)y_1),
\end{align*}
which are not equal for generic parameter values. However, for the right network
$$CF_{ac|df }= CF_{ac|ef}  = \frac 13 \ell y_1( \gamma+(1-\gamma)x_2)$$
since $d,e$ can be exchanged in the graph because they form a cherry. Thus the variety for the right network has a defining equation that does not hold for the left.

This example illustrates an important point that in passing to smaller induced networks a cycle other than the one of immediate focus may have an impact on  $CF$s. In particular, a more detailed treatment of networks with multiple cycles, such as our work in \Cref{ssec:4top}, seems to be a necessary part of addressing the full identifiability question.

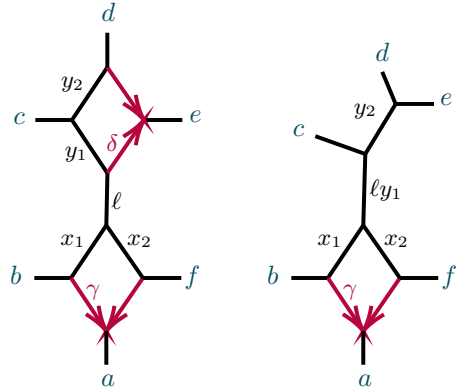
\begin{figure}
	\centering
		\tikzset{every picture/.style={line width=0.75pt}} 

\begin{tikzpicture}[x=0.75pt,y=0.75pt,yscale=-0.9,xscale=0.9]

\draw [line width=1.5]    (87.05,167.67) -- (107.55,167.67) ;
\draw [line width=1.5]    (107.55,167.67) -- (127.38,138.13) ;
\draw [color={rgb, 255:red, 182; green, 5; blue, 60 }  ,draw opacity=1 ][line width=1.5]    (107.55,167.67) -- (125.72,194.93) ;
\draw [shift={(127.38,197.42)}, rotate = 236.31] [color={rgb, 255:red, 182; green, 5; blue, 60 }  ,draw opacity=1 ][line width=1.5]    (14.21,-4.28) .. controls (9.04,-1.82) and (4.3,-0.39) .. (0,0) .. controls (4.3,0.39) and (9.04,1.82) .. (14.21,4.28)   ;
\draw [color={rgb, 255:red, 182; green, 5; blue, 60 }  ,draw opacity=1 ][line width=1.5]    (147.9,167.67) -- (129.09,194.95) ;
\draw [shift={(127.38,197.42)}, rotate = 304.59] [color={rgb, 255:red, 182; green, 5; blue, 60 }  ,draw opacity=1 ][line width=1.5]    (14.21,-4.28) .. controls (9.04,-1.82) and (4.3,-0.39) .. (0,0) .. controls (4.3,0.39) and (9.04,1.82) .. (14.21,4.28)   ;
\draw [line width=1.5]    (127.38,197.42) -- (127.38,216.33) ;
\draw [line width=1.5]    (127.38,138.13) -- (147.9,167.67) ;
\draw [line width=1.5]    (128,108.69) -- (127.38,138.13) ;
\draw [color={rgb, 255:red, 0; green, 0; blue, 0 }  ,draw opacity=1 ][line width=1.5]    (147.9,167.67) -- (169.4,167.67) ;
\draw [color={rgb, 255:red, 0; green, 0; blue, 0 }  ,draw opacity=1 ][line width=1.5]    (108.11,79.07) -- (127.94,49.52) ;
\draw [color={rgb, 255:red, 0; green, 0; blue, 0 }  ,draw opacity=1 ][line width=1.5]    (108.11,79.07) -- (127.94,108.82) ;
\draw [color={rgb, 255:red, 182; green, 5; blue, 60 }  ,draw opacity=1 ][line width=1.5]    (146.75,81.54) -- (127.94,108.82) ;
\draw [shift={(148.46,79.07)}, rotate = 124.59] [color={rgb, 255:red, 182; green, 5; blue, 60 }  ,draw opacity=1 ][line width=1.5]    (14.21,-4.28) .. controls (9.04,-1.82) and (4.3,-0.39) .. (0,0) .. controls (4.3,0.39) and (9.04,1.82) .. (14.21,4.28)   ;
\draw [color={rgb, 255:red, 0; green, 0; blue, 0 }  ,draw opacity=1 ][line width=1.5]    (148.46,79.07) -- (169.96,79.07) ;
\draw [color={rgb, 255:red, 182; green, 5; blue, 60 }  ,draw opacity=1 ][line width=1.5]    (127.94,49.52) -- (146.74,76.6) ;
\draw [shift={(148.46,79.07)}, rotate = 235.22] [color={rgb, 255:red, 182; green, 5; blue, 60 }  ,draw opacity=1 ][line width=1.5]    (14.21,-4.28) .. controls (9.04,-1.82) and (4.3,-0.39) .. (0,0) .. controls (4.3,0.39) and (9.04,1.82) .. (14.21,4.28)   ;
\draw [color={rgb, 255:red, 0; green, 0; blue, 0 }  ,draw opacity=1 ][line width=1.5]    (127.94,30.12) -- (127.94,49.52) ;
\draw [line width=1.5]    (87.61,79.07) -- (108.11,79.07) ;
\draw [line width=1.5]    (231.05,167.67) -- (251.55,167.67) ;
\draw [line width=1.5]    (251.55,167.67) -- (271.38,138.13) ;
\draw [color={rgb, 255:red, 182; green, 5; blue, 60 }  ,draw opacity=1 ][line width=1.5]    (251.55,167.67) -- (269.72,194.93) ;
\draw [shift={(271.38,197.42)}, rotate = 236.31] [color={rgb, 255:red, 182; green, 5; blue, 60 }  ,draw opacity=1 ][line width=1.5]    (14.21,-4.28) .. controls (9.04,-1.82) and (4.3,-0.39) .. (0,0) .. controls (4.3,0.39) and (9.04,1.82) .. (14.21,4.28)   ;
\draw [color={rgb, 255:red, 182; green, 5; blue, 60 }  ,draw opacity=1 ][line width=1.5]    (291.9,167.67) -- (273.09,194.95) ;
\draw [shift={(271.38,197.42)}, rotate = 304.59] [color={rgb, 255:red, 182; green, 5; blue, 60 }  ,draw opacity=1 ][line width=1.5]    (14.21,-4.28) .. controls (9.04,-1.82) and (4.3,-0.39) .. (0,0) .. controls (4.3,0.39) and (9.04,1.82) .. (14.21,4.28)   ;
\draw [line width=1.5]    (271.38,197.42) -- (271.38,216.33) ;
\draw [line width=1.5]    (271.38,138.13) -- (291.9,167.67) ;
\draw [line width=1.5]    (272.5,98.13) -- (271.38,138.13) ;
\draw [color={rgb, 255:red, 0; green, 0; blue, 0 }  ,draw opacity=1 ][line width=1.5]    (291.9,167.67) -- (313.4,167.67) ;
\draw [color={rgb, 255:red, 0; green, 0; blue, 0 }  ,draw opacity=1 ][line width=1.5]    (272.5,98.13) -- (289.5,70.13) ;
\draw [color={rgb, 255:red, 0; green, 0; blue, 0 }  ,draw opacity=1 ][line width=1.5]    (244.5,88.13) -- (272.5,98.13) ;
\draw [color={rgb, 255:red, 0; green, 0; blue, 0 }  ,draw opacity=1 ][line width=1.5]    (289.5,70.07) -- (311,70.07) ;
\draw [color={rgb, 255:red, 0; green, 0; blue, 0 }  ,draw opacity=1 ][line width=1.5]    (282,52.07) -- (289.5,70.07) ;

\draw (172.36,72.14) node [anchor=north west][inner sep=0.75pt]  [font=\normalsize,color={rgb, 255:red, 6; green, 79; blue, 97 }  ,opacity=1 ]  {$e$};
\draw (122.88,220.08) node [anchor=north west][inner sep=0.75pt]  [font=\normalsize,color={rgb, 255:red, 6; green, 79; blue, 97 }  ,opacity=1 ]  {$a$};
\draw (73.68,72.98) node [anchor=north west][inner sep=0.75pt]  [font=\normalsize,color={rgb, 255:red, 6; green, 79; blue, 97 }  ,opacity=1 ]  {$c$};
\draw (114.5,169.89) node [anchor=north west][inner sep=0.75pt]  [font=\small,color={rgb, 255:red, 182; green, 5; blue, 60 }  ,opacity=1 ]  {$\gamma $};
\draw (100.05,141.67) node [anchor=north west][inner sep=0.75pt]  [font=\small]  {$x_{1}$};
\draw (169.48,159.08) node [anchor=north west][inner sep=0.75pt]  [font=\normalsize,color={rgb, 255:red, 6; green, 79; blue, 97 }  ,opacity=1 ]  {$f$};
\draw (137.38,141.67) node [anchor=north west][inner sep=0.75pt]  [font=\small]  {$x_{2}$};
\draw (71.79,159.14) node [anchor=north west][inner sep=0.75pt]  [font=\normalsize,color={rgb, 255:red, 6; green, 79; blue, 97 }  ,opacity=1 ]  {$b$};
\draw (128.68,118.32) node [anchor=north west][inner sep=0.75pt]  [font=\small,color={rgb, 255:red, 0; green, 0; blue, 0 }  ,opacity=1 ]  {$\ell $};
\draw (122.98,11.84) node [anchor=north west][inner sep=0.75pt]  [font=\normalsize,color={rgb, 255:red, 6; green, 79; blue, 97 }  ,opacity=1 ]  {$d$};
\draw (125.84,85.39) node [anchor=north west][inner sep=0.75pt]  [font=\small,color={rgb, 255:red, 182; green, 5; blue, 60 }  ,opacity=1 ]  {$\delta $};
\draw (102.39,92.67) node [anchor=north west][inner sep=0.75pt]  [font=\small]  {$y_{1}$};
\draw (100.39,53.17) node [anchor=north west][inner sep=0.75pt]  [font=\small]  {$y_{2}$};
\draw (312.86,62.64) node [anchor=north west][inner sep=0.75pt]  [font=\normalsize,color={rgb, 255:red, 6; green, 79; blue, 97 }  ,opacity=1 ]  {$e$};
\draw (266.88,220.08) node [anchor=north west][inner sep=0.75pt]  [font=\normalsize,color={rgb, 255:red, 6; green, 79; blue, 97 }  ,opacity=1 ]  {$a$};
\draw (230.18,79.98) node [anchor=north west][inner sep=0.75pt]  [font=\normalsize,color={rgb, 255:red, 6; green, 79; blue, 97 }  ,opacity=1 ]  {$c$};
\draw (258.5,169.89) node [anchor=north west][inner sep=0.75pt]  [font=\small,color={rgb, 255:red, 182; green, 5; blue, 60 }  ,opacity=1 ]  {$\gamma $};
\draw (244.05,141.67) node [anchor=north west][inner sep=0.75pt]  [font=\small]  {$x_{1}$};
\draw (313.48,159.08) node [anchor=north west][inner sep=0.75pt]  [font=\normalsize,color={rgb, 255:red, 6; green, 79; blue, 97 }  ,opacity=1 ]  {$f$};
\draw (281.38,141.67) node [anchor=north west][inner sep=0.75pt]  [font=\small]  {$x_{2}$};
\draw (215.79,159.14) node [anchor=north west][inner sep=0.75pt]  [font=\normalsize,color={rgb, 255:red, 6; green, 79; blue, 97 }  ,opacity=1 ]  {$b$};
\draw (273.18,110.82) node [anchor=north west][inner sep=0.75pt]  [font=\small,color={rgb, 255:red, 0; green, 0; blue, 0 }  ,opacity=1 ]  {$\ell y_{1}$};
\draw (276.48,33.84) node [anchor=north west][inner sep=0.75pt]  [font=\normalsize,color={rgb, 255:red, 6; green, 79; blue, 97 }  ,opacity=1 ]  {$d$};
\draw (264.39,68.17) node [anchor=north west][inner sep=0.75pt]  [font=\small]  {$y_{2}$};

\end{tikzpicture}

\vspace*{-1cm}
		\caption{(L) A network with a 4-cycle of interest at bottom, and (R) a network with a single cycle obtained by removing a hybrid edge from each pair not in the cycle of interest. }\label{fig:counterex}
 	\end{figure}

\section{Propositions from Computations}\label{app:props}

\subsection{Ideals of $CF$ invariants}\label{sec:ideals}

Several observations simplify investigating the ideals $\mathcal I(N)$ described in \Cref{sec:NMSC}. First, since the entries of any vector concordance factor add to 1, 
at most two of the 3 entries of a $\overline{CF}$ are linearly independent. We consider the polynomials expressing that $\overline{CF}$ 
entries add to 1 to be \emph{trivial invariants}, as they hold for all network topologies. If only 4 taxa $a,b,c,d$ are considered, then the  trivial invariant is
$$ CF_{ab|cd} +CF_{ac|bd} + CF_{ad|bc}-1,$$
while for $n$ taxa there are $n\choose {4}$ such invariants, one for each possible quartet.

Second, if there is a cut edge on an induced quartet network separating two of the taxa, $a,b$ from the others $c,d$, then the $CF$s 
associated with the discordant topologies $ac|bd$ and $ad|bc$ must be equal. This was shown in the level-1 case in \cite{Banos2019} 
and for general networks in \cite{AllmanEtAl2022}. The invariant expressing this is
$$CF_{ac|bd}-CF_{ad|bc}.$$ 
We call such polynomials  \emph{cut invariants}. 

Third, when a network has one or more cherries (2 taxa joined by pendant edges to a common node), we will label the taxa in each cherry 
by the same subscripted letter, such as $a_1,a_2$. (See for instance \cref{fig:5_3cycle,fig:222nets}.) In $CF$s involving such taxa, 
we may then suppress the subscripts, since under the NMSC model on a suitably rooted version of the semidirected network 
the taxa $a_1$ and $a_2$ are exchangeable, giving the same $CF$ values when they are interchanged. Thus, for example, on a 
network with a  cherry formed by $a_1,a_2$,
$$CF_{ab|cd}=CF_{a _1b|cd}=CF_{a _2b|cd},$$
and
$$CF_{ab|ac}=CF_{a_1b|a_2c }=CF_{a_2b|a_1c}.$$
We call these \emph{exchange invariants}, but note that some of these these are also cut invariants. For computations, our notational simplification of surpressing subscripts allows for their omission.

\begin{defi}
For a topological phylogenetic network $N,$ the ideal generated by all trivial, cut, and exchange invariants of $N$ is denoted $\mathcal J( N)$.
\end{defi}

Note that $\mathcal J( N)$ depends on the topology  of the  network $N$ due to the cut and exchange invariants. However, it only depends on the undirected topology. Also, while $\mathcal J( N)\subseteq \mathcal I(N)$ these ideals are typically not equal. Since the generators of $\mathcal J(N)$ are linear, and simple to enumerate, we use them
for removing some $CF$s from calculations of $\mathcal I(N)$, and for stating results on ideal generators more succinctly. 

\subsection{Propositions from 3-cycle computations}

\begin{proposition}\label{prop:T5}
Let $T_5=((a_1,a_2),c,(b_1,b_2))$ be a $5$-taxon tree, as shown in \cref{fig:5_3cycle}(L). Then 
\begin{description}
\item[(a)]	The quartet concordance factors of $T_5$ are		
{\tiny
\begin{align*}
		 		\overline{CF}_{aabc} &= \left(1 -\frac{2}{3}\ell_1,\ \frac{1}{3}\ell_1,\ \frac{1}{3}\ell_1 \right), \\
		 		\overline{CF}_{acbb} &= \left(1 -\frac{2}{3}\ell_2,\ \frac{1}{3}\ell_2,\ \frac{1}{3}\ell_2 \right), \\
		 		\overline{CF}_{aabb} &= \left(1 -\frac{2}{3}\ell_2\ell_1,\ \frac{1}{3}\ell_2\ell_1,\  \frac{1}{3}\ell_2\ell_1 \right).
\end{align*}
}
\item[(b)] The ideal defining $\mathcal{V}_{T_5}\subset\mathbb{C}^{9}$ is
	 			$\mathcal I(T_5) = \mathcal J(T_5)  +  \langle f_{abc} \rangle,$ where 			
		{\tiny \begin{equation*}
		 f_{abc}= 3CF_{ab|ac}CF_{ba|bc} - CF_{ab|ab}.\end{equation*}}

        The variety $\mathcal{V}_{T_5}$ has dimension $2$.
\item[(c)] The numerical parameters $\ell_1$ and $\ell_2$ can be determined from the quartet $CF$s  by
	 		{\tiny \begin{equation*}
	 			\ell_1 = 3CF_{ab|ac}, \qquad \ell_2 = 3CF_{ab|cb}.
	 		\end{equation*}}
\end{description}
\end{proposition}

\begin{proposition}\label{prop:N5-3-1}
 			Let $N_{5-3_1}$ be a $5$-taxon level-$1$ network with a central $3_1$-cycle, as in \cref{fig:5_3cycle}(C). Then 
\begin{enumerate}
\item[(a)] The quartet concordance factors of $N_{5-3_1}$ are
 {\tiny				
 \begin{align*}
 			\overline{CF}_{aabc} &= \left(1 - \frac{2}{3}\ell_1(\gamma + x-\gamma x),\ \frac{1}{3}\ell_1(\gamma + x-\gamma x),\ \frac{1}{3}\ell_1(\gamma + x-\gamma x) \right), \\
 			\overline{CF}_{acbb} &= \left( 1  - \frac{2}{3}\ell_2(1-\gamma  + \gamma x),\ \frac{1}{3}\ell_2(1-\gamma + \gamma x),\ \frac{1}{3}\ell_2(1-\gamma + \gamma x)  \right), \\
 			\overline{CF}_{aabb} &= \left(1 -\frac{2}{3}\ell_2\ell_1x,\ \frac{1}{3}\ell_2\ell_1x,\ \frac{1}{3}\ell_2\ell_1x \right).
\end{align*}
 }
\item[(b)] The ideal  defining $\mathcal{V}({N_{5-3_1}})\subset\mathbb{C}^{9}$ is $\mathcal I(N_{5-3_1})=\mathcal J(N_{5-3_1})=\mathcal J(T_5)$, The variety $\mathcal{V}_{N_{5-3_1}}$ has dimension $3$.

\item[(c)] 			The numerical parameters $\gamma$, $x$, $\ell_1$ and $\ell_2$ cannot be determined from the quartet $CF$s. They can be determined with one degree of freedom, for instance if $\ell_1$ is known:
{\tiny
 			\begin{equation*}
 			\begin{gathered}
 				\ell_2 = \frac{3\ell_1CF_{ab|cb} - 3CF_{ab|ab}}{\ell_1 - 3CF_{ab|ac}}, \quad		
 				x = \frac{\ell_1CF_{ab|ab} - 3CF_{ab|ab}CF_{ab|ac}}{\ell_1^2CF_{ab|cb} - \ell_1CF_{ab|ab}}, \\
 				\gamma = \frac{3\ell_1CF_{ab|cb}CF_{ab|ac} - \ell_1CF_{ab|ab}}{\ell_1^2CF_{ab|cb} - 2\ell_1CF_{ab|ab}  +  3CF_{ab|ab}CF_{ab|ac}}.
 			\end{gathered}
 			\end{equation*}	
}
\end{enumerate}	
\end{proposition}

\begin{proposition}\label{prop:N5-3-2}
	 		Let $N_{5-3_2}$ be a $5$-taxon level-$1$ semidirected network with a central $3_2$-cycle, as in \cref{fig:5_3cycle}(R). Then, 
\begin{enumerate}
\item[(a)]			
			The quartet concordance factors of $N_{5-3_2}$ are
{\tiny
			\begin{align*}
			\overline{CF}_{aabc} =& \left( 1-\frac{2}{3}\ell_1u,\ \frac{1}{3}\ell_1u,\ \frac{1}{3}\ell_1u \right), \text{ with } u = \gamma^2h_1 + \gamma(1-\gamma)(3-x)+ (1-\gamma)^2h_2, \\
			\overline{CF}_{acbb} =& \left(1-\frac{2}{3}\ell_2v,\ \frac{1}{3}\ell_2v,\ \frac{1}{3}\ell_2v \right), \text{ with } v=\gamma + (1-\gamma) x ,\\
			\overline{CF}_{aabb} =& \left( 1-\frac{2}{3}\ell_1\ell_2w,\ \frac{1}{3}\ell_1\ell_2w,\ \frac{1}{3}\ell_1\ell_2 w \right), \text{ with } w = \gamma^2h_1 + 2\gamma(1-\gamma)+ (1-\gamma)^2 h_2x .
			\end{align*}
}	 		
\item[(b)]
	 		The ideal defining $\mathcal{V}_{N_{5-3_2}}\subset\mathbb{C}^{9}$ is $\mathcal I(N_{5-3_2})=\mathcal J(N_{5-3_2})=\mathcal J(T_5)$. The variety $\mathcal{V}_{N_{5-3_2}}$ has dimension $3$.

	\item[(c)] The $6$ numerical parameters $\gamma$, $h_1$, $h_2$, $x$, $\ell_1$ and $\ell_2$ cannot be determined from the quartet $CF$s. They can be determined with three degrees of freedom. If $\gamma$, $\ell_1$ and $\ell_2$ are known then:
	{\tiny
	\begin{equation*}
		\begin{gathered}
			x = \frac{\gamma\ell_2 - 3CF_{ab|bc}}{\gamma\ell_2 - \ell_2}, \quad		
			h_2 = \frac{\gamma\ell_1\ell_2 - 3\gamma\ell_1CF_{ab|bc} - 3\ell_2CF_{ab|ac} + 3CF_{ab|ab}}{\ell_1\left(\ell_2 - 3CF_{ab|bc}\right)(\gamma - 1)}, \\
			h_1 = \frac 1{\gamma^2\ell_1\ell_2\left (\ell_2-3CF_{ab|bc}\right )}		
			\Big[ 3\left(\gamma\ell_1\left (-3CF_{ab|bc}-\gamma\ell_2+3\ell_2\right )-3\ell_2CF_{ab|ac}\right)CF_{ab|bc} + 3\gamma\ell_2^2CF_{ab|ac} +\\
			 \ \hfill 3\ell_2CF_{ab|ab}(1-\gamma) - \gamma\ell_1\ell_2^2(2-\gamma) \Big].
		\end{gathered}
		\end{equation*}	
	}
\end{enumerate}
\end{proposition}

\begin{proposition}\label{prop:T6}
	Let $T_6$ be the $6$-taxon tree with a central node and 3 cherries, as in \cref{fig:222nets} (R). Then, 
\begin{enumerate}
\item[(a)]			
	   The quartet concordance factors of $T_6$ are
{\tiny
	   \begin{align*}
	        \overline{CF}_{aabb}=&\left( 1-\frac{2}{3}\ell_2\ell_1,\frac{1}{3}\ell_2\ell_1,\frac{1}{3}\ell_2\ell_1\right), \\
		\overline{CF}_{aabc}=&\left(1- \frac{2}{3}\ell_1 , \frac{1}{3}\ell_1, \frac{1}{3}\ell_1\right),\\
		\overline{CF}_{abbc}=&\left(\frac{1}{3}\ell_2,\frac{1}{3}\ell_2 ,1-\frac{2}{3}\ell_2\right), \\
		\overline{CF}_{abcc}=&\left(1-\frac{2}{3}\ell_3 , \frac{1}{3}\ell_3,\frac{1}{3}\ell_3 \right), \\
		\overline{CF}_{aacc}=&\left(1-\frac{2}{3}\ell_3\ell_1 ,\frac{1}{3}\ell_3\ell_1,\frac{1}{3}\ell_3\ell_1\right), \\
		\overline{CF}_{bbcc}=&\left(1- \frac{2}{3}\ell_3\ell_2 , \frac{1}{3}\ell_3\ell_2, \frac{1}{3}\ell_3\ell_2 \right).
            \end{align*}
}
 \item[(b)]

	The ideal defining $\mathcal{V}_{T_6}\subset\mathbb C^{18}$ is 
	
	$\mathcal I(T_6) = \mathcal J(T_6)  +  \langle f_{abc}, f_{bca}, f_{cab}\rangle$ where 
		{\tiny \begin{align*}
			&f_{abc} = 3CF_{ab|ac}CF_{ab|bc} - CF_{ab|ab},\\
			&f_{bca} = 3CF_{ab|bc}CF_{ac|bc} - CF_{bc|bc},\\
			&f_{cab} = 3CF_{ab|ac}CF_{ac|bc} - CF_{ac|ac}. \nonumber 
		\end{align*} }
	The variety $\mathcal{V}_{T_6}$ has dimension $3$.

\item[(c)]The numerical parameters  $\ell_1$, $\ell_2$, $\ell_3$ can be determined from the quartet $CF$s by: 
{\tiny \begin{equation*}
	\ell_1 = 3CF_{ab|ac},\ \ \
	\ell_2 = 3CF_{ab|bc},\ \ \ 	
	\ell_3 = 3CF_{ac|bc}.
\end{equation*} }

\end{enumerate}
\end{proposition}

\begin{proposition}\label{prop:N6-2-2-2}
		Let $N_{6-3_2}=N_a$ be a $6$-taxon level-$1$ semidirected network with a central $3_2$-cycle and 3 cherries, as shown in \cref{fig:222nets}(L). Then, 
\begin{enumerate}
\item[(a)]			
	   The quartet concordance factors of $N_{6-3_2}$ are
	   {\tiny
	   \begin{align*}
	    \overline{CF}_{aabb}=&( 1-2\alpha,\alpha,\alpha), \text{ with }
	            \alpha= \frac{1}{3}\ell_1\ell_2(\gamma^2h_1 +2\gamma(1-\gamma)+ (1-\gamma)^2h_2x ) ,\\
	\overline{CF}_{aabc}=&(1-2\beta ,\beta ,\beta ), \text{ with }
		\beta = \frac{1}{3}\ell_1(\gamma^2h_1 + \gamma(1-\gamma)(3-x)+ (1-\gamma)^2h_2 ),\\
	\overline{CF}_{abbc}=&(\gamma,\gamma ,1-2\gamma ), \text{ with }
		\gamma= \frac{1}{3}\ell_2(x + \gamma-\gamma x), \\
	\overline{CF}_{abcc}=&(1-2\delta , \delta,\delta ), \text{ with }
		\delta = \frac{1}{3}\ell_3(1-\gamma + \gamma x), \\
	\overline{CF}_{aacc}=&(1-2\epsilon ,\epsilon, \epsilon ), \text{ with }
		\epsilon  = \frac{1}{3}\ell_1\ell_3(\gamma^2h_1x +2\gamma(1-\gamma)+ (1-\gamma)^2h_2), \\
	\overline{CF}_{bbcc}=&(1-2\zeta ,\zeta ,\zeta ), \text{ with }
		\zeta= \frac{1}{3}x\ell_2\ell_3.
	    \end{align*}
	    }	    
 \item[(b)]

	The ideal  defining $\mathcal{V}_{N_{6-3_2}}\subset\mathbb{C}^{18}$ is $\mathcal I(N_{6-3_2})=\mathcal J(N_{6-3_2})=\mathcal J(T_6)$.
	The variety $\mathcal{V}_{N_{6-3_2}}$ has dimension {$6$}.
	
\item[(c)]The 7 numerical parameters $\gamma$, $x$, $\ell_1$, $\ell_2$, $\ell_3$, $h_1$, $h_2$ cannot be determined from the quartet $CF$s. They can be determined with one degree of freedom, for instance if $\ell_3$ is given:
{\tiny
\begin{equation*}
\begin{gathered}
	\gamma  = \frac{(\ell_3-3CF_{ac|bc})(\ell_3CF_{ab|bc}-CF_{bc|bc})}{(\ell_3^2CF_{ab|bc}-2\ell_3CF_{bc|bc} + 3CF_{bc|bc}CF_{ac|bc})}, \quad
	x = \frac{CF_{bc|bc}(\ell_3-3CF_{ac|bc})}{\ell_3(\ell_3CF_{ab|bc}-CF_{bc|bc})}, \\
	\ell_1  = \frac{-g}{(CF_{bc|bc}-3CF_{ab|bc}CF_{ac|bc})(\ell_3-3CF_{ac|bc})} ,\quad
	\ell_2 = \frac{3(\ell_3CF_{ab|bc}-CF_{bc|bc})}{(\ell_3-3CF_{ac|bc})}, \\
	h_1 = \frac{g_{h_1}(CF_{bc|bc}-3CF_{ab|bc}CF_{ac|bc})}{g(\ell_3-3CF_{ac|bc})(\ell_3CF_{ab|bc}-CF_{bc|bc})},\quad
	h_2 = \frac{g_{h_2}(\ell_3-3CF_{ac|bc})}{g\ell_3(CF_{bc|bc}-3CF_{ab|bc}CF_{ac|bc})}. \\
\end{gathered} 
\end{equation*}	

\noindent
where 
\begin{align*}
	g =\ & \ell_3^2CF_{ab|ab} - 3\ell_3^2CF_{ab|ac}CF_{ab|bc} - 3\ell_3CF_{ab|ab}CF_{ac|bc} + 3\ell_3CF_{ac|ac}CF_{ab|bc} \\
	& - 3CF_{bc|bc}CF_{ac|ac} + 9CF_{bc|bc}CF_{ab|ac}CF_{ac|bc}, \\
	g_{h_1} =\ & \ell_3^3CF_{ab|ab} - 6\ell_3^3CF_{ab|ac}CF_{ab|bc} + 6\ell_3^2CF_{bc|bc}CF_{ab|ac} - 3\ell_3^2CF_{ab|ab}CF_{ac|bc} \\
	&+ 6\ell_3^2CF_{ac|ac}CF_{ab|bc} - 9\ell_3CF_{bc|bc}CF_{ac|ac} + 9CF_{bc|bc}CF_{ac|ac}CF_{ac|bc}, \\
	g_{h_2} =\ & 2\ell_3^3CF_{ab|ab}CF_{ab|bc} - 6\ell_3^3CF_{ab|ac}CF_{ab|bc}^2 - 3\ell_3^2CF_{bc|bc}CF_{ab|ab} \\
	& + 12\ell_3^2CF_{bc|bc}CF_{ab|ac}CF_{ab|bc} - 6\ell_3^2CF_{ab|ab}CF_{ab|bc}CF_{ac|bc} + 3\ell_3^2CF_{ac|ac}CF_{ab|bc}^2 \\
	& - 6\ell_3CF_{bc|bc}^2CF_{ab|ac} + 12\ell_3CF_{bc|bc}CF_{ab|ab}CF_{ac|bc} - 6\ell_3CF_{bc|bc}CF_{ac|ac}CF_{ab|bc} \\
	& + 3CF_{bc|bc}^2CF_{ac|ac} - 9CF_{bc|bc}CF_{ab|ab}CF_{ac|bc}^2.
\end{align*}
}
\end{enumerate}
\end{proposition}

\begin{proposition}\label{prop:factorVar}
Let the parameterized variety $\mathcal V_D$ be as in \cref{prop:mapfactor}.
\begin{enumerate}
\item[(a)]
$\mathcal V_D$ is defined by the 0 ideal, and thus $\mathcal V_D =\mathbb C^6$.
\item[(b)]
From a point $(p_1,p_2,p_3,p_4,p_5,p_6)$ in the image of the parameterization of $\mathcal V_D$, parameters can be recovered with 1 degree of freedom. For instance, if $\ell_3$ is given,
{\tiny
	\begin{align*}
		\gamma &= -\frac{3p_4\ell_3^2+(-3p_1p_4+p_3-3p_4-1)\ell_3+p_1p_3-p_1-p_3+1}{3p_4\ell_3^2+(2p_3-2)\ell_3+p_1p_3-p_1-p_3+1},\\
		x &= -\frac{(p_3-1)(\ell_3+p_1-1)}{\ell_3(3p_4\ell_3+p_3-1)},\\
		\ell_1 &= \big [ (9p_4p_5+3p_6)\ell_3^2+(-3p_2p_4+3p_1p_6+3p_4-3p_6)\ell_3+3p_1p_3p_5-p_2p_3-3p_1p_5-3p_3p_5\\
		&\ \ \ \ \ \ \ \ \ +p_2+p_3+3p_5-1 \big]/{(3p_1p_4-p_3-3p_4+1)(\ell_3+p_1-1)},\\
			\ell_2 &= \frac{3p_4\ell_3-p_3+1}{\ell_3+p_1-1},\\
		h_1 &= \frac{f_1}{g_1}, \quad h_2 = \frac{f_2}{g_2},
	\end{align*}
	where
	\begin{align*}
		f_1 =& -(54p_4^2p_5-18p_4p_6)\ell_3^4+(-54p_1p_4^2p_5+9p_2p_4^2+36p_3p_4p_5-54p_4^2p_5-36p_1p_4p_6-9p_4^2-36p_4p_5\\
		&\ -9p_3p_6+36p_4p_6+9p_6)\ell_3^3 +(9p_1p_2p_4^2+36p_1p_3p_4p_5-18p_1^2p_4p_6+6p_2p_3p_4-9p_1p_4^2-9p_2p_4^2\\
		&\ +6p_3^2p_5-36p_1p_4p_5-36p_3p_4p_5 -21p_1p_3p_6+36p_1p_4p_6-6p_2p_4-6p_3p_4+9p_4^2-12p_3p_5\\
		&\ +36p_4p_5+21p_1p_6+21p_3p_6-18p_4p_6 +6p_4+6p_5-21p_6)\ell_3^2\\
		&\ +(6p_1p_2p_3p_4+6p_1p_3^2p_5-15p_1^2p_3p_6p_2p_3^2-6p_1p_2p_4-6p_1p_3p_4-6p_2p_3p_4-12p_1p_3p_5-6p_3^2p_5\\
		&\ +15p_1^2p_6+30p_1p_3p_6-2p_2p_3-p_3^2+6p_1p_4+6p_2p_4+6p_3p_4+6p_1p_5+12p_3p_5-30p_1p_6\\
		&\ -15p_3p_6+p_2+2p_3-6p_4-6p_5+15p_6-1)\ell_3 -3p_1^3p_3p_6+p_1p_2p_3^2+3p_1^3p_6+9p_1^2p_3p_6\\
		&\ -2p_1p_2p_3-p_1p_3^2-p_2p_3^2-9p_1^2p_6-9p_1p_3p_6 +p_1p_2+2p_1p_3+2p_2p_3+p_3^2+9p_1p_6+3p_3p_6\\
		&\ -p_1-p_2-2p_3-3p_6+1,\\
		g_1 =& \ell_3(3p_1p_4-p_3-3p_4+1) \big((9p_4p_5-3p_6)\ell_3^2+(3p_2p_4-3p_1p_6-3p_4+3p_6)\ell_3\\
		&\ -3p_1p_3p_5+p_2p_3+3p_1p_5+3p_3p_5-p_2-p_3-3p_5+1\big),\\	
		\end{align*}
		\begin{align*}
		f_2 =& +(54p_1p_4^2p_5+18p_3p_4p_5+54p_4^2p_5+9p_1p_4p_6-18p_4p_5-3p_3p_6-9p_4p_6+3p_6)\ell_3^3\\
		&\ +(-18p_1p_2p_4^2-18p_1p_3p_4p_5+9p_1^2p_4p_6+6p_2p_3p_4+18p_1p_4^2+18p_2p_4^2+6p_3^2p_5 \\ 
		&\ +18p_1p_4p_5+18p_3p_4p_5-3p_1p_3p_6-18p_1p_4p_6-6p_2p_4-6p_3p_4-18p_4^2-12p_3p_5\\
		&\ -18p_4p_5+3p_1p_6+3p_3p_6+9p_4p_6+6p_4+6p_5-3p_6)\ell_3^2 \\
		&\ +(-9p_1p_2p_3p_4+3p_2p_3^2+9p_1p_2p_4+9p_1p_3p_4+9p_2p_3p_4-6p_2p_3-3p_3^2-9p_1p_4\\
		&\ -9p_2p_4-9p_3p_4+3p_2+6p_3+9p_4-3)\ell_3 \\
		&\ -3p_1^2p_2p_3p_4+p_1p_2p_3^2+3p_1^2p_2p_4+3p_1^2p_3p_4+6p_1p_2p_3p_4-2p_1p_2p_3-p_1p_3^2-p_2p_3^2\\
		&\ -3p_1^2p_4-6p_1p_2p_4-6p_1p_3p_4-3p_2p_3p_4+p_1p_2+2p_1p_3+2p_2p_3+p_3^2+6p_1p_4+3p_2p_4 \\
		&\ +3p_3p_4-p_1-p_2-2p_3-3p_4+1,\\
		g_2 =& (\ell_3+p_1-1)(3p_4\ell_3+p_3-1)\cdot\big((9p_4p_5-3p_6)\ell_3^2 +(3p_2p_4-3p_1p_6-3p_4+3p_6)\ell_3\\
		&\ -3p_1p_3p_5+p_2p_3+3p_1p_5+3p_3p_5-p_2-p_3-3p_5+1\big).\\
	\end{align*}
}
\end{enumerate}
\end{proposition}


\subsection{Propositions from 4-cycle computations}\label{sec:props4cyc}

\begin{proposition}\label{prop:S}
 		Consider the $5$-taxon level-1 network $N_s$ with a single $4$-cycle shown in \cref{fig:4cyc}(L). Then 
\begin{enumerate}		
   \item[(a)] The quartet concordance factors of $N_s$ are
{\tiny
\begin{align*}
\overline{CF}_{abcd} &= \left(-\frac{1}{3}(2\gamma x_1  +  \gamma x_2 - 3\gamma - x_2), \frac{1}{3}(\gamma x_1 - \gamma x_2  +  x_2), 1  +  \frac{1}{3}(\gamma x_1  +  2\gamma x_2 - 3\gamma - 2x_2) \right),\\
\overline{CF}_{xyzw}&=\left(1-\frac{2}{3}\ell u_{xyzw}, \frac{1}{3} \ell u_{xyzw}, \frac{1}{3} \ell u_{xyzw} \right) \text{ for } xyzw=aabc,\ aabd,\ aacd,
\end{align*}
where 
\begin{align*}
u_{aabc} &= \gamma^2h_2x_2  +  \gamma^2h_1  +  \gamma^2x_1 - \gamma h_2x_2 - 3\gamma^2 - 2\gamma x_1  +  3\gamma  +  x_1, \\
 u_{aabd} &= \gamma^2x_1x_2  +  \gamma^2h_1  +  \gamma^2h_2 - 2\gamma x_1x_2 - 3\gamma^2 - \gamma h_2  +  x_1x_2  +  3\gamma, \\
 u_{aacd}&= \gamma^2h_1x_1  +  \gamma^2h_2  +  \gamma^2x_2 - 3\gamma^2 - \gamma h_2 - 2\gamma x_2  +  3\gamma  +  x_2.
\end{align*}
 }	
\item[(b)] The ideal defining $\mathcal{V}_{N_s}\subset\mathbb{C}^{12}$ is $\mathcal I(N_s) =\mathcal J(N_s)$.
	The variety $\mathcal V_{N_s}$ has dimension $5$.
 	
\item[(c)]
 		The edge probabilities $x_1$, $x_2$, $h_1$, $h_2$ and $\ell$ can be expressed in terms of the quartet $CF$s and $\gamma$:
 {\tiny		
     		\begin{equation*}	
 		\begin{gathered}
	 		x_1 = \frac{\gamma-CF_{ab|cd} + CF_{ac|bd}}{\gamma}, \quad x_2 = \frac{\gamma - CF_{ab|cd} - 2CF_{ac|bd}}{\gamma - 1}, \\ 
			h_1 = \frac{n_{h_1}}{d_h}, \qquad h_2 = \frac{n_{h_2}}{d_h,} \qquad \ell = \frac{n_{\ell}}{d_{\ell}},
 		\end{gathered}
		\end{equation*}
 		where 
		 \begin{align*}
			n_{h_1} = &\Big((CF_{ab|ac} - 2CF_{ab|ad})CF_{ab|cd}  +  (CF_{ab|cd} - 1)CF_{ac|ad} - (CF_{ab|ac}  +  CF_{ab|ad}\\
			 &\ - 2CF_{ac|ad})CF_{ac|bd}  +  CF_{ab|ad}\Big)\gamma^3 -  \Big((CF_{ab|ac}  +  CF_{ab|ad})CF_{ab|cd}^2 - 2(CF_{ab|ac}\\
			 &\  - 2CF_{ab|ad})CF_{ac|bd}^2 +  2(CF_{ab|ac} - 2CF_{ab|ad})CF_{ab|cd} -  3(CF_{ab|cd}^2 - CF_{ab|cd})CF_{ac|ad}\\
			 &\   +  ((CF_{ab|ac}  +  4CF_{ab|ad})CF_{ab|cd} - 6CF_{ab|cd}CF_{ac|ad} - 2CF_{ab|ac}  +  CF_{ab|ad})CF_{ac|bd}\Big)\gamma^2\\
			 &\  +   \Big(4CF_{ac|ad}CF_{ac|bd}^3  +  (2CF_{ab|ac}  +  CF_{ab|ad})CF_{ab|cd}^2 - 4(CF_{ab|ac} -  CF_{ab|ad})CF_{ac|bd}^2\\
		         &\ +  (CF_{ab|ac} - 3CF_{ab|ad})CF_{ab|cd} -  (CF_{ab|cd}^3  +  3CF_{ab|cd}^2 - 3CF_{ab|cd} - 1)CF_{ac|ad}\\
		         &\ +  (2(CF_{ab|ac}  +  2CF_{ab|ad})CF_{ab|cd} - 3(CF_{ab|cd}^2  +  2CF_{ab|cd}  +  1)CF_{ac|ad} - CF_{ab|ac}\\
		         &\   +  3CF_{ab|ad})CF_{ac|bd}   - CF_{ab|ad}\Big)\gamma -4CF_{ac|ad}CF_{ac|bd}^3 - CF_{ab|ac}CF_{ab|cd}^2  +  2CF_{ab|ac}CF_{ac|bd}^2\\
		         &\   +   CF_{ab|ad}CF_{ab|cd}  +  (CF_{ab|cd}^3  - CF_{ab|cd})CF_{ac|ad} - (CF_{ab|ac}CF_{ab|cd} - (3CF_{ab|cd}^2  +  1)CF_{ac|ad}\\
		         &\   +  CF_{ab|ad})CF_{ac|bd},  \\
		   n_{h_2} = &\Big((CF_{ab|ac} - 2CF_{ab|ad})CF_{ab|cd}  +  (CF_{ab|cd} - 1)CF_{ac|ad} - (CF_{ab|ac}  +  CF_{ab|ad}\\
		   &\  - 2CF_{ac|ad})CF_{ac|bd}  +  CF_{ab|ad}\Big)\gamma^3  +   \Big((3CF_{ab|ac} - CF_{ab|ad})CF_{ab|cd}^2 - (3CF_{ab|ac}  +  CF_{ab|ad}\\
		   &\  - 2CF_{ac|ad})CF_{ac|bd}^2 -  CF_{ab|ac}CF_{ab|cd} - (CF_{ab|cd}^2 - 1)CF_{ac|ad} + (2CF_{ab|ad}CF_{ab|cd}\\
		   &\  - (CF_{ab|cd}  +  3)CF_{ac|ad} +  CF_{ab|ac}  +  3CF_{ab|ad})CF_{ac|bd} - CF_{ab|ad}\Big)\gamma^2-  \Big(CF_{ab|ac}CF_{ab|cd}^3\\
		   &\  +  2CF_{ab|ac}CF_{ac|bd}^3  +  (3CF_{ab|ac}   -  2CF_{ab|ad})CF_{ab|cd}^2- (3CF_{ab|ac}CF_{ab|cd}  +  3CF_{ab|ac}\\
		   &\   +  2CF_{ab|ad} - 2CF_{ac|ad})CF_{ac|bd}^2  -  2CF_{ab|ad}CF_{ab|cd} - (CF_{ab|cd}^2 - CF_{ab|cd})CF_{ac|ad}\\
		   &\   +  (4CF_{ab|ad}CF_{ab|cd} - (CF_{ab|cd}  +  1)CF_{ac|ad} +  2CF_{ab|ad})CF_{ac|bd}\Big)\gamma +   CF_{ab|ac}CF_{ab|cd}^3\\
		   &\   +  2CF_{ab|ac}CF_{ac|bd}^3 - CF_{ab|ad}CF_{ab|cd}^2 +  2CF_{ab|ad}CF_{ab|cd}CF_{ac|bd} - (3CF_{ab|ac}CF_{ab|cd}\\
		   &\  +  CF_{ab|ad})CF_{ac|bd}^2, \\
		   d_h = &\Big((CF_{ab|ac} - 2CF_{ab|ad})CF_{ab|cd}  +  (CF_{ab|cd} - 1)CF_{ac|ad} - (CF_{ab|ac}  +  CF_{ab|ad}\\
		   &\  - 2CF_{ac|ad})CF_{ac|bd} +  CF_{ab|ad}\Big)\gamma^3+ \Big(CF_{ab|ad}CF_{ab|cd}^2 - 2CF_{ab|ad}CF_{ac|bd}^2\\
		   &\  - CF_{ab|ac}CF_{ab|cd}  +  (CF_{ab|ad}CF_{ab|cd}  +  CF_{ab|ac})CF_{ac|bd}\Big)\gamma^2, \\
\end{align*}
\begin{align*}	
			n_{\ell} =  &+   3\Big((CF_{ab|ac} - 2CF_{ab|ad})CF_{ab|cd}  +  (CF_{ab|cd} - 1)CF_{ac|ad} - (CF_{ab|ac}  +  CF_{ab|ad}\\
			&\   - 2CF_{ac|ad})CF_{ac|bd}+  CF_{ab|ad}\Big)\gamma^2+   3\Big(CF_{ab|ad}CF_{ab|cd}^2 - 2CF_{ab|ad}CF_{ac|bd}^2\\
			&\   - CF_{ab|ac}CF_{ab|cd}  +  (CF_{ab|ad}CF_{ab|cd} +  CF_{ab|ac})CF_{ac|bd}\Big)\gamma,\\
			d_{\ell} =  &-\Big(4CF_{ab|cd}^3 - 6CF_{ab|cd}CF_{ac|bd}^2 - 4CF_{ac|bd}^3 - 3CF_{ab|cd}^2\\
			&\   +  (6CF_{ab|cd}^2  +  3CF_{ab|cd}  +  1)CF_{ac|bd} - CF_{ab|cd}\Big)\gamma^2  \\
			 &\ +  \Big (CF_{ab|cd}^4 - 4(CF_{ab|cd}  +  1)CF_{ac|bd}^3  +  4CF_{ac|bd}^4  +  4CF_{ab|cd}^3\\
			 &\    - (3CF_{ab|cd}^2  +  6CF_{ab|cd}  +  1)CF_{ac|bd}^2 - 4CF_{ab|cd}^2  \\
			 &\ +   (2CF_{ab|cd}^3  +  6CF_{ab|cd}^2  +  5CF_{ab|cd}  +  1)CF_{ac|bd} - CF_{ab|cd}\Big)\gamma -CF_{ab|cd}^4  +  4CF_{ab|cd}CF_{ac|bd}^3  \\
			&\   - 4CF_{ac|bd}^4+  (3CF_{ab|cd}^2  +  1)CF_{ac|bd}^2  +  CF_{ab|cd}^2 - 2(CF_{ab|cd}^3  +  CF_{ab|cd})CF_{ac|bd}.
	\end{align*}
		}
\end{enumerate}
\end{proposition}

\begin{proposition}\label{prop:W}
 Let $N_w$ be the $5$-taxon level-1 network with a single $4$-cycle shown in \cref{fig:4cyc}(C). Then 
 \begin{enumerate}
\item[(a)]
The quartet concordance factors of $N_w$ are
{\tiny
 \begin{align*} 
\overline{ CF}_{abcd} &= \left(\frac{1}{3}(-2\gamma x_1-\gamma x_2 + 3\gamma  + x_2), \frac{1}{3}(\gamma x_1-\gamma x_2 + x_2), 1 + \frac{1}{3} (\gamma x_1 + 2\gamma x_2-3\gamma -2x_2) \right),\\
\overline{CF}_{xyzw} &= \left(1 - \frac{2}{3}\ell u_{xyzw}, \frac{1}{3}\ell u_{xyzw}, \frac{1}{3}\ell u_{xyzw} \right) \text{ for } xyzw=aabc,\ aabd,\ aacd,
\end{align*}
 where 
 \begin{align*}
 u_{aabc}&=x_1,\\
 u_{aabd}&= -\gamma x_1 + \gamma  + x_1,\\
 u_{aacd}&=-\gamma x_1x_2  +  x_1x_2 + \gamma.
 \end{align*}
} 
\item[(b)]
 		The ideal defining $\mathcal{V}_{N_w}\subset\mathbb{C}^{12}$ is
 	{\tiny 
		\begin{equation*}
 		\mathcal I(N_w) = \mathcal J(N_w)  +  \langle CF_{ab|cd}CF_{ab|ac} - CF_{ac|bd}CF_{ab|ac} - CF_{ab|ad}  +  CF_{ac|ad} \rangle, 
 		\end{equation*}
 	}
		and $\mathcal{V}_{N_w}$ has dimension $4$.

 \item[(c)]
 		The numerical parameters $x_1$, $x_2$, $\ell$, and $\gamma$ can be determined from the quartet $CF$s:
 		{\tiny
		\begin{align*}
 		x_1 =& \frac{3CF_{ab|ac}CF_{ac|bd} - CF_{ab|ac}  +  CF_{ab|ad} - CF_{ac|ad}}{CF_{ab|ac} - CF_{ab|ad}} ,\\
 		x_2 =&- \big [ CF_{ab|ad}CF_{ab|cd} - (CF_{ab|cd} - 2)CF_{ac|ad} - (CF_{ab|ad}  +  2CF_{ac|ad})CF_{ac|bd}  +  CF_{ab|ac}\\
		&\ \ - 2CF_{ab|ad} \big ]/{CF_{ab|ad}CF_{ab|cd}  +  2CF_{ab|ad}CF_{ac|bd} - CF_{ab|ac}}, \\
 		\ell =& \frac{3(CF_{ab|ac} - CF_{ab|ad})}{CF_{ab|cd}  +  2CF_{ac|bd} - 1}, \\
 		\gamma =& -\frac{CF_{ab|ad}CF_{ab|cd} - (3CF_{ab|ac} - 2CF_{ab|ad})CF_{ac|bd}  +  CF_{ab|ac} - 2CF_{ab|ad}  +  CF_{ac|ad}}{3CF_{ab|ac}CF_{ac|bd} - 2CF_{ab|ac}  +  2CF_{ab|ad} - CF_{ac|ad}}.
 		\end{align*}
		}
		\end{enumerate}
\end{proposition}

\begin{proposition}\label{prop:N}		
Let $N_n$ be the $5$-taxon level-1 network with a single $4$-cycle shown in \cref{fig:4cyc}(R). Then 		
\begin{enumerate}
\item[(a)]
The quartet concordance factors of $N_n$ are
{\tiny 
\begin{align*}
CF_{abcd} &= \left(\frac{1}{3}(-2\gamma x_1-\gamma x_2 + 3\gamma + x_2), \frac{1}{3}(\gamma x_1-\gamma x_2 + x_2), 1  +  \frac{1}{3}(\gamma x_1 + 2\gamma x_2-3\gamma -2x_2) \right),\\
CF_{xyzw} &= \left(1 - \frac{2}{3}\ell u_{xyzw}, \frac{1}{3}\ell u_{xyzw}, \frac{1}{3}\ell u_{xyzw} \right) \text{ for } xyzw=aabc,\ aabd,\ aacd,
\end{align*} 
where 
 \begin{equation*}
 u_{aabc} =\gamma  +  x_2 - \gamma x_2, \ \ \ 
u_{aabd}=1,\ \ \ 
u_{aacd} = 1 - \gamma  + \gamma x_1.
\end{equation*}
}		
\item[(b)]
	the ideal defining $\mathcal{V}_{N_n}\subset\mathbb{C}^{12}$ is
		{\tiny \begin{align*}
		\mathcal I(N_n) = \mathcal J(N_n)  &+  \langle  3CF_{ac|bd}CF_{ad|ab}  +  CF_{ad|ab} - CF_{ac|ab} - CF_{ac|ad}, \\
		&CF_{ab|cd}CF_{ac|ab}  +  CF_{ab|cd}CF_{ac|ad} - CF_{ac|bd}CF_{ac|ab}  +  2CF_{ac|bd}CF_{ac|ad} - CF_{ac|ab}, \\
		&3CF_{ab|cd}CF_{ad|ab} - 2CF_{ad|ab} - CF_{ac|ab}  +  2CF_{ac|ad}  \rangle, 
		\end{align*}
		}
				and $\mathcal{V}_{N_n}$ has dimension $3$. 

\item[(c)] The  parameter $\ell$ can be identified from the quartet $CF$s:
{\tiny $$\ell = 3CF_{a_1ba_2d},$$}
but $x_1$, $x_2$, and $\gamma$ cannot be. However, the quantities $x_1\gamma - \gamma$ and $x_2(\gamma-1)-\gamma$ can be determined from the quartet $CF$s, so if $\gamma$ is known
{\tiny
\begin{equation*}
x_1 = \frac{\gamma-CF_{abcd} + CF_{acbd}}{\gamma}, \quad x_2 = \frac{\gamma - CF_{abcd} - 2CF_{acbd}}{\gamma - 1}. 
\end{equation*}
}
\end{enumerate}
\end{proposition}

\begin{proposition}\label{prop:SN4cycle} 
 			Let $N_{sn}$ be a $6$-taxon level-$1$ network with a central $4$-cycle as shown in \cref{fig:WSandNS}(R). Then 
\begin{enumerate}
\item[(a)] The quartet concordance factors of $N_{sn}$ are
 {\tiny				
\begin{align*}
\overline{CF}_{aabc} &=(1-2u,u,u) ,\text{ with } u=\ell_1\left[  \frac 13 \gamma^2h_1+\frac 13 (1-\gamma)^2 h_2x_2+ \gamma(1-\gamma)\left(1-\frac13 x_1\right)  \right],\\
\overline{CF}_{aabd} &=(1-2u,u,u) ,\text{ with } u=\ell_1\left[  \frac 13 \gamma^2h_1+\frac 13 (1-\gamma)^2 h_2+ \gamma(1-\gamma)\left(1-\frac13 x_1x_2\right)  \right], \\
\overline{CF}_{aacc} &=(1-2u,u,u) ,\text{ with } u= \ell_1\ell_2\left [  \frac 13 \gamma^2h_1x_1+\frac 13 (1-\gamma)^2 h_2x_2+ \frac 23\gamma(1-\gamma) \right ], \\
\overline{CF}_{aacd} &= (1-2u,u,u) ,\text{ with } u= \ell_1\left [  \frac 13 \gamma^2h_1x_1+\frac 13 (1-\gamma)^2 h_2+ \gamma(1-\gamma)\left (1-\frac13 x_2\right ) \right ], \\
\overline{CF}_{abcc} &=  (u,u,1-2u), \text{ with } u=   \ell_2\frac 13 \left ( \gamma x_1 +(1-\gamma) \right ) , \\
\overline{CF}_{abcd} &= \left(  \gamma\left(1-\frac23 x_1 \right) +(1-\gamma)\frac 13 x_2, \frac 13 (\gamma x_1+(1-\gamma)x_2),\gamma \frac 13  x_1+(1-\gamma)\left ( 1-\frac 23 x_2\right )                   \right), \\
\overline{CF}_{accd} &=  (u,u,1-2u), \text{ with } u=\frac13\ell_2(\gamma +(1-\gamma)x_2), \\
\overline{CF}_{bccd} &=  (u,u,1-2u), \text{ with } u=\frac13\ell_2.
\end{align*}
 }
\item[(b)] The ideal  defining $\mathcal{V}({N_{sn}})\subset\mathbb{C}^{9}$ is
{\tiny \begin{align*}
	\mathcal I(N_{sn}) = \mathcal J(N_{sn})  +  \langle
		&C_{ab|cd}C_{ac|cd+}C_{ab|cd}C_{ac|bc}-C_{ac|cd}C_{ac|bd}-C_{ac|cd}+2C_{ac|bd}C_{ac|bc},\\
		&\ \ 3C_{bc|cd}C_{ac|bd}+C_{bc|cd}-C_{ac|cd}-C_{ac|bc},\\
		&\ \ 3C_{bc|cd}C_{ab|cd}-2C_{bc|cd}-C_{ac|cd}+2C_{ac|bc} \rangle.
	\end{align*}
}
The variety $\mathcal{V}_{N_{sn}}$ has dimension $7$.
\item[(c)] 	The numerical parameters $\gamma$, $h_1$, $h_2$, $x_1$, $x_2$, $\ell_1$ and $\ell_2$ can be determined from the quartet $CF$s. 
{\tiny
\begin{equation*}	
	\begin{gathered}
		\gamma = \frac{n_{\gamma}}{d_{\gamma}}, \qquad h_1 = \frac{n_{h_1}}{d_{h_1}}, \qquad h_2 = \frac{n_{h_2}}{d_{h_2}}, \qquad x_1 = \frac{n_{x_1}}{d_{x_1}}, \qquad x_2 = \frac{n_{x_2}}{d_{x_2}},\\
		\ell_1 = \frac{-3C_{ab|ac}C_{bc|cd}+3C_{ab|ad}C_{bc|cd}-3C_{bc|cd}C_{ac|ad}+C_{ac|ac}}{-C_{ab|cd}C_{ac|bc}-2CacbdC_{ac|bc}-C_{bc|cd}+C_{ac|cd}+C_{ac|bc}}, \qquad
		\ell_2 = 3C_{bc|cd}. 
	\end{gathered}
   \end{equation*}
	where 
    \begin{align*}
		n_{\gamma}  = &-C_{ac|ac}C_{ab|cd}^2-C_{ac|ac}C_{ab|cd}Cacbd+2C_{ac|ac}Cacbd^2-3C_{ab|ac}C_{ab|cd}C_{ac|bc}\\
			&-3C_{ab|cd}C_{ac|ad}C_{ac|bc}-6C_{ab|ac}CacbdC_{ac|bc}-6C_{ac|ad}CacbdC_{ac|bc}-3C_{ab|ad}C_{bc|cd}\\
			&-C_{ac|ac}C_{ab|cd}+3C_{bc|cd}C_{ac|ad}+3C_{ab|ac}C_{ac|cd}+3C_{ac|ad}C_{ac|cd}+C_{ac|ac}Cacbd\\
			&+3C_{ab|ad}C_{ac|bc}-3C_{ac|ad}C_{ac|bc},\\
        d_{\gamma}  = &-3C_{ab|ac}C_{bc|cd}-4C_{ac|ac}C_{ab|cd}+3C_{bc|cd}C_{ac|ad}+6C_{ab|ac}C_{ac|cd}-3C_{ab|ad}C_{ac|cd}\\
			&+3C_{ac|ad}C_{ac|cd}-2C_{ac|ac}Cacbd-3C_{ab|ac}C_{ac|bc}+3C_{ab|ad}C_{ac|bc}-6C_{ac|ad}C_{ac|bc}+2C_{ac|ac},\\
		n_{h_1} = &(-3C_{ac|ad}C_{ac|cd}+C_{ac|ac})x_1+6C_{ab|ac}C_{bc|cd}-3C_{ab|ad}C_{bc|cd}+3C_{ac|ac}C_{ab|cd}\\
		    &-6C_{ab|ac}C_{ac|cd}+3C_{ab|ad}C_{ac|cd}+3C_{ac|ac}Cacbd+3C_{ac|ad}C_{ac|bc}-3C_{ac|ac},\\
		d_{h_1} = &+3C_{ab|ac}C_{bc|cd}-3C_{ab|ad}C_{bc|cd}-C_{ac|ac}C_{ab|cd}+3C_{bc|cd}C_{ac|ad}+C_{ac|ac}Cacbd\\
			&-3C_{ab|ac}C_{ac|bc}+3C_{ab|ad}C_{ac|bc}-3C_{ac|ad}C_{ac|bc},\\
		n_{h_2} = &(3C_{ab|ac}C_{ac|bc}-C_{ac|ac})x_2+3C_{ab|ad}C_{bc|cd}+3C_{ac|ac}C_{ab|cd}-6C_{bc|cd}C_{ac|ad}\\
			&-3C_{ab|ac}C_{ac|cd}-3C_{ab|ad}C_{ac|bc}+6C_{ac|ad}C_{ac|bc},\\			
        d_{h_2} = &-3C_{ab|ac}C_{bc|cd}+3C_{ab|ad}C_{bc|cd}-C_{ac|ac}C_{ab|cd}-3C_{bc|cd}C_{ac|ad}+3C_{ab|ac}C_{ac|cd}\\
		    &-3C_{ab|ad}C_{ac|cd}+3C_{ac|ad}C_{ac|cd}-2C_{ac|ac}Cacbd+C_{ac|ac},\\
        n_{x_1} = &+3C_{ab|ac}C_{bc|cd}-3C_{ab|ad}C_{bc|cd}+3C_{ac|ac}C_{ab|cd}-3C_{ab|ac}C_{ac|cd}\\
		d_{x_1} = &-3C_{ab|ad}C_{bc|cd}-C_{ac|ac}C_{ab|cd}+3C_{bc|cd}C_{ac|ad}+3C_{ab|ac}C_{ac|cd}\\
			&+3C_{ac|ad}C_{ac|cd}-2C_{ac|ac}Cacbd-C_{ac|ac},\\
			\end{align*}
\begin{align*}	
		n_{x_2} = &+3C_{ab|ad}C_{bc|cd}+3C_{ac|ac}C_{ab|cd}-3C_{bc|cd}C_{ac|ad}-6C_{ab|ac}C_{ac|cd}\\
					&+3C_{ab|ad}C_{ac|cd}+3C_{ab|ac}C_{ac|bc}-3C_{ab|ad}C_{ac|bc}+6C_{ac|ad}C_{ac|bc}-3C_{ac|ac},\\
			&+3C_{ab|ad}C_{ac|cd}-3C_{ac|ad}C_{ac|cd}+3C_{ac|ac}Cacbd-3C_{ab|ad}C_{ac|bc}+3C_{ac|ad}C_{ac|bc},\\
	    d_{x_2} = &-3C_{ab|ac}C_{bc|cd}+3C_{ab|ad}C_{bc|cd}-C_{ac|ac}C_{ab|cd}+C_{ac|ac}Cacbd\\
			&-3C_{ab|ac}C_{ac|bc}-3C_{ac|ad}C_{ac|bc}+2C_{ac|ac}.\\
           \end{align*}
}
\end{enumerate}	
\end{proposition}

 	
 	
\bibliographystyle{siamplain}
\bibliography{Hybridization}

\begin{thebibliography}{10}

\bibitem{AllmanEtAl2022}
{\sc E.~Allman, H.~Ba\~{n}os, J.~Mitchell, and J.~Rhodes}, {\em The tree of
  blobs of a species network: Identifiability under the coalescent}, J. Math.
  Biol., 86 (2022), p.~10.

\bibitem{Allman2019}
{\sc E.~Allman, H.~Ba\~nos, and J.~Rhodes}, {\em {NANUQ}: a method for
  inferring species networks from gene trees under the coalescent model},
  Algorithms for Molecular Biology, 14 (2019), p.~24,
  \url{https://doi.org/10.1186/s13015-019-0159-2}.

\bibitem{ARSV2015}
{\sc E.~Allman, J.~Rhodes, E.~Stanghellini, and M.~Valtorta}, {\em Parameter
  identifiability of discrete {B}ayesian networks with hidden variables},
  Journal of Causal Inference, 3 (2015), pp.~189--205,
  \url{https://doi.org/doi:10.1515/jci-2014-0021}.

\bibitem{logdetNet}
{\sc E.~S. Allman, H.~Ba\~{n}os, and J.~A. Rhodes}, {\em Identifiability of
  species network topologies from genomic sequences using the log{D}et
  distance}, J. Math. Biol., 84 (2022), pp.~Paper No. 35, 38,
  \url{https://doi.org/10.1007/s00285-022-01734-2}.

\bibitem{AneEtAl2023}
{\sc C.~An{\'e}, J.~Fogg, E.~S. Allman, H.~Ba{\~n}os, and J.~A. Rhodes}, {\em
  Anomalous networks under the multispecies coalescent: theory and prevalence},
  bioRxiv,  (2023), \url{https://doi.org/10.1101/2023.08.18.553582}.

\bibitem{Banos2019}
{\sc H.~Ba\~nos}, {\em Identifying species network features from gene tree
  quartets}, Bulletin of Mathematical Biology, 81 (2019), pp.~494--534.

\bibitem{Singular}
{\sc W.~Decker, G.-M. Greuel, G.~Pfister, and H.~Sch\"onemann}, {\em {\sc
  Singular} {4-3-0} --- {A} computer algebra system for polynomial
  computations}, 2022, \url{http://www.singular.uni-kl.de}.

\bibitem{Degnan2018}
{\sc J.~H. Degnan}, {\em {Modeling Hybridization Under the Network Multispecies
  Coalescent}}, Systematic Biology, 67 (2018), pp.~786--799,
  \url{https://doi.org/10.1093/sysbio/syy040}.

\bibitem{M2}
{\sc D.~R. Grayson and M.~E. Stillman}, {\em Macaulay2, a software system for
  research in algebraic geometry}.
\newblock Available at \url{http://www2.macaulay2.com}.

\bibitem{GrossEtAl2023}
{\sc E.~Gross, R.~Krone, and S.~Martin}, {\em Dimensions of level-1 group-based
  phylogenetic networks}, 2023, \url{https://arxiv.org/abs/2307.15166}.

\bibitem{Gusfield2007}
{\sc D.~Gusfield, V.~Bansal, V.~Bafna, and Y.~S. Song}, {\em {A decomposition
  theory for phylogenetic networks and incompatible characters.}}, Journal of
  Computational Biology, 14 (2007), pp.~1247--1272,
  \url{https://doi.org/10.1089/cmb.2006.0137}.

\bibitem{Huber2015}
{\sc K.~T. Huber, L.~van Iersel, V.~Moulton, C.~Scornavacca, and T.~Wu}, {\em
  Reconstructing phylogenetic level-1 networks from nondense binet and trinet
  sets}, Algorithmica, 77 (2017), pp.~173--200,
  \url{https://doi.org/10.1007/s00453-015-0069-8}.

\bibitem{HusonRuppScorn}
{\sc D.~Huson, R.~Rupp, and C.~Scornavacca}, {\em Phylogenetic Networks},
  Cambridge University Press, Cambridge, 2010.

\bibitem{Meng2009}
{\sc C.~Meng and L.~Kubatko}, {\em Detecting hybrid speciation in the presence
  of incomplete lineage sorting using gene tree incongruence: A model},
  Theoretical Population Biology, 75 (2009), pp.~35--45,
  \url{https://doi.org/10.1016/j.tpb.2008.10.004}.

\bibitem{Rossello2009}
{\sc F.~Rossell{\'{o}} and G.~Valiente}, {\em {All that glisters is not
  galled}}, Mathematical Biosciences, 221 (2009), pp.~54--59,
  \url{https://doi.org/10.1016/j.mbs.2009.06.007},
  \url{https://arxiv.org/abs/arXiv:0904.2448v1}.

\bibitem{Solis-Lemus2016}
{\sc C.~Sol{\'{i}}s-Lemus and C.~An{\'{e}}}, {\em Inferring phylogenetic
  networks with maximum pseudolikelihood under incomplete lineage sorting},
  PLoS Genetics, 12 (2016), \url{https://doi.org/10.1371/journal.pgen.1005896}.

\bibitem{Solis-Lemus2020}
{\sc C.~Solis-Lemus, A.~Coen, and C.~Ane}, {\em On the identifiability of
  phylogenetic networks under a pseudolikelihood model}.
\newblock arXiv, 2020, \url{https://arxiv.org/abs/2010.01758}.

\bibitem{Steel2016}
{\sc M.~Steel}, {\em {Phylogeny: Discrete and Random Processes in Evolution}},
  SIAM, Philadelphia, 2016.

\bibitem{TileySolisLemus2023}
{\sc G.~Tiley and C.~Solis-Lemus}, {\em Extracting diamonds: Identifiability of
  4-node cycles in level-1 phylogenetic networks under a pseudolikelihood
  coalescent model}.
\newblock bioRxiv, 2023,
  \url{https://doi.org/doi.org/10.1101/2023.10.25.564087}.

\bibitem{Nakhleh2015}
{\sc Y.~Yu and L.~Nakhleh}, {\em A maximum pseudo-likelihood approach for
  phylogenetic networks}, BMC Genomics, 16 (2015), p.~S10.

\end{thebibliography}

\end{document}